\documentclass[12pt]{article}

\usepackage{amsmath}
\usepackage{amsfonts}
\usepackage{amsthm}
\usepackage{braket}
\usepackage{mathrsfs}
\usepackage[titletoc,toc,title]{appendix}
\usepackage{url}

\def\lb{\label}
\def\be{\begin{equation}}
\def\ee{\end{equation}}

\def\MR{\mathrm}
\def\MC{\mathcal}

\def\tsum{\textstyle\sum\limits}

\def\wh{\widehat}

\def\bpsi{\bar{\psi}}

\newtheorem{prop}{Proposition}

\begin{document}

\begin{center}
{\Large \bf Remarks towards the spectrum of the Heisenberg spin chain type models}
\end{center}


\begin{center}
\large{\v{C}.~Burd\'{\i}k$^{\dag}$, J.~Fuksa$^{\dag \, *}$,
A.P.~Isaev$^*$, S.O.~Krivonos$^*$ and O.Navr\'atil$^{\ddag}$}
\end{center}

{\small{
\begin{center}
$^*$ Bogoliubov Laboratory of Theoretical Physics,
JINR, \\
141980, Dubna, Moscow region, Russia \\
E-mail: fuksa@theor.jinr.ru; isaevap@theor.jinr.ru; krivonos@theor.jinr.ru

\vskip5mm

$^{\dag}$ Department of Mathematics, Faculty of Nuclear Sciences and Physical Engineering\\
Czech Technical University in Prague\\
E-mail: burdices@kmlinux.fjfi.cvut.cz

\vskip5mm

$^{\ddag}$ Department of Mathematics, Faculty of Transportation Sciences\\
Czech Technical University in
Prague\\
E-mail: navratil@fd.cvut.cz
\end{center}
}}


\begin{abstract}
The integrable close and open chain models can be formulated in
terms of ge\-ne\-ra\-tors of the Hecke algebras. In this review paper, we describe in
detail the Bethe ansatz for the XXX and the XXZ integrable
close chain models. We find the Bethe vectors for two--component and
inhomogeneous models. We also find the Bethe vectors for the
fermionic realization of the integrable XXX and XXZ close chain models by
means of the algebraic and coordinate Bethe ansatz. Special modification
 of the XXZ closed spin chain model ("small polaron model") is consedered. Finally, we
discuss some questions relating to the general open Hecke chain
models.
\end{abstract}

 \vskip 0.5cm

\newpage
\setcounter{page}1
\section{Introduction}
\setcounter{equation}0

A braid group ${\cal B}_{L}$ in the Artin presentation is generated by
invertible elements $T_i$ $(i=1,
\dots, L-1)$ subject to the relations:
\be
\label{braidg}
T_i \, T_{i+1} \, T_i =
T_{i+1} \, T_i \,  T_{i+1} \; , \;\;\;
T_{i} \,  T_{j} = T_{j} \,  T_{i} \;\;
{\rm for} \;\; i \neq j \pm 1 \; .
\ee
An $A$-Type Hecke algebra $H_{L}(q)$ (see, e.g., \cite{ChPr}) is a quotient of the group algebra of ${\cal B}_{L}$ by
additional Hecke relations
\be
\label{ahecke}
T^2_i  = (q-q^{-1}) \,  T_i + 1 \;\; , \;\;\;\;\;\; (i = 1, \dots , L-1) \; ,
\ee
where $q$ is a parameter (deformation parameter).
Let $x$ be a parameter (spectral parameter) and we define elements
\be
\lb{baxtH}
T_k(x) : =  x^{-1/2} \; T_k - x^{1/2} \; T_k^{-1} \in H_{n}(q) \; ,
\ee
which called baxterized elements.
By using (\ref{braidg}) and (\ref{ahecke}) one can check
that the  baxterized elements (\ref{baxtH}) satisfy
the Yang-Baxter equation in the braid group form
\be
\lb{ybeH}
T_k(x) \, T_{k+1}(xy) \, T_k(y) =
T_{k+1}(y) \, T_k(xy) \, T_{k+1}(x)  \; ,
\ee
and
 \be
\lb{bax1}
 T_k(x)T_k(y)=(q-q^{-1}) \, T_k(x y) + (x^{-1/2}-x^{1/2})(y^{-1/2}-y^{1/2}) \; .
 \ee
Equations (\ref{ybeH}) and (\ref{bax1}) are baxterized analogs of the first
 relation in (\ref{braidg})  and Hecke condition (\ref{ahecke}).

The Hamiltonian of the open Hecke chain model of the length $L$ is
\be
 \lb{intrH2}
 {\cal H}_{L} = \sum\limits_{k=1}^{L-1} \, T_k \; \in \; H_{L}(q)  \; ,
  \ee
 (see, e.g., \cite{Isa07} and references therein).
Any representation $\rho$ of the Hecke algebra gives an integrable open spin chain
with the Hamiltonian
 $\rho({\cal H}_{L}) = \sum_{k=1}^{L-1} \, \rho(T_k).$
 Define the closed Hecke algebra $\hat{H}_{L}(q)$ by adding additional generator
 $T_L$ to the set $\{ T_1, \dots , T_{L-1} \}$ such that $T_L$ satisfies the same relations
 (\ref{braidg}) and  (\ref{ahecke}) for any $i$ and $j \neq i \pm 1$,
 where we have to use the periodic condition $T_{L+k} = T_k$.
  Then the closed Hecke chain of the length $L$ is described by the Hamiltonian
  ${\cal H}_{L} = \sum_{k=1}^{L} \, T_k \in \hat{H}_{L}(q)$ and any representation $\rho$
 of $\hat{H}_{L}(q)$ leads to the integrable
 closed spin chain with the Hamiltonian
 \be
 \lb{intrH}
 \rho({\cal H}_{L}) = \sum_{k=1}^{L} \, \rho(T_k) \; .
 \ee

In Sections 2 -- 4, special representations
$\rho =\rho_R$ of the algebra $H_{L}(q)$,
called the $R$-matrix representations, are considered. In the case of $GL_q(2)$-type
$R$-matrix representation $\rho_R$, the Hamiltonian (\ref{intrH}) coincides with the XXZ spin chain
Hamiltonian. It is clear that in the case of $q=1$ we recover the XXX spin chain.
 The integrable structures for XXX spin chain are introduced in Subsection 2.1. We discuss some results
of the algebraic Bethe  ansatz for these models. In Section 3, we formulate
 the so-called two-component model (see \cite{IzKor}, \cite{Slav} and references therein). The two-component model was introduced to avoid problems with computation of correlation functions for local operators attached to some site $x$ of the chain.  Using this approach we
 obtain in Sect. 4 the explicit formulas for the Bethe vectors, which show the equivalence of the
 algebraic and coordinate Bethe ansatzes.

In Section 5 we generalize the results of Sections 2 -- 4 to the
case of inhomogeneous  XXX spin chain.

The realization of the XXX spin chains in terms of free fermions
is considered in Sections 6-8. Here we explicitly construct Bethe
vectors for XXX spin chains in the sectors of one, two and three
magnons. In Section 9, we discuss another special representation
$\rho$ of the Hecke algebra $H_{L}(q)$ which we call the fermionic
representation. In this representation the Hamiltonian
(\ref{intrH}) describes the so-called "small polaron model" (see
\cite{Korep} and references therein). In Sections 10 and 11, we
construct the Bethe vectors and obtain the Bethe ansatz equations for the
"small polaron model" and for the XXZ closed spin chains by means of the
coordinate Bethe ansatz and compare the results with those obtained
by means of the algebraic Bethe ansatz in Sect. 2. We show
 that the  Hamiltonian  of the "small polaron model"
has the different spectrum comparing to the XXZ model in the sector of an even number of magnons.

Finally, in Section 12, we discuss the general open Hecke chain models which are formulated
in terms of the elements of the Hecke algebra $H_n(q)$. We present the characteristic
polynomials (in the case of the finite length of the chain) which define the spectrum of the Hamiltonian
of this model in some special irreducible representations of $H_n(q)$.
The method of construction of irreducible representations of the algebra $H_n(q)$
is formulated at the end of Section 12.

In Appendix, we give some details of our calculations.


\section{Algebraic Bethe Ansatz}
\setcounter{equation}0

At the beginning, we describe some basic features of the agebraic Bethe ansatz. The method was formulated as a part of the quantum inverse scattering method proposed by Faddeev, Sklyanin and Takhtadjan \cite{TakFad,SkTakFad}. The main object of this method is the Yang-Baxter algebra generated by matrix elements of the monodromy matrix. The main rules for the Yang-Baxter algebra were elaborated in the very first papers \cite{KulSk,Kor,Iz}. Many quantum integrable systems were described in terms of this method, cf. \cite{Skl79,SklFad,IzKor82}. We strongly recommend the review paper \cite{Fad} for introductory reading and \cite{KorBogIz} for more detailed review.

\subsection{L-operator and transfer matrix for XXX spin chain}

Suppose we have a chain of $L$ sites. The local Hilbert space $h_j$ corresponds to the $j$-th site. For our purposes, it is sufficient to suppose  $h_j = \mathbb{C}^2$. The total Hilbert space of the chain is
\be
\lb{hilbe}
    \mathscr{H} = \prod_{j=1}^L \otimes h_j.
\ee

The basic tool of algebraic Bethe ansatz is the Lax operator. For its definition, we need an auxiliary vector space $V_a=\mathbb{C}^2$. The Lax operator is a parameter depending object acting on the tensor product $V_a\otimes h_i$
\begin{equation}
L_{a,i}:\quad V_a\otimes h_i \rightarrow V_a\otimes h_i
\end{equation}
explicitly defined as
\begin{equation} \label{Lax}
    L_{a,i}(\lambda) =(\lambda+\frac{1}{2})\mathbb{I}_{a,i} + \sum_{\alpha=1}^3 \sigma_a^{\alpha} S^\alpha_i
\end{equation}
where $S^\alpha_i = \frac{1}{2}\sigma^\alpha_i$ is the spin operator on the $i$-th site,
 $\sigma_a^{\alpha} =(\sigma_a^{x},\sigma_a^{y},\sigma_a^{z})$ -- are Pauli sigma-matrices
 which act in the space $V_a$ ($\sigma_i^{\alpha}$ -- are Pauli sigma-matrices
 which act in the space $h_i$)
 and $\mathbb{I}_{a,i}$ is the identity matrix in $\quad V_a\otimes h_i$.
 Operator $L_{a,i}(\lambda)$ can be expressed as a matrix in the auxiliary space
\begin{equation}\label{Lax1}
    L_{a,i}(\lambda) = \left(
    \begin{array}{cc}
    \lambda+\frac{1}{2} + S^z_i & S^-_i \\
    S^+_i & \lambda+\frac{1}{2}-S^z_i
    \end{array} \right).
\end{equation}
Its matrix elements form an associative algebra of local operators in the quantum space $h_i$.

Introducing the permutation operator $P$
\begin{equation}\label{perm}
     P= \frac{1}{2} \left( \mathbb{I}\otimes \mathbb{I} + \sum_{\alpha=1}^3 \sigma^\alpha \otimes \sigma^\alpha \right)
\end{equation}
(here $\mathbb{I}$ denotes a $2 \times 2$ unit matrix) we can rewrite the Lax operator as
\begin{equation}\label{Lax2}
    L_{a,i}(\lambda) = \lambda\mathbb{I}_{a,i} + P_{a,i}.
\end{equation}

Assume two Lax operators $L_{a,i}(\lambda)$ resp. $L_{b,i}(\mu)$ in the same quantum space $h_i$ but in different auxiliary spaces $V_a$ resp. $V_b$. The product of $L_{a,i}(\lambda)$ and $L_{b,i}(\mu)$ makes sense in the tensor product $V_a\otimes V_b \otimes h_i$. It turns out that there is an operator $R_{ab}(\lambda-\mu)$ acting nontrivially in $V_a\otimes V_b$ such that the following equality holds:
\begin{equation} \label{FCR}
    R_{ab}(\lambda-\mu)L_{a,i}(\lambda)L_{b,i}(\mu) = L_{b,i}(\mu) L_{a,i}(\lambda) R_{ab}(\lambda-\mu).
\end{equation}
Relation \eqref{FCR} is called the fundamental commutation relation. The explicit expression for $R_{ab}(\lambda-\mu)$ is
\begin{equation} \label{Rmatrix}
    R_{ab}(\lambda-\mu) = (\lambda-\mu) \mathbb{I}_{a,b} + P_{a,b}
\end{equation}
where $\mathbb{I}_{a,b}$ resp. $P_{a,b}$ is identity resp. permutation operator in $V_a\otimes V_b$.  In the matrix form we get for $R_{ab}(\lambda-\mu)$
\be\lb{Rmatrix1}
R_{ab}(\lambda-\mu) = \left(
    \begin{array}{cccc}
        \lambda-\mu+1 & 0 & 0 & 0 \\
        0 & \lambda-\mu & 1 & 0 \\
        0 &  1 & \lambda-\mu & 0 \\
        0 & 0 & 0 & \lambda-\mu+1 \\
    \end{array} \right).
\ee

The operator $R_{ab}(\lambda-\mu)$ is called the R-matrix. It satisfies the Yang-Baxter equation
\be \lb{YBE}
    R_{ab}(\lambda-\mu)R_{ac}(\lambda)R_{bc}(\mu) = R_{bc}(\mu)R_{ac}(\lambda)R_{ab}(\lambda-\mu)
\ee
in $V_{a}\otimes V_b \otimes V_c$.

Comparing \eqref{Lax2} and \eqref{Rmatrix} we see that the Lax operator and the R-matrix are the same.

We define a monodromy matrix
\be \label{monod}
    T_a(\lambda)= L_{a,1}(\lambda)L_{a,2}(\lambda)\dots L_{a,L}(\lambda)
\ee
as a product of the Lax operators along the chain, i.e. over all quantum spaces $h_i$. As a matrix in the auxiliary space $V_a$, the monodromy matrix
\begin{equation}\label{monod1}
    T_a(\lambda) =
    \left( \begin{array}{cc}
        A(\lambda) & B(\lambda) \\
        C(\lambda) & D(\lambda)
    \end{array} \right)
\end{equation}
defines an algebra of global operators $A(\lambda),B(\lambda),C(\lambda),D(\lambda)$ on the Hilbert space $\mathscr{H}$. It is called the Yang-Baxter algebra. The monodromy matrix $T_a(\lambda)$ is a step from local observables $S^\alpha_i$ on $h_i$ to global observables on $\mathscr{H}$.

The trace
\begin{equation} \label{TransferMatrix}
\tau(\lambda) \equiv\textsl{Tr}_a T_{a}(\lambda) = A(\lambda)+D(\lambda)
\end{equation}
of $T_a(\lambda)$ in the auxiliary space $V_a$ is called the transfer matrix. It constitutes a generating function for commutative conserved charges. Assume the Lax operators $L_a=L_{a,i}(\lambda),L_b=L_{b,i}(\mu),L'_a=L_{a,i+1}(\lambda),L'_b=L_{b,i+1}(\mu)$ and $R_{ab}=R_{ab}(\lambda-\mu)$ in the tensor product $V_a\otimes V_b \otimes \mathscr{H}$, then
\begin{equation}
R_{ab}L_a L'_a L_b L'_b = R_{ab}L_a L_b L'_a L'_b = L_b L_a R_{ab} L'_a  L'_b = L_b L_a L'_b L'_a R_{ab}= L_b L'_b L_a L'_a R_{ab}.
\end{equation}
Here we used \eqref{YBE} and the fact that operators acting nontrivially in different vector spaces commute. Hence, we can deduce commutation relations between the elements of the monodromy matrix
\be \label{FCRglobal}
R_{ab}(\lambda-\mu) T_a(\lambda) T_b(\mu) = T_b(\mu) T_a(\lambda) R_{ab}(\lambda-\mu).
\ee
Equation \eqref{FCRglobal} is a consequence of \eqref{FCR} for global observables $A(\lambda),B(\lambda),C(\lambda)$ and $D(\lambda)$. We call it the global fundamental commutation relation.  The commutativity of transfer matrices obviously follows from \eqref{FCRglobal}. After multiplying \eqref{FCRglobal} by $R_{ab}^{-1}(\lambda-\mu)$ we get
\begin{equation}
    T_a(\lambda) T_b(\mu) = R_{ab}^{-1}(\lambda-\mu) T_b(\mu) T_a(\lambda) R_{ab}(\lambda-\mu).
\end{equation}
Taking the trace over auxiliary spaces $V_a$ and $V_b$ we obtain
\be \label{TransferCommute}
    \tau(\lambda) \tau(\mu)
    = \tau(\mu) \tau(\lambda).
\ee

Obviously, the monodromy matrix \eqref{monod} is a polynomial of degree $L$ with respect to the parameter $\lambda$
\be
 \lb{Texp}
    T_{a}(\lambda) = \left(\lambda + \frac{1}{2} \right)^L \mathbb{I} + \left(\lambda + \frac{1}{2} \right)^{L-1} \sum_{i=1}^L \sum_{\alpha=1}^3 \sigma_a^{\alpha}\otimes S^\alpha_i + \dots
\ee
Therefore, the transfer matrix $\tau(\lambda)$ is also a polynomial of degree $L$
\be
 \lb{Texp1}
    \tau(\lambda) = 2 \left(\lambda + \frac{1}{2} \right)^L + \sum_{k=0}^{L-2} \lambda^k Q_k.
\ee
The term of order $\lambda^{L-1}$ vanishes because Pauli matrices are traceless. Due to commutativity \eqref{TransferCommute} of transfer matrices also
\be
\lb{Texp2}
    [Q_j,Q_k] = 0.
\ee
We see that the transfer matrix is a generating function for a set of commuting observables.

The Hamiltonian of the system appears naturally amongst the observables $Q_k$. From the definition of the Lax operator \eqref{Lax} we see that
\be
    L_{a,i}(-1/2) = P_{a,i}.
\ee
Hence,
\be
    T_a(-1/2) = P_{a,1} P_{a,2} \dots P_{a,L} = P_{L-1,L}\dots P_{2,3}P_{1,2}P_{a,1}.
\ee
If we differentiate $T_a(\lambda)$ with respect to $\lambda$, we get
\be
    \frac{dT_a(\lambda)}{d\lambda} \Big|_{\lambda=-1/2} = \sum_{k=1} ^L P_{a,1}\dots \underbrace{P_{a,k}}_{\text{missing}} \dots P_{a,L} = \sum_{k=1}^L P_{L-1,L}\dots P_{k-1,k+1}\dots P_{1,2} P_{a,1}.
\ee
Remind that $\textsl{Tr}_a P_{a,j} = \mathbb{I}_j$. Differentiating the logarithm of the transfer matrix
\begin{align}
 & \frac{d}{d\lambda} \ln \tau(\lambda) \Big|_{\lambda=-1/2} = \frac{d\tau(\lambda)}{d\lambda} \tau^{-1}(\lambda) \Big|_{\lambda=-1/2} = \\
 & = \sum_{k=1}^L \left( P_{L-1,L}\dots P_{k-1,k+1}\dots P_{1,2} \right) \left( P_{1,2}P_{2,3}\dots P_{L-1,L} \right) = \sum_{k=1}^L P_{k,k+1}.
\end{align}
The Hamiltonian of the system is
\be
    H= \sum_{k=1}^L \sum_{\alpha=1}^3 S_k^\alpha S_{k+1}^\alpha = \frac{1}{2} \sum_{k=1}^{L} P_{k,k+1} - \frac{L}{4}
\ee
where we set $S_{n+L}=S_n$ resp. $P_{L,L+1}=P_{L,1}$.
We can see that
\be
    H = \frac{1}{2} \frac{d}{d\lambda} \ln \tau(\lambda) \Big|_{\lambda=-1/2} - \frac{L}{4}.
\ee
This is the reason why we can say that the transfer matrix $\tau(\lambda)$ is a generating function for commuting conserved charges.

\vspace{0.3cm}

\noindent
{\bf Remark.} Let $S_i^\alpha$ be generators of the Lie algebra $su(2)$ in $i$-th site
 \be
 \lb{defRel}
 [ S_i^\alpha \, , \; S_j^\beta] = i \varepsilon^{\alpha \beta \gamma} S_i^\gamma \; \delta_{ij} \; ,
 \ee
 and we take generators $S_i^\alpha$ in any representation of $su(2)$ which acts in the space $h_i$.
 Then equations (\ref{Lax}) and (\ref{Lax1}) define $L$-operator for the integrable chain model with arbitrary
 spin in each site. Relations (\ref{FCR}) with $R$-matrix (\ref{Rmatrix}) are equivalent to the defining relations
 (\ref{defRel}). Formulas (\ref{monod}), (\ref{monod1}), (\ref{TransferMatrix}), (\ref{TransferCommute}),
 (\ref{Texp}), (\ref{Texp1}) and (\ref{Texp2})
 are valid for this generalized spin chain models as well.

\subsection{Some remarks on the XXZ chain}

The fundamental $R$-matrix for the quantum group $GL_q(N)$ is \cite{Jimb1,FRT}
\be
\lb{Rmatr}
\hat{R}
= q \, \sum\limits_{i=1}^N e_{ii} \otimes e_{ii} + \sum\limits_{i\neq j} e_{ij} \otimes e_{ji} +
(q-q^{-1}) \, \sum\limits_{i < j} e_{ii} \otimes e_{jj} \; ,
\ee
where $e_{ij}$ is the $N \times N$ matrix unity.
 In a particular case of
 $GL_q(2)$, the $R$-matrix (\ref{Rmatr}) can be written in terms of Pauli matrices
 \begin{align}
 \hat{R}  &=\frac{1}{2} \left(  \sigma^x \otimes \sigma^x +  \sigma^y \otimes \sigma^y  +
 \frac{q+q^{-1}}{2} \sigma^z \otimes \sigma^z  \right)  + \notag\\
 &\quad + \frac{q-q^{-1}}{4} \left( \sigma^z  \otimes \mathbb{I}_2 - \mathbb{I}_2 \otimes \sigma^z  \right)
 + \frac{3q-q^{-1}}{4} \mathbb{I}_2 \otimes \mathbb{I}_2 .\lb{Rxxz}
 \end{align}
 Here and below we use notation $\mathbb{I}_N$ for the $N \times N$ unit matrix.
The fundamental R-matrix \eqref{Rmatr} satisfies the Hecke condition \eqref{ahecke}
\be
    \hat{R}^2=(q-q^{-1})\hat{R}+\mathbb{I}_N \otimes \mathbb{I}_N \; .
\ee
If we define
\be \lb{Rxxz1}
 \hat{R}^{(q)}_{k k+1}  =
  \mathbb{I}_N^{\otimes (k-1)} \otimes \hat{R} \otimes \mathbb{I}_N^{\otimes (L - k-1)}
\ee
we obtain the $R$-matrix representation $\rho_R$ of the Hecke algebra
 (\ref{braidg}), (\ref{ahecke})
 \be
 \lb{rmatrR}
 \rho_R \; : \;\;\; T_k \;\; \to \;\; \hat{R}^{(q)}_{k k+1}  \; .
 \ee
Then, the baxterized R-matrix is (see eq. (\ref{baxtH}))
\be
\lb{Rbaxt}
\begin{array}{c}
\hat{R}_{k k+1}(\mu) = \rho_R \left(\mu^{-1/2} T_k - \mu^{1/2} T_k^{-1}\right)
= \mu^{-1/2} \hat{R}^{(q)}_{k k+1} - \mu^{1/2} (\hat{R}^{(q)}_{k k+1})^{-1} = \\ [0.2cm]
= (\mu^{-1/2} -\mu^{1/2}) \hat{R}^{(q)}_{k k+1} + \mu^{1/2} (q-q^{-1}) \; .
\end{array}
\ee
This R-matrix is a solution of the Yang-Baxter equation in the braid group form
 \be
\lb{YBEbaxt}
 \hat{R}_{k k+1}(\lambda) \hat{R}_{k+1 k+2}(\lambda \cdot \mu) \hat{R}_{k k+1}(\mu) =
 \hat{R}_{k+1 k+2}(\mu) \hat{R}_{k k+1}(\lambda \cdot \mu) \hat{R}_{k+1 k+2}(\lambda).
 \ee

Note that if there is a solution of equation \eqref{YBEbaxt}, the solution of the equation
\be
    R_{k,k+1}(\lambda)R_{k,k+2}(\lambda\mu)R_{k+1,k+2}(\mu) = R_{k+1,k+2}(\mu)R_{k,k+2}(\lambda\mu)R_{k,k+1}(\lambda) \lb{YBEbaxt1}
\ee
can be easily found as
\be\lb{Rbaxt0}
R_{k,k+1}(\lambda)=\hat{R}_{k,k+1}(\lambda) P_{k,k+1}.
\ee
The R-matrix $R_{k,k+2}(\lambda)$ has to be defined as $P_{k+1,k+2} R_{k,k+1}(\lambda) P_{k+1,k+2}$. The validity of \eqref{YBEbaxt1} is very important for correct definition of the transfer matrix. We are able to define the Lax operator as the R-matrix
\be
    L_{a,i}(\lambda) = R_{a,i}(\lambda)
\ee
and the monodromy matrix in the form \eqref{monod}. Commutativity of the transfer matrix is just a matter of proving
\be
\lb{RRII}
    R_{ab}(\mu) T_a(\lambda\mu) T_b(\lambda) = T_b(\lambda) T_a(\lambda\mu) R_{ab}(\mu).
\ee

The R-matrices  (\ref{Rmatr}), (\ref{Rbaxt}) for $N=2$ are the basic building blocks for the XXZ spin chain.
Let us write \eqref{Rbaxt0} for $N=2$ as following
\be \lb{RxxzR}
 R_{k k+1}(\lambda)  =
  \mathbb{I}_2^{\otimes (k-1)} \otimes R(\lambda) \otimes \mathbb{I}_2^{\otimes (L - k-1)}  \; ,
\ee
where $R(\lambda) = (\lambda^{-1/2} \hat{R} - \lambda^{1/2} \hat{R}^{-1}) \cdot P$,
 the matrix $\hat{R}$ is given in \eqref{Rxxz} and $R(\lambda)$ has the matrix form
\begin{align}
\lb{RRRI}
    R(\lambda) =
    \begin{pmatrix}
    \lambda^{-1/2}q-\lambda^{1/2} q^{-1} & 0 & 0 & 0 \\
    0 & \lambda^{-1/2}-\lambda^{1/2} & \lambda^{-1/2}(q-q^{-1}) & 0 \\
    0 & \lambda^{1/2}(q-q^{-1}) & \lambda^{-1/2}-\lambda^{1/2} & 0\\
    0&0&0& \lambda^{-1/2}q-\lambda^{1/2} q^{-1}
    \end{pmatrix}
\end{align}
which is important to write the commutation relations \eqref{RRII} in components. We see that the form \eqref{RRRI}
of  the R-matrix is not symmetric to transposition, as usually appears in literature, cf. \cite{Fad}, \cite{Slav} etc. We use the Drinfel'd--Reshetikhin twist to symmetrize it, cf. \cite{Resh}.

The R-matrix $R_{ab}(\lambda)$ acts in the tensor product of the auxiliary spaces $V_a \otimes V_b$. The monodromy matrix $T_a(\lambda)$ acts in $V_a\otimes \mathscr{H}$. Let $U$ be a diagonal matrix. It can be easily seen that $[U\otimes \mathbb{I}_b +\mathbb{I}_a\otimes U, R_{ab}(\lambda)]=0$. We introduce the twisted R-matrix resp. the monodromy matrix
\begin{align}
    \tilde{R}_{ab}(\lambda) = (\lambda^U \otimes \mathbb{I}_b) R_{ab}(\lambda) (\lambda^{-U} \otimes \mathbb{I}_b),
    \lb{monodXXZ1a} \\
     \tilde{T}_a(\lambda) = (\lambda^U\otimes \mathbb{I}_{\mathscr{H}}) T_a(\lambda) (\lambda^{-U}\otimes \mathbb{I}_{\mathscr{H}}). \lb{monodXXZ1}
\end{align}
If
\be\lb{FCRglobalXXZ}
    R_{ab}(\lambda) T_a(\lambda\mu) T_b(\mu) =  T_b(\mu)T_a(\lambda\mu) R_{ab}(\lambda)
\ee
is satisfied, then also
\begin{align}
    \tilde{R}_{ab}(\lambda) \tilde{T}_a(\lambda\mu) \tilde{T}_b(\mu)  =\tilde{T}_b(\mu) \tilde{T}_a(\lambda\mu) \tilde{R}_{ab}(\lambda).\lb{FCRglobalXXZ1}
\end{align}
In other words, the global fundamental commutation relations remain unchanged.

This twist differs slightly from the twist proposed in \cite{Slav}. The author supposes a matrix $\omega$ whose tensor square commutes with R-matrix $[\omega\otimes\omega,R_{ab}(\lambda)]=0$ and concludes that the matrix $\hat{T}_a(\lambda)=\omega T_a(\lambda)$ satisfies \eqref{FCRglobalXXZ} as the original untwisted matrix $T_a(\lambda)$. In both cases, the crucial premise for usability of the twist is the commutativity of $\omega\otimes \omega$, or
 its infinitesimal form $U\otimes \mathbb{I} +\mathbb{I}\otimes U$, with the R-matrix.

Below we will consider only the case of $N=2$.
Taking $U=\left(\begin{smallmatrix}1/4 & 0 \\ 0 & -1/4\end{smallmatrix}\right)$
in \eqref{monodXXZ1a}, where $R_{ab}(\lambda)$ is given by \eqref{RRRI}, we get
\be\lb{RmatXXZ}
    \tilde{R}(\lambda) =
\begin{pmatrix}
\lambda^{-1/2}q-\lambda^{1/2} q^{-1} & 0 & 0 & 0 \\
0 & \lambda^{-1/2}-\lambda^{1/2} & q-q^{-1} & 0 \\
0 &q-q^{-1} & \lambda^{-1/2}-\lambda^{1/2} & 0\\
0&0&0& \lambda^{-1/2}q-\lambda^{1/2} q^{-1}
\end{pmatrix}
\ee
which corresponds to the R-matrix appearing in \cite{Fad}, \cite{Slav}. Moreover, it is easy to see that
\be\lb{monodXXZ2}
    \tilde{T}_a(\lambda)= \tilde{R}_{a,1}(\lambda) \tilde{R}_{a,2}(\lambda)\dotsb \tilde{R}_{a,L}(\lambda).
\ee
It can also be seen that
\be
    \tilde{T}_a(\lambda)=
    \begin{pmatrix}
        \tilde{A}(\lambda) & \tilde{B}(\lambda) \\
        \tilde{C}(\lambda) & \tilde{D}(\lambda) \\
    \end{pmatrix} =
    \begin{pmatrix}
        A(\lambda) & \lambda^{1/2} B(\lambda) \\
        \lambda^{-1/2} C(\lambda) & D(\lambda) \\
    \end{pmatrix}
\ee
where $A,\ B,\ C,\ D$ correspond to the original monodromy matrix $T_a(\lambda)$. Moreover, one can easily realize that
\begin{align}
\tilde{A}(\lambda)&=\frac{1}{2} \left(\lambda^{-\frac{1}{2}}-\lambda^{\frac{1}{2}}+q\lambda^{-\frac{1}{2}}-q^{-1}\lambda^{\frac{1}{2}}\right) \mathbb{I}_2 + \frac{1}{2}\left(\lambda^{\frac{1}{2}}-\lambda^{-\frac{1}{2}}+q\lambda^{-\frac{1}{2}}-q^{-1}\lambda^{\frac{1}{2}}\right)\sigma^z,\\
\tilde{B}(\lambda)&=(q-q^{-1})\sigma^-,\\
\tilde{C}(\lambda)&=(q-q^{-1})\sigma^+,\\
\tilde{D}(\lambda)&=\frac{1}{2} \left(\lambda^{-\frac{1}{2}}-\lambda^{\frac{1}{2}}+q\lambda^{-\frac{1}{2}}-q^{-1}\lambda^{\frac{1}{2}}\right)\mathbb{I}_2 + \frac{1}{2} \left(\lambda^{-\frac{1}{2}}-\lambda^{\frac{1}{2}}-q\lambda^{-\frac{1}{2}}+q^{-1}\lambda^{\frac{1}{2}}\right)\sigma^z \; ,
\end{align}
 where $\sigma^\pm = \frac{1}{2}(\sigma^x \pm i \sigma^y)$.
The twisted  R-matrix $\tilde{R}_{k,k+1}(\lambda)$ resp. the
twisted monodromy matrix $\tilde{T}_a(\lambda)$ will be used throughout the text.

\subsection{Global fundamental commutation relations} \lb{subsec:FCR}

Global commutation relations are determined by equation \eqref{FCRglobal} resp. \eqref{FCRglobalXXZ} for XXX resp. XXZ in the tensor product $V_a\otimes V_b\otimes\mathscr{H}$. They are explicitly expressed by multiplication of matrices in the tensor product of the auxiliary spaces $V_a\otimes V_b$. After simple factorization, the R-matrices \eqref{Rmatrix1} resp. \eqref{RmatXXZ} can be written uniformly in the following way:
\be
R_{ab}(\lambda) = \left(
    \begin{array}{cccc}
        f(\lambda) & 0 & 0 & 0 \\
        0 & 1 & g(\lambda) & 0 \\
        0 &  g(\lambda) &1 & 0 \\
        0 & 0 & 0 & f(\lambda) \\
    \end{array} \right)
\ee
where for the XXX chain we have
\begin{align}
    & f(\lambda)=\frac{\lambda+1}{\lambda}, & g(\lambda)& =\frac{1}{\lambda},\\
    \intertext{and for XXZ}
    & f(\lambda)=\frac{\lambda^{-1/2}q -\lambda^{1/2} q^{-1}}{\lambda^{-1/2} -\lambda^{1/2}}, & g(\lambda) &= \frac{q - q^{-1}}{\lambda^{-1/2} -\lambda^{1/2}}.
\end{align}

We take the monodromy matrix \eqref{monodXXZ1} resp. \eqref{monodXXZ2} for the XXZ chain. For more comfort, we omit the tilde over the corresponding operators. The matrices $T_a(\lambda)$ resp. $T_b(\mu)$ take the form
\be
    T_a(\lambda) = \left(
    \begin{array}{cccc}
    A(\lambda)& & B(\lambda) & \\
    &   A(\lambda)& & B(\lambda) \\
    C(\lambda)& & D(\lambda) & \\
    &   C(\lambda)& & D(\lambda) \\
    \end{array} \right)
\ee
resp.
\be
    T_b(\mu) = \left(
    \begin{array}{cccc}
    A(\mu)& B(\mu) & & \\
    C(\mu)& D(\mu) & & \\
    & & A(\mu)& B(\mu) \\
    & & C(\mu)& D(\mu) \\
    \end{array} \right).
\ee
Multiplying and comparing the left- and right-hand side of \eqref{FCRglobal} resp. \eqref{FCRglobalXXZ1}, we obtain the set of commutation relations. Comparing the matrix elements on the positions $(1,1),\ (1,4)$, $(4,1)$, $(4,4)$ we obtain
\be \label{com}
    [A(\lambda),A(\mu)] = [B(\lambda),B(\mu)] = [C(\lambda),C(\mu)] = [D(\lambda),D(\mu)] = 0.
\ee
Comparing $(2,3)$ resp. $(2,2)$
\begin{align}
    [B(\lambda),C(\mu)] = g(\lambda,\mu) \left( D(\mu)A(\lambda) - D(\lambda)A(\mu) \right), \\
    [A(\lambda),D(\mu)] = g(\lambda,\mu) \left( C(\mu)B(\lambda) - C(\lambda)B(\mu) \right).
\end{align}
From a comparison of the matrix elements $(1,3),\ (3,4),\ (2,1)$ resp. $(4,3)$ we obtain
\begin{align}
A(\mu)B(\lambda) &= f(\lambda,\mu) B(\lambda)A(\mu) + g(\mu,\lambda) B(\mu)A(\lambda), \label{com1}\\
D(\mu)B(\lambda) &= f(\mu,\lambda)B(\lambda)D(\mu) + g(\lambda,\mu) B(\mu)D(\lambda), \label{com2}\\
A(\mu)C(\lambda) &= f(\mu,\lambda)C(\lambda)A(\mu) + g(\lambda,\mu) C(\mu)A(\lambda), \label{com3}\\
D(\mu)C(\lambda) &= f(\lambda,\mu)C(\lambda)D(\mu) + g(\mu,\lambda) C(\mu)D(\lambda), \label{com4}
\end{align}
where for the XXX chain we have
\be \label{com5}
f(\lambda,\mu) =f(\lambda-\mu)= \frac{\lambda-\mu+1}{\lambda-\mu},\qquad g(\lambda,\mu) =g(\lambda-\mu)= \frac{1}{\lambda-\mu},
\ee
and for XXZ
\begin{align}
    & f(\lambda,\mu)=f(\lambda/\mu)= \frac{\mu q-\lambda q^{-1}}{\mu -\lambda}, & g(\lambda,\mu) = g(\lambda/\mu) =\sqrt{\lambda\mu}\; \frac{q-q^{-1}}{\mu-\lambda}. \lb{com6}
\end{align}
We see that $g(\mu,\lambda)=-g(\lambda,\mu)$.

\subsection{Eigenstates of the transfer matrix}

To uncover the spectrum of the transfer matrix $\tau(\lambda)=A(\lambda)+D(\lambda)$ is now the natural next step. In \ref{subsec:FCR}, we get four operators $A(\lambda),B(\lambda),C(\lambda)$ and $D(\lambda)$ under commutation relations \eqref{com}-\eqref{com5}. They generate an associative algebra. Relations \eqref{com}-\eqref{com5} together with an assumption that the Hilbert space $\mathscr{H}$ has the structure of the Fock space are sufficient to find the spectrum $\tau(\lambda)$.  From the beginning, we work on the Hilbert space $\mathscr{H} = (\mathbb{C}^2)^{\otimes L}$, i.e. we choose a specific representation. But the content of this chapter is valid in general, i.e. also for other representations.

To uncover the Fock space structure in $\mathscr{H}$, let us find a pseudovacuum vector $\ket{0}\in \mathscr{H}$ such that $C(\lambda)\ket{0} = 0$ which is an eigenvector of the operators $A(\lambda)$ and $D(\lambda)$
\be \label{pseudovac}
    A(\lambda) \ket{0} = \alpha(\lambda) \ket{0},\qquad D(\lambda) \ket{0} = \delta(\lambda) \ket{0}.
\ee
Let us remind that there is a state $\ket{0}_k$ in each $h_k$ such that the corresponding Lax operator is of the upper triangular form
\be
    L_{a,k}(\lambda)\ket{0}_k = \left( \begin{array}{cc}
    a(\lambda) & \text{something} \\
    0 & d(\lambda) \\
    \end{array} \right)\ket{0}_k,
\ee
where $a(\lambda)$, $d(\lambda)$ are the functions of the parameter $\lambda$. We see that for XXX
\begin{align}
&a(\lambda)=\lambda+1, & d(\lambda)&=\lambda \lb{adXXX}\\
\intertext{and for XXZ}
&a(\lambda)=\lambda^{-1/2} q-\lambda^{1/2}q^{-1}, & d(\lambda)&=\lambda^{-1/2} -\lambda^{1/2}.\lb{adXXZ}
\end{align}
The vector $\ket{0}\in \mathscr{H}$ is of the form
\be
    \ket{0} = \ket{0}_1\otimes \ket{0}_2\otimes \dots \otimes \ket{0}_L. \lb{vacuum}
\ee
In our particular representation, $h_k=\mathbb{C}^2$, we have $\ket{0}_k=\left(\begin{smallmatrix} 1 \\ 0\\ \end{smallmatrix}\right)$. It can be easily seen that
\be
    T_a(\lambda) \ket{0}= \left( \begin{array}{cc}
    a(\lambda)^L & \text{something} \\
    0 & d(\lambda)^L \\
    \end{array}\right)\ket{0}.
\ee
We have found that the state $\ket{0}\in \mathscr{H}$ satisfies
\be\label{pseudovac1}
    C(\lambda)\ket{0}=0,\qquad A(\lambda)\ket{0} = \alpha(\lambda) \ket{0},\qquad D(\lambda)\ket{0} = \delta(\lambda) \ket{0}
\ee
where
\be
    \alpha(\lambda) = a(\lambda)^L,\qquad \delta(\lambda)= d(\lambda)^L,
\ee
i.e., $\ket{0}\in\mathscr{H}$ is an eigenstate of the transfer matrix $\tau(\lambda)=A(\lambda)+D(\lambda)$.

Other eigenstates of the transfer matrix \eqref{TransferMatrix} are of the form
\be\label{BetheVect}
    \ket{\{\lambda\}} =\ket{\lambda_1,\dots,\lambda_M} \equiv
    B(\lambda_1)B(\lambda_2)\dots B(\lambda_M) \ket{0} \equiv | M \rangle  \; .
\ee
They are called the Bethe vectors. For $M\in\mathbb{N}$ we will call the Bethe vector
 $\ket{\lambda_1,\dots,\lambda_M}$ the $M$-magnon state. It turns out that there have to be some restrictions on the parameters $\{\lambda\}=\{\lambda_1,\dots,\lambda_M\}$ to get the eigenstates of the transfer matrix. First, we note that in view of commutativity of the operators $B$ (\ref{com}) we have
 \be
 \label{bae02}
 |\lambda_1,\dots,\lambda_M \rangle = |\sigma( \lambda_1,\dots,\lambda_M) \rangle  \; ,
 \ee
 for any permutation $\sigma\in S_M$ of $\{\lambda_1,\dots, \lambda_M \}$.
 Then, using (\ref{com1}), (\ref{pseudovac1}) and (\ref{bae02}), we deduce
  \be
 \label{bae03}
  \begin{array}{l}
 A(\lambda) \, |\lambda_1,\dots,\lambda_M \rangle =
 A(\lambda) \, B(\lambda_1)\cdots B(\lambda_M)\ket{0} = \\ [0.3cm]
= \Bigl( f(\lambda_1,\lambda)  B(\lambda_1) A(\lambda) + g(\lambda,\lambda_1)  B(\lambda) A(\lambda_1) \Bigr)
   B(\lambda_2)\cdots B(\lambda_M)\ket{0} =  \\ [0.3cm]
   = \Bigl( f(\lambda_1,\lambda)   + g(\lambda,\lambda_1)  P_{\lambda\lambda_1} \Bigr)
   B(\lambda_1) A(\lambda) B(\lambda_2)\cdots B(\lambda_M)\ket{0} =  \\ [0.3cm]
   = \dots = \prod\limits_{k=1}^M \Bigl( f(\lambda_k,\lambda)   + g(\lambda,\lambda_k)  P_{\lambda\lambda_k} \Bigr)
   \alpha(\lambda) |\lambda_1,\dots,\lambda_M \rangle  = \\ [0.3cm]
   = \alpha(\lambda) \prod\limits_{k=1}^M f(\lambda_k,\lambda) |\lambda_1,\dots,\lambda_M \rangle +
   \sum\limits_{i=1}^M \Phi_i(\lambda,\lambda_1,...,\lambda_M)
   |\lambda_1,..., \lambda_{i-1},\lambda,\lambda_{i+1} ,...\rangle \; ,
 \end{array}
 \ee
 where $P_{\lambda\lambda_k}$ is a permutation operator of the parameters $\lambda$ and $\lambda_k$
 and it is clear that
 \be \label{bae04}
 \Phi_1(\lambda,\lambda_1,...,\lambda_M) = \alpha(\lambda_1) g(\lambda,\lambda_1) \prod\limits_{k=2}^M
 f(\lambda_k,\lambda_1) \; .
 \ee
 Since the left-hand side of (\ref{bae03})
 is symmetric under all permutations of $\{\lambda_1,\dots, \lambda_M \}$,
 we obtain
 \be \label{bae05}
 \Phi_i(\lambda,\lambda_1,...,\lambda_M) = \alpha(\lambda_i) g(\lambda,\lambda_i)
 \prod\limits_{\stackrel{k=1}{_{k \neq i}}}^M
 f( \lambda_k,\lambda_i) \; , \;\;\;\; \forall i =1,\dots,M \; .
 \ee
 In the same way by using (\ref{com2}), (\ref{pseudovac1}) and (\ref{bae02}) we deduce
  \be
 \label{bae06}
 \begin{array}{l}
 D(\lambda) \, |\lambda_1,\dots,\lambda_M \rangle =
 D(\lambda) \, B(\lambda_1)\cdots B(\lambda_M)\ket{0} = \\ [0.3cm]
 = \prod\limits_{k=1}^M \Bigl( f(\lambda,\lambda_k)   + g(\lambda_k,\lambda)  P_{\lambda\lambda_k} \Bigr)
   \delta(\lambda) |\lambda_1,\dots,\lambda_M \rangle  = \\ [0.3cm]
   = \delta(\lambda) \prod\limits_{k=1}^M f(\lambda,\lambda_k) |\lambda_1,\dots,\lambda_M \rangle +
   \sum\limits_{i=1}^M \Psi_i(\lambda,\lambda_1,...,\lambda_M)
   |\lambda_1, ..., \lambda_{i-1},\lambda,\lambda_i ,...,\lambda_M \rangle \; ,
 \end{array}
 \ee
 where
 \be
 \label{bae07}
 \Psi_i(\lambda,\lambda_1,...,\lambda_M) = \delta(\lambda_i) g(\lambda_i,\lambda)
 \prod\limits_{\stackrel{k=1}{_{k \neq i}}}^M
 f(\lambda_i,\lambda_k) \; , \;\;\;\; \forall i =1,\dots,M \; .
 \ee
 The combination of (\ref{bae03})  and (\ref{bae06}) gives that $|\lambda_1,\dots,\lambda_M \rangle$
 is the eigenvector of the transfer matrix (\ref{TransferMatrix}) $\tau(\lambda) = A(\lambda)+D(\lambda)$
 \be
 \label{EigenVal}
 \begin{array}{c}
 \Bigl( A(\lambda)+D(\lambda) \Bigr) |\; \{ \lambda \} \; \rangle =
 \Lambda(\lambda, \{ \lambda \}) |\; \{ \lambda \} \; \rangle \; , \;\;\;\;
 \{ \lambda \} \equiv \{ \lambda_1,\dots,\lambda_M \} \; , \\ [0.3cm]
 \Lambda(\lambda, \{ \lambda \}) = \alpha(\lambda) \prod\limits_{i=1}^M f(\lambda_i,\lambda) +
 \delta(\lambda) \prod\limits_{i=1}^M f(\lambda,\lambda_i)
 \end{array}
 \ee
 if the set of parameters $\{ \lambda_1,\dots,\lambda_M \}$ satisfies the so-called Bethe equations:
 \be
 \label{BAE}
 \begin{array}{c}
 \Phi_i(\lambda,\lambda_1,...,\lambda_M) + \Psi_i(\lambda,\lambda_1,...,\lambda_M) = 0\; \Rightarrow \; \\ [0.2cm]
 \alpha(\lambda_i) g(\lambda,\lambda_i)
 \prod\limits_{\stackrel{k=1}{_{k \neq i}}}^M
 f(\lambda_k,\lambda_i) + \delta(\lambda_i) g( \lambda_i, \lambda)
 \prod\limits_{\stackrel{k=1}{_{k \neq i}}}^M
 f(\lambda_i, \lambda_k) = 0.
 \end{array}
 \ee
If the Bethe equations are satisfied, we call the Bethe vectors $\ket{\lambda_1,\dots,\lambda_M}$ on-shell, otherwise off-shell.

For the XXX chain, using explicit formulas (\ref{com5}) and (\ref{pseudovac1}) we write (\ref{EigenVal}) and (\ref{BAE}) in the form
 \be
  \label{EigenVal1}
  \begin{array}{c}
  \left( A(\lambda) + D(\lambda) \right) |\lambda_1,\dots,\lambda_M \rangle =
  \Lambda(\lambda, \lambda_i) |\lambda_1,\dots,\lambda_M \rangle, \\ [0.2cm]
  \displaystyle{ \Lambda(\lambda, \{ \lambda \}) = (\lambda+1)^L \prod\limits_{i=1}^M
  \frac{\lambda_i - \lambda+1}{\lambda_i - \lambda} +
 \lambda^L \prod\limits_{i=1}^M \frac{\lambda - \lambda_i+1}{\lambda - \lambda_i} } \;
  \end{array}
\ee
if the set of parameters $\{ \lambda_1, \dots , \lambda_M \}$ satisfies
the Bethe equations in the following form:
\be
\left( \frac{\lambda_k+1}{\lambda_k} \right)^L = \prod_{\substack{j=1\\ j\neq k}}^{M} \frac{\lambda_k-\lambda_j+1}{\lambda_k-\lambda_j-1} = -\prod_{j=1}^{M} \frac{\lambda_k-\lambda_j+1}{\lambda_k-\lambda_j-1}. \label{BAE1}
\ee
The Bethe equations for the XXZ chain possess the following form:
\be
\left( \frac{q-\lambda_k q^{-1}}{1-\lambda_k} \right)^L = (-1)^{M-1} \prod_{\substack{j=1\\ j\neq k}}^{M} \frac{\lambda_j q-\lambda_k q^{-1}}{\lambda_k q-\lambda_j q^{-1}} =
 \prod_{\substack{j=1\\ j\neq k}}^{M} \frac{\lambda_k q^{-1}-\lambda_j q}{\lambda_k q-\lambda_j q^{-1}}.\lb{BAExxz}
\ee Setting $\lambda_j=q^{-2\alpha_j}$ and $q=e^{h/2}$ in
\eqref{BAExxz}, we get \eqref{BAE1} in
 the limit $h\rightarrow 0$. The corresponding eigenvalue is
 \begin{equation}
 \Lambda(\lambda, \{ \lambda \}) = (\lambda^{-1/2}q - \lambda^{1/2}q^{-1})^L \prod\limits_{i=1}^M
  \frac{\lambda_i q - \lambda q^{-1}}{\lambda_i - \lambda} +
 (\lambda^{-1/2} - \lambda^{1/2})^L \prod\limits_{i=1}^M \frac{\lambda q - \lambda_i q^{-1}}{\lambda - \lambda_i}.
 \end{equation}

\section{Generalization of the two-component model}
\setcounter{equation}0

In the literature, cf. \cite{IzKor}, \cite{Slav} etc., there appears a so-called two-component model. The two-component model was introduced to avoid problems with computation of correlation functions for local operators attached to some site $x$ of the chain in the algebra of global operators \eqref{monod1} $A(\lambda)$, $B(\lambda)$, $C(\lambda)$ and $D(\lambda)$ defined on the chain as a whole.

We divide the chain $[1,\dots,L]$ into two components $[1,\dots,x]$ and $[x+1,\dots,L]$. Then we have the Hilbert space splitted into two parts $\mathscr{H}=\mathscr{H}_1\otimes \mathscr{H}_2$ where $\mathscr{H}_1=h_1\otimes\dotsm\otimes h_x$ and $\mathscr{H}_2=h_{x+1}\otimes\dotsm\otimes h_{L}$. We see that pseudovacuum $\ket{0}\in\mathscr{H}$ is of the form $\ket{0}=\ket{0}_1\otimes\ket{0}_2$ where $\ket{0}_1\in\mathscr{H}_1$ and $\ket{0}_2\in\mathscr{H}_2$. We define on $V_a\otimes \mathscr{H}_1\otimes\mathscr{H}_2$ the monodromy matrix for each component
\begin{align}
T_1(\lambda) & = L_{a,1}(\lambda)\dotsm L_{a,x}(\lambda)=\begin{pmatrix} A_1(\lambda) & B_1(\lambda)\\ C_1(\lambda) & D_1(\lambda)\\ \end{pmatrix}, \\
\intertext{resp.}
T_2(\lambda) & = L_{a,{x+1}}(\lambda)\dotsm L_{a,L}(\lambda)=\begin{pmatrix} A_2(\lambda) & B_2(\lambda)\\ C_2(\lambda) & D_2(\lambda)\\ \end{pmatrix}.
\end{align}
Each of these monodromy matrices satisfies exactly the same commutation relations \eqref{FCRglobal} as the original undivided monodromy matrix \eqref{monod}.
Moreover, we have
\begin{align}
\label{a12d12}
    A_j(\lambda) \ket{0}_j &= \alpha_j(\lambda)\ket{0}_j, & D_j(\lambda) \ket{0}_j &= \delta_j(\lambda)\ket{0}_j, & C_j(\lambda) \ket{0}_j &= 0.
\end{align}
Operators corresponding to different components mutually commute. From construction, it is easy to see that
\be\lb{2comp00}
\alpha(\lambda)=\alpha_1(\lambda)\alpha_2(\lambda),\qquad \delta(\lambda)=\delta_1(\lambda)\delta_2(\lambda).
\ee
For the whole chain $[1,\dots,L]$ the full monodromy matrix $T$ is
 \be
 \lb{full}
 \begin{array}{c}
 T(\lambda) =\begin{pmatrix} A(\lambda) & B(\lambda)\\ C(\lambda) & D(\lambda)\\ \end{pmatrix} =
  T_1(\lambda) \, T_2(\lambda) = \\ [0.3cm]
= \begin{pmatrix} A_1(\lambda) A_2(\lambda) + B_1(\lambda)C_2(\lambda) &
 A_1(\lambda) B_2(\lambda) + B_1(\lambda)D_2(\lambda) \\
 C_1(\lambda)A_2(\lambda) + D_1(\lambda)C_2(\lambda) &
  C_1(\lambda)B_2(\lambda) + D_1(\lambda) D_2(\lambda) \\ \end{pmatrix}  ,
  \end{array}
 \ee
and the $M$-magnon state is represented in the form
  \be
 \lb{bae-M}
 |\lambda_1,\dots,\lambda_M \rangle = \prod_{k=1}^M B(\lambda_k) \ket{0} =
 \prod_{k=1}^M
 \Bigl(A_1(\lambda_k) B_2(\lambda_k) + B_1(\lambda_k)D_2(\lambda_k) \Bigr) \ket{0}_1\otimes\ket{0}_2 \; .
 \ee

The beautiful result of Izergin and Korepin \cite{IzKor} states that the Bethe vectors of the full model can be expressed in terms of the Bethe vectors of its components. To obtain this expression, we should
 commute in (\ref{bae-M}) all operators $A_1(\lambda_k)$ and $D_2(\lambda_k)$  to the
 right with the help of (\ref{com1}) and (\ref{com2}) and then use (\ref{a12d12}).
 Finally, we obtain the following result \cite{IzKor}.

\begin{prop}\lb{prop:2comp}
An arbitrary Bethe vector corresponding to the full system can  be expressed in terms of the Bethe vectors of the first and second component. Let $I=\{\lambda_1,\dots,$ $ \lambda_M\}$ be a finite set
 of spectral parameters.  To concise notation below, we will consider the set $I$
 as a finite set of indices
 $I=\{1,\dots, M\}$, then
\begin{align}\lb{2comp}
    & \prod_{k\in I} B(\lambda_k) \ket{0} = \nonumber\\
    & \sum_{I_1\cup I_2} \prod_{k_1\in I_1} \Bigl( \delta_2(\lambda_{k_1}) B_1(\lambda_{k_1}) \Bigr) \prod_{k_2\in I_2} \Bigl(\alpha_1(\lambda_{k_2}) B_2(\lambda_{k_2}) \Bigr) \ket{0}_1\otimes\ket{0}_2 \prod_{k_1\in I_1} \prod_{k_2\in I_2} f(\lambda_{k_1},\lambda_{k_2})
\end{align}
where $f(\lambda_{k_1},\lambda_{k_2})$ is defined in \eqref{com5} resp. \eqref{com6} and the summation is performed over all divisions of the index set $I$ into two disjoint subsets $I_1$ and $I_2$ where $I=I_1\cup I_2$.
\end{prop}

\begin{proof}
The proof is just a matter of commutation relations \eqref{FCRglobal} resp. \eqref{com1}-\eqref{com5}. We use induction on the number of elements $M$ of the index set $I$. We see that
\be
B(\lambda) \ket{0} = \Bigl( A_1(\lambda) B_2(\lambda) + B_1(\lambda)D_2(\lambda) \Bigr) \ket{0}_1\otimes\ket{0}_2 = \Bigl( \alpha_1(\lambda) B_2(\lambda) + \delta_2(\lambda) B_1(\lambda) \Bigr) \ket{0}_1\otimes\ket{0}_2
\ee
which is exactly the formula \eqref{2comp} for $M=1$. Let us suppose that \eqref{2comp} is valid for the index set $I=\{1,\dots,M-1\}$. Then we have
\begin{align}
& B(\lambda) \prod_{k\in I} B(\lambda_k) \ket{0} = \Bigl( A_1(\lambda) B_2(\lambda) + B_1(\lambda) D_2(\lambda)  \Bigr) \times \nonumber\\
&\; \times \sum_{\stackrel{I_1,I_2}{I=I_1\cup I_2}} \prod_{k_1\in I_1} \Bigl( \delta_2(\lambda_{k_1}) B_1(\lambda_{k_1}) \Bigr) \prod_{k_2\in I_2} \Bigl(\alpha_1(\lambda_{k_2}) B_2(\lambda_{k_2}) \Bigr) \ket{0}_1\otimes\ket{0}_2 \prod_{k_1\in I_1} \prod_{k_2\in I_2} f(\lambda_{k_1},\lambda_{k_2}) =\nonumber
\end{align}
\begin{align}
&= \sum_{\stackrel{I_1,I_2}{I=I_1\cup I_2}}  \Bigl( A_1(\lambda) \prod_{k_1\in I_1} \delta_2(\lambda_{k_1}) B_1(\lambda_{k_1}) \Bigr) \Bigl( B_2(\lambda) \prod_{k_2\in I_2} \alpha_1(\lambda_{k_2}) B_2(\lambda_{k_2}) \Bigr) \ket{0}_1\otimes\ket{0}_2 \times \nonumber\\
&\qquad \times \prod_{k_1\in I_1} \prod_{k_2\in I_2} f(\lambda_{k_1},\lambda_{k_2})\  + \nonumber\\
&\quad + \sum_{\stackrel{I_1,I_2}{I=I_1\cup I_2}}  \Bigl( B_1(\lambda) \prod_{k_1\in I_1} \delta_2(\lambda_{k_1}) B_1(\lambda_{k_1}) \Bigr)  \Bigl( D_2(\lambda) \prod_{k_2\in I_2} \alpha_1(\lambda_{k_2}) B_2(\lambda_{k_2}) \Bigr) \ket{0}_1\otimes\ket{0}_2 \times \nonumber\\
&\qquad \times \prod_{k_1\in I_1} \prod_{k_2\in I_2} f(\lambda_{k_1},\lambda_{k_2}).\lb{2comp01}
\end{align}
In the first sum we use \eqref{bae03} to commute $A_1(\lambda)$ with $\prod_{k_1\in I_1} B_1(\lambda_{k_1})$ resp. \eqref{bae06} to commute $D_2 (\lambda)$ with $\prod_{k_2\in I_2}B_2(\lambda_{k_2})$ in the second sum. Using just the second term in \eqref{bae03} we get for the first sum:
\begin{align}
& \sum_{\stackrel{I_1,I_2}{I=I_1\cup I_2}} \sum_{k\in I_1} g(\lambda,\lambda_k) \alpha_1(\lambda_k)  \delta_2(\lambda_k) B_1(\lambda) B_2(\lambda) \prod_{\substack{j\in I_1 \\ j\neq k}} \delta_2(\lambda_j) B_1(\lambda_j) \prod_{i\in I_2} \alpha_1(\lambda_i) B_2(\lambda_i) \ket{0}_1\otimes\ket{0}_2 \nonumber\\
& \qquad \times \prod_{\substack{l\in I_1\\ l\neq k}} f(\lambda_l,\lambda_k) \prod_{k_1\in I_1}  \prod_{k_2\in I_2}  f(\lambda_{k_1},\lambda_{k_2}).\lb{2comp02}
\end{align}
Similarly, using just the second term in \eqref{bae06} we get for the second sum:
\begin{align}
& \sum_{\stackrel{I_1,I_2}{I=I_1\cup I_2}} \sum_{k\in I_2} g(\lambda_k,\lambda) \alpha_1(\lambda_k) \delta_2(\lambda_k) B_1(\lambda) B_2(\lambda)
\prod_{j\in I_1} \delta_2(\lambda_j) B_1(\lambda_j) \prod_{\substack{i\in I_2\\ i\neq k}} \alpha_1(\lambda_i) B_2(\lambda_i) \ket{0}_1\otimes\ket{0}_2 \nonumber\\
&\qquad\times \prod_{\substack{l\in I_2\\ l\neq k}} f(\lambda_k,\lambda_l) \prod_{k_1\in I_1}  \prod_{k_2\in I_2}  f(\lambda_{k_1},\lambda_{k_2}) = \nonumber\\
&= \sum_{\stackrel{I_1',I_2'}{I=I_1'\cup I_2'}} \sum_{k\in I'_1} g(\lambda_k,\lambda) \alpha_1(\lambda_k)  \delta_2(\lambda_k) B_1(\lambda) B_2(\lambda) \prod_{\substack{j\in I'_1 \\ j\neq k}} \delta_2(\lambda_j) B_1(\lambda_j) \prod_{i\in I'_2} \alpha_1(\lambda_i) B_2(\lambda_i) \ket{0}_1\otimes\ket{0}_2 \nonumber\\
&\qquad\times \prod_{k_1\in I'_1}  \prod_{k_2\in I'_2}  f(\lambda_{k_1},\lambda_{k_2}) \prod_{\substack{m\in I'_1\\m\neq k}} f(\lambda_m,\lambda_k)\lb{2comp03}
\end{align}
where we introduced new partition $I'_1=I_1\cup \{k\}$ and $I'_2 = I_2\backslash \{k\}$. We see that \eqref{2comp03} is almost the same as \eqref{2comp02} with only one difference. In \eqref{2comp02} there appears a factor $g(\lambda,\lambda_k)$ and in \eqref{2comp03} there appears $g(\lambda_k,\lambda)$. Using the fact that $g(\lambda,\lambda_k)=-g(\lambda_k,\lambda)$, cf. \eqref{com5}, we see that these two sums cancel each other. Therefore, only the first parts of \eqref{bae03} and \eqref{bae06} contribute to \eqref{2comp01}. We get
\begin{align}
&   B(\lambda) \prod_{k\in I} B(\lambda_k) \ket{0} = \nonumber\\
&   =\sum_{\stackrel{I_1,I_2}{I=I_1\cup I_2}}  \Bigl( \alpha_1(\lambda) \prod_{k_1\in I_1} f(\lambda_{k_1},\lambda) \delta_2(\lambda_{k_1}) B_1(\lambda_{k_1}) \Bigr) \Bigl( B_2(\lambda) \prod_{k_2\in I_2} \alpha_1(\lambda_{k_2}) B_2(\lambda_{k_2}) \Bigr) \ket{0}_1\otimes\ket{0}_2 \nonumber\\
&\qquad\times \prod_{k_1\in I_1} \prod_{k_2\in I_2} f(\lambda_{k_1},\lambda_{k_2})\ + \nonumber\\
&+ \sum_{\stackrel{I_1,I_2}{I=I_1\cup I_2}}  \Bigl( B_1(\lambda) \prod_{k_1\in I_1} \delta_2(\lambda_{k_1}) B_1(\lambda_{k_1}) \Bigr)  \Bigl( \delta_2(\lambda) \prod_{k_2\in I_2} f(\lambda,\lambda_{k_2}) \alpha_1(\lambda_{k_2}) B_2(\lambda_{k_2}) \Bigr) \ket{0}_1\otimes\ket{0}_2  \nonumber\\
&\qquad\times \prod_{k_1\in I_1} \prod_{k_2\in I_2} f(\lambda_{k_1},\lambda_{k_2})\lb{2comp04}
\end{align}
which proves the induction.
\end{proof}

This result can be straightforwardly generalized to an arbitrary number of components $N\leq L$. 

\begin{prop}\lb{prop:Ncomp}
An arbitrary Bethe vector of the full system can be expressed in terms of the Bethe vectors of its components. For $N\leq L$ components the Bethe vector is of the form
\begin{align}
    & \prod_{k\in I} B(\lambda_k) \ket{0} = \sum_{I_1\cup I_2\cup\dotsb\cup I_N} \prod_{k_1\in I_1} \prod_{k_2\in I_2} \dotsm \prod_{k_N\in I_N} \prod_{1\leq i<j\leq N} \Bigl( \alpha_i(\lambda_{k_j}) \delta_j(\lambda_{k_i}) f(\lambda_{k_i},\lambda_{k_j}) \Bigr) \nonumber\\
  & \qquad \times B_1(\lambda_{k_1}) \ket{0}_1 \otimes B_2(\lambda_{k_2}) \ket{0}_2 \otimes \dots \otimes B_N(\lambda_{k_N}) \ket{0}_N\lb{Ncomp}
\end{align}
where summation is performed over all divisions of the set $I$ into its $N$ mutually disjoint subsets $I_1,I_2,\dots, I_{N}$.
\end{prop}

\begin{proof}
The proof is simply performed by induction on the number of components $N$ and by using \eqref{2comp}. For $N=2$ is \eqref{Ncomp} just \eqref{2comp}. Let us suppose that \eqref{Ncomp} is valid for some $N<L$ and make induction step to $N+1$. The chain $[1,\dots,L]$ is divided into $N$ subchains $[1,\dots,x_1]$, $[x_1+1,\dots,x_2]$, etc. up to $[x_{N-1}+1,\dots,L]$. Let us divide the last interval, if possible, into two subchains $[x_{N-1}+1,\dots,x_{N}]$ and $[x_{N}+1,\dots,L]$ and apply \eqref{2comp} to set of $B$ operators $\prod_{k_N\in I_N} B_N(\lambda_{k_N}) \ket{0}_N$. We get
\begin{align}
    & \prod_{k_N\in I_N} B_N(\lambda_{k_N}) \ket{0}_N =   \sum_{I'_{N}\cup I'_{N+1}} \prod_{k_{N}\in I'_{N}} \prod_{k_{N+1}\in I'_{N+1}} \delta'_{N+1}(\lambda_{k_{N}}) \alpha'_{N}(\lambda_{k_{N+1}}) f(\lambda_{k_{N}},\lambda_{k_{N+1}}) \nonumber\\
    &\qquad \times B'_{N}(\lambda_{k_{N}}) B'_{N+1}(\lambda_{k_{N+1}}) \ket{0}'_{N}\otimes \ket{0}'_{N+1} \lb{Ncomp1}
\end{align}
where the sum goes over all divisions of $I_N$ into its two disjoint subsets $I'_N$ and $I'_{N+1}$ such that $I_N=I'_N \cup I'_{N+1}$ and operators $B'_N(\lambda)$ and $B'_{N+1}(\lambda)$ act on the new subchains $[x_{N-1}+1,\dots,x_{N}]$ resp. $[x_{N}+1,\dots,L]$; the same for $\alpha'_{N}(\lambda)$ resp. $\delta'_{N+1}(\lambda)$ and the pseudovacuum vectors $\ket{0}'_{N}$ resp. $\ket{0}'_{N+1}$. Let us remind that
\be
\prod_{k_N \in I_N} = \prod_{k_N \in I'_N} \ \prod_{k_{N+1} \in I'_{N+1}}.
\ee
Inserting \eqref{Ncomp1} into  induction assumption \eqref{Ncomp}  we get
\begin{align}
    & \prod_{k\in I} B(\lambda_k) \ket{0} = \nonumber\\
    & \sum_{I_1\cup I_2\cup\dotsb\cup I'_N \cup I'_{N+1}} \prod_{k_1\in I_1} \prod_{k_2\in I_2} \dotsm \prod_{k_N\in I'_N} \prod_{k_{N+1}\in I'_{N+1}}  \prod_{1\leq i<j\leq N+1} \Bigl( \alpha_i(\lambda_{k_j}) \delta_j(\lambda_{k_i}) f(\lambda_{k_i},\lambda_{k_j}) \Bigr) \nonumber\\
  & \qquad \times B_1(\lambda_{k_1}) \ket{0}_1 \otimes B_2(\lambda_{k_2}) \ket{0}_2 \otimes \dots \otimes B'_N(\lambda_{k_N}) \ket{0}'_N \otimes B'_{N+1}(\lambda_{k_{N+1}}) \ket{0}'_{N+1} \lb{Ncomp2}
\end{align}
which proves the induction.
\end{proof}

\section{Bethe vectors}
\setcounter{equation}0

In this section, we will see that computation of the Bethe vectors in the algebraic Bethe ansatz is just a matter of using proposition \ref{prop:Ncomp}. By assumption we have a chain of length $L$. Let us divide it into $L$ components, i.e. into $L$ subchains of length one (1-chains). Using proposition \ref{prop:Ncomp} we get for the $M$-magnon (Bethe vector) with $M\leq L$:
\begin{align}
  \prod_{k=1}^M B(\lambda_k) \ket{0} & =  \sum_{I_1\cup I_2\cup\dots\cup I_L}
    \prod_{k_1\in I_1} \prod_{k_2\in I_2} \dotsm \prod_{k_L\in I_L}  \prod_{1\leq i<j\leq L} \Bigl( \alpha_i(\lambda_{k_j}) \delta_j(\lambda_{k_i}) f(\lambda_{k_i},\lambda_{k_j}) \Bigr) \nonumber\\
    &\qquad\times  B_1(\lambda_{k_1}) \ket{0}_1 \otimes B_2(\lambda_{k_2}) \ket{0}_2 \otimes \dots \otimes B_L(\lambda_{k_L}) \ket{0}_L.\lb{Bethe}
\end{align}
It can be easily seen that for 1-chain, i.e. for a chain with Hilbert space $h=\mathbb{C}^2$,
\be
    B(\lambda) B(\mu)\ket{0} = 0.
\ee
Therefore, the sum over all divisions of $\{1,\dots,M\}$ into $L$ subsets contains just divisions into subsets containing at most one element, i.e. $|I_j|=0,1$. Moreover, only $M$ of them is nonempty, let us denote them $I_{n_1},I_{n_2},\dots,I_{n_M}$. We have to sum over all possible combinations of such sets, i.e. over all $M$-tuples $n_1<n_2<\dotsb<n_M$. Next, we have to sum over all distributions of the parameters $\lambda_1,\lambda_2,\dots,\lambda_M$ into the sets $I_{n_1},\dots,I_{n_M}$. We can simplify our life assuming that $\lambda_j\in I_{n_j}$. Then, by summing over all permutations $\sigma_\lambda\in S_M$ of $\{\lambda_1,\dots,\lambda_M\}$, we get exactly all the other distributions.

Let us study what happens to the coefficient
\be
    \prod_{1\leq i<j\leq L} \Bigl( \alpha_i(\lambda_{k_j}) \delta_j(\lambda_{k_i}) f(\lambda_{k_i},\lambda_{k_j}) \Bigr).
\ee
It is easy to see that
\be
    \prod_{1\leq i<j\leq L} \alpha_i(\lambda_{k_j})= \prod_{j=1}^L \prod_{i=1}^{j-1} \alpha_i(\lambda_{k_j}),
\ee
but only $\lambda_{k_j}$ from the sets $I_{n_1},\dots,I_{n_M}$ are relevant and by assumption $\lambda_j\in I_{n_j}$. Therefore, we can replace
\be
    \prod_{1\leq i<j\leq L} \alpha_i(\lambda_{k_j}) \longrightarrow \prod_{j=1}^M \prod_{i=1}^{n_j-1} \alpha_i(\lambda_{j}).
\ee
Similar considerations can be conducted for both $\delta_j(\lambda_{k_i})$ and both $f(\lambda_{k_i},\lambda_{k_j})$. We get
\begin{align}
 & \prod_{k=1}^M B(\lambda_k) \ket{0} = \sum_{1\leq n_1 < n_2 <\dots< n_M\leq L} \; \sum_{\sigma_\lambda \in S_M} \sigma_\lambda \Biggl( \prod_{j=1}^M \Bigl( \prod_{i=1}^{n_j-1} \alpha_i(\lambda_j) \prod_{i=n_j+1}^L \delta_i(\lambda_j) \prod_{i=1}^{j-1} f(\lambda_i,\lambda_j) \Bigr) \nonumber\\
    &\qquad\times  B_{n_1}(\lambda_{1}) B_{n_2}(\lambda_{2}) \dotsb B_{n_M}(\lambda_{M}) \Biggr) \ket{0}_1 \otimes \ket{0}_2 \otimes \dots \otimes  \ket{0}_L. \lb{Bethe1}
\end{align}
For 1-chain, it holds that $B(\lambda)=B$ is parameter independent. Moreover, eigenvalues $\alpha_i(\lambda)=a(\lambda)$,  $\delta_i(\lambda)=d(\lambda)$ are still the same for all components $i=1,\dots,L$, where $a(\lambda)$ and $d(\lambda)$ are defined in \eqref{adXXX} resp. \eqref{adXXZ}. We get
\begin{align}\label{Bethe2}
    & \prod_{k=1}^M B(\lambda_k) \ket{0} = \sum_{1\leq n_1 < n_2 <\dots< n_M\leq L} \; \sum_{\sigma\in S_M} \sigma_\lambda \Biggl( \prod_{j=1}^M a(\lambda_{j})^{n_j -1}  d(\lambda_{j})^{L-n_j} \prod_{i=1}^{j-1}  f(\lambda_{i},\lambda_{j}) \Biggr) \nonumber\\
    & \qquad \times B_{n_1}B_{n_2}\dotsb B_{n_M} \ket{0}_1 \otimes  \ket{0}_2 \otimes \dots \otimes \ket{0}_L= \nonumber\\
    & = \prod_{j=1}^M \frac{d(\lambda_j)^L}{a(\lambda_j)} \sum_{1\leq n_1 < n_2 <\dots< n_M\leq L} \; \sum_{\sigma\in S_M} \sigma_\lambda \Biggl( \prod_{1\leq i<j \leq M}  f(\lambda_{i},\lambda_{j}) \prod_{j=1}^M \left( \frac{a(\lambda_j)}{d(\lambda_j)} \right)^{n_j}  \Biggr) \nonumber\\
    & \qquad \times B_{n_1}B_{n_2}\dotsb B_{n_M} \ket{0}_1 \otimes  \ket{0}_2 \otimes \dots \otimes \ket{0}_L.
\end{align}

\section{Inhomogeneous Bethe ansatz}
\setcounter{equation}0

We start with the inhomogeneous monodromy matrix
\be
    T_a^{\vec{\xi}}(\lambda) = L_{a,1}(\lambda+\xi_1) L_{a,2}(\lambda+\xi_2)\dotsm L_{a,L}(\lambda+\xi_L)
\ee
where $L_{a,j}(\lambda)$ are the Lax operators defined in \eqref{Lax} resp. \eqref{RmatXXZ} depending on whether we consider XXX or XXZ spin chain. Let us remark that for the XXZ chain the monodromy matrix is of the form
\be
    T_a^{\vec{\xi}}(\lambda) = L_{a,1}(\lambda\cdot \xi_1) L_{a,2}(\lambda\cdot\xi_2)\dotsm L_{a,L}(\lambda\cdot\xi_L).
\ee
In what follows, we will use the notation connected with the XXX chain but we can do for the XXZ chain the same as well.

Expressing $T_a^{\vec{\xi}}(\lambda)$ in the auxiliary space $V_a$ we get
\be
    T_a^{\vec{\xi}}(\lambda) = \begin{pmatrix} A^{\vec{\xi}}(\lambda)  & B^{\vec{\xi}}(\lambda) \\ C^{\vec{\xi}}(\lambda) & D^{\vec{\xi}}(\lambda)  \end{pmatrix}
\ee
where, again, the operators $A^{\vec{\xi}} (\lambda)$, $B^{\vec{\xi}}(\lambda)$, $C^{\vec{\xi}}(\lambda)$ and $D^{\vec{\xi}}(\lambda)$ act in $\mathscr{H} = h_1\otimes\dotsb\otimes h_L$. Acting on the pseudovacuum vector $\ket{0}\in\mathscr{H}$
we get
\begin{align}
    A^{\vec{\xi}}(\lambda)\ket{0} &= \alpha^{\vec{\xi}}(\lambda) \ket{0}, \\
     D^{\vec{\xi}}(\lambda)\ket{0} &= \delta^{\vec{\xi}}(\lambda) \ket{0}, \\
      C^{\vec{\xi}}(\lambda)\ket{0} & = 0
\end{align}
where
\begin{align}
    \alpha^{\vec{\xi}}(\lambda) & = a(\lambda+\xi_1) a(\lambda+\xi_2)\dotsm a(\lambda+\xi_L), \\
    \delta^{\vec{\xi}}(\lambda) & = d(\lambda+\xi_1) d(\lambda+\xi_2)\dotsm d(\lambda+\xi_L).
\end{align}
Here, the functions $a(\lambda)$ and $d(\lambda)$ are defined in \eqref{adXXX} for XXX resp. in \eqref{adXXZ} for XXZ.

For the inhomogeneous version we can introduce the same $N$-component model as for the homogeneous Bethe ansatz. For the 2-component model, for example, we have
\be
    T_a^{\vec{\xi}}(\lambda) = \underbrace{L_{a,1}(\lambda+\xi_1) \dotsm L_{a,x}(\lambda+\xi_x)}_{\text{1st component}} \underbrace{L_{a,x+1}(\lambda+\xi_{x+1})\dotsm L_{a,L}(\lambda+\xi_L)}_{\text{2nd component}} = T^{\vec{\xi}_1}_a(\lambda) T^{\vec{\xi}_2}_a(\lambda)
\ee
where $\vec{\xi}_1=(\xi_1,\dots,\xi_x)$ resp. $\vec{\xi}_2=(\xi_{x+1},\dots,\xi_L)$. We have
\be
 A^{\vec{\xi}}(\lambda) \ket{0} = \alpha_1^{\vec{\xi}_1}(\lambda) \alpha_2^{\vec{\xi}_2}(\lambda) \ket{0}, \qquad
 D^{\vec{\xi}}(\lambda) \ket{0} = \delta_1^{\vec{\xi}_1}(\lambda) \delta_2^{\vec{\xi}_2}(\lambda) \ket{0}.
\ee
A very important property of the inhomogeneous chain is that its operators  $A^{\vec{\xi}} (\lambda)$, $B^{\vec{\xi}}(\lambda)$, $C^{\vec{\xi}}(\lambda)$ and $D^{\vec{\xi}}(\lambda)$ satisfy the same fundamental commutation relations as the homogeneous chain \eqref{com}-\eqref{com5}, i.e. commutation relations are independent of the inhomogeneity parameters $\vec{\xi}$. Therefore, an analogy of propositions \ref{prop:2comp} and \ref{prop:Ncomp} can be easily formulated.

\begin{prop} Let $N\leq L$. An arbitrary Bethe vector of the full system can be expressed in terms of the Bethe vectors of its $N$ components
\begin{align}
    \prod_{k\in I} B^{\vec{\xi}}(\lambda_k) \ket{0} & =
    \sum_{I_1\cup\dots\cup I_N} \prod_{k_1\in I_1} \dotsm \prod_{k_N\in I_N} \prod_{1\leq i<j\leq N} \Bigl(\alpha_i^{\vec{\xi}_i}(\lambda_{k_j}) \delta_j^{\vec{\xi}_j}(\lambda_{k_i}) f(\lambda_{k_i},\lambda_{k_j}) \Bigr) \nonumber\\
    & \qquad\times B_1^{\vec{\xi}_1}(\lambda_{k_1}) B_2^{\vec{\xi}_2}(\lambda_{k_2}) \dotsm B_N^{\vec{\xi}_N}(\lambda_{k_N}) \ket{0}.
\end{align}
\end{prop}

To get an explicit formula for the Bethe vectors, we have to divide the chain into $L$ components of length 1, as we did in the last section. We get for the $M$-magnon
\begin{align}
    & \prod_{k=1}^M B^{\vec{\xi}}(\lambda_k) \ket{0} = \nonumber\\
    & = \sum_{1\leq n_1<\dotsb<n_M\leq L}\ \sum_{\sigma_\lambda\in S_M} \sigma_\lambda \Biggl( \prod_{j=1}^M \Bigl( \prod_{i=1}^{n_j-1} \alpha_i^{\xi_i}(\lambda_j) \prod_{i=n_j+1}^L \delta_i^{\xi_i}(\lambda_j) \prod_{i=1}^{j-1} f(\lambda_i,\lambda_j) \Bigr) \nonumber\\
    & \qquad \times B_{n_1}^{\xi_{n_1}}(\lambda_1)\dotsm B_{n_M}^{\xi_{n_M}}(\lambda_M) \Biggr) \ket{0}= \nonumber\displaybreak[0]\\
    & = \sum_{1\leq n_1<\dotsb<n_M\leq L}\ \sum_{\sigma_\lambda\in S_M} \sigma_\lambda   \Biggl( \prod_{j=1}^M \Bigl( \prod_{i=1}^{n_j-1} \alpha_i^{\xi_i}(\lambda_j) \prod_{i=n_j+1}^L \delta_i^{\xi_i}(\lambda_j) \prod_{i=1}^{j-1} f(\lambda_i,\lambda_j) \Bigr) \Biggr)  B_{n_1} \dotsm B_{n_M}  \ket{0} =\nonumber\displaybreak[0]\\
    & = \prod_{j=1}^M \prod_{i=1}^L d(\lambda_j +\xi_i) \sum_{1\leq n_1<\dotsb<n_M\leq L}  \sum_{\sigma_\lambda\in S_M} \sigma_\lambda   \Biggl( \prod_{j=1}^M \frac{1}{a(\lambda_j+\xi_{n_j})}  \prod_{i=1}^{n_j} \frac{a(\lambda_j+\xi_i)}{d(\lambda_j+\xi_i)}  \prod_{i=1}^{j-1} f(\lambda_i,\lambda_j)  \Biggr) \nonumber\\
    & \qquad\times B_{n_1} \dotsm B_{n_M} \ket{0}
\end{align}
where again the $B$-operators $B_{n_j}^{\xi_{n_j}}(\lambda) = B_{n_j}$ are parameter independent for 1-chains.


\section{Free Fermions}
\setcounter{equation}0

In this Section we recall the well-known construction \cite{JoWi} of $L$-dimensional free fermion algebra in terms
of the Pauli matrices. First,
in $\mathbb{C}^2$ one can easily define $1$-dimensional fermions using the properties of the Pauli matrices.
Let
    \begin{align}
        \psi \equiv \sigma^+=\frac{1}{2}(\sigma^x+i\sigma^y), \qquad \bpsi\equiv\sigma^- =\frac{1}{2}(\sigma^x-i\sigma^y).
    \end{align}
Thus defined $\psi,\bpsi$ satisfy the fermionic relations
    \begin{equation}
        [\bpsi,\psi]_+ = \mathbb{I} \; , \;\;\;\; \psi^2 = 0 \; , \;\;\;\; \bpsi^2 = 0 \; .
    \end{equation}
For a tensor product of $L$ copies of $\mathbb{C}^2$ we can define fermions as
\begin{align}\label{fermions}
    \psi_k \equiv \left( \prod_{j=1}^{k-1} \sigma_j^z\right)\sigma_k^+,\qquad \bpsi_k \equiv \left( \prod_{j=1}^{k-1} \sigma_j^z\right)\sigma_k^- \; , \;\;\;\; (k=1,\dots,L) \; ,
\end{align}
where $\sigma_j^\alpha$ denotes the sigma matrix attached to the $j$-th vector space, i.e.
\begin{equation}
\sigma_j^\alpha = \mathbb{I}^{\otimes (j-1)}\otimes \sigma^\alpha \otimes \mathbb{I}^{\otimes (L-j)}.
\end{equation}
This concise notation is used throughout the whole text. Commutation relations for the fermions \eqref{fermions} are of the form
\begin{align}\label{FermCom}
[\bpsi_i,\psi_j]_+=\delta_{ij} \mathbb{I},\quad [\bpsi_i,\bpsi_j]_+=0, \quad [\psi_i,\psi_j]_+ = 0.
\end{align}
It is a straightforward task to check the following identities:
    \begin{align}
        &\bpsi_{k+1}\psi_k + \bpsi_k\psi_{k+1} + \bpsi_k\bpsi_{k+1} + \psi_{k+1}\psi_k = \sigma_k^x \sigma_{k+1}^x, \label{ferm1}\\
        &\bpsi_{k+1}\psi_k + \bpsi_k\psi_{k+1} - \bpsi_k\bpsi_{k+1} - \psi_{k+1}\psi_k = \sigma_k^y \sigma_{k+1}^y, \label{ferm2}\\
        &[\psi_k,\bpsi_{k}]=\sigma_k^z,\label{ferm3}\\
        &(1-2\bpsi_k\psi_k)(1-2\bpsi_{k+1}\psi_{k+1}) = \sigma_k^z\sigma_{k+1}^z.\label{ferm4}
    \end{align}

\section{Fermionic realization of XXX}\lb{sec:Fermi}
\setcounter{equation}0

We have seen that our definition \eqref{Lax} of the Lax operator $L_{a,i}(\lambda)$ led to expression \eqref{Lax1} which is in fact identical to the definition of the R-matrix \eqref{Rmatrix}. Let us remind that the identity operator $\mathbb{I}$ is  a member of the algebra of fermions because of commutation relation \eqref{FermCom}. Therefore, from expression \eqref{Lax1} for $L_{a,i}(\lambda)$ we see that it remains to know a fermionic realization only for the permutation operator $P_{a,i}$.

Let us start with the permutation operator $P_{k,k+1}$ which permutes the neighboring vector spaces $h_k$ and $h_{k+1}$. Due to identities \eqref{ferm1}-\eqref{ferm4} and definition of permutation operator \eqref{perm} it is straightforward to check that
\be \label{perm1}
    P_{k,k+1} = \mathbb{I} + \bpsi_{k+1}\psi_k +\bpsi_k\psi_{k+1}-\bpsi_k\psi_k-\bpsi_{k+1}\psi_{k+1}+2\bpsi_k\psi_k\bpsi_{k+1}\psi_{k+1}.
\ee
Problems appear when we try to find a fermionic realization of the permutation operator $P_{j,k}$ in non-neighboring vector spaces $h_j$, $h_k$ where $j<k-1$. It turns out that $P_{j,k}$ becomes non-local in terms of fermions. Using properties of the Pauli matrices, $P_{j,k}$ could be rewritten as
\be
    P_{j,k} = \frac{1}{2}(\mathbb{I} + \sigma_j^z\sigma_k^z) + (\sigma_j^+\sigma_k^- +\sigma_j^-\sigma^+_k).
\ee
The first part is local even in terms of fermions
\be
\frac{1}{2}(\mathbb{I} + \sigma_j^z\sigma_k^z )= \mathbb{I} -\bpsi_k\psi_k-\bpsi_{j}\psi_{j}+2\bpsi_k\psi_k\bpsi_{j}\psi_{j},
\ee
but the second part is nonlocal
\be
\sigma_j^+\sigma_k^- +\sigma_j^-\sigma^+_k = (\psi_j\bpsi_k +\bpsi_j\psi_k)\prod_{l=j}^{k-1} \sigma_l^z = (\psi_j\bpsi_k +\bpsi_j\psi_k)\prod_{l=j}^{k-1} (\mathbb{I}-2\bpsi_l\psi_l).
\ee
Therefore, the fermionic realization of $P_{j,k}$ for $j<k-1$ is a nonlocal operator.

The nonlocality of $P_{j,k}$ resp. $R_{j,k}(\lambda)$ is a serious problem. There appear difficulties when we attempt to express the monodromy matrix \eqref{monod} in terms of such nonlocal operators. We need to avoid the nonlocality.

Let us remind once again that $L_{a,i}(\lambda)=R_{a,i}(\lambda)$. For the R-matrix $R_{ab}(\lambda)$ satisfying the Yang-Baxter equation \eqref{YBE} we can define the matrix $\hat{R}_{ab}(\lambda)=R_{ab}(\lambda)P_{ab}$ which satisfies
\be\lb{YBEbraid}
    \hat{R}_{ab}(\lambda) \hat{R}_{bc}(\lambda+\mu) \hat{R}_{ab}(\mu) = \hat{R}_{ab}(\mu) \hat{R}_{bc}(\lambda+\mu) \hat{R}_{ab}(\lambda).
\ee
We substitute $L_{a,i}(\lambda)= \hat{R}_{a,i}(\lambda)P_{a,i}$ in the monodromy matrix \eqref{monod} and obtain a very convenient expression
\begin{align}
 T_{a}(\lambda) & = L_{a,1}(\lambda)L_{a,2}(\lambda)\dots L_{a,L}(\lambda) =
     \hat{R}_{a,1}(\lambda)P_{a,1} \hat{R}_{a,2}(\lambda)P_{a,2} \dots \hat{R}_{a,L}(\lambda)P_{a,L} = \nonumber \\
    & =\hat{R}_{a,1}(\lambda) \hat{R}_{1,2}(\lambda) \dots \hat{R}_{L-1,L}(\lambda) P_{a,1} P_{a,2} \dots P_{a,L} = \nonumber \\
    & =\hat{R}_{a,1}(\lambda) \hat{R}_{1,2}(\lambda) \dots \hat{R}_{L-1,L}(\lambda) P_{L-1,L} \dots P_{1,2} P_{a,1}. \label{monod2}
\end{align}
It contains the operators $\hat{R}_{k,k+1}$ resp. $P_{k,k+1}$ acting only in the neighboring spaces $h_k\otimes h_{k+1}$. From \eqref{perm1} we know the fermionic realization of $P_{k,k+1}$ and the fermionic realization of the R-matrix $\hat{R}_{k,k+1}(\lambda)$ is
\be \lb{hatR}
    \hat{R}_{k,k+1}(\lambda) = \lambda P_{k,k+1} +\mathbb{I} = (\lambda+1) \mathbb{I} + \lambda(\bpsi_{k+1}\psi_k +\bpsi_k\psi_{k+1}-\bpsi_k\psi_k-\bpsi_{k+1}\psi_{k+1}+2\bpsi_k\psi_k\bpsi_{k+1}\psi_{k+1}).
\ee

The natural next step is to express the monodromy matrix \eqref{monod2} as the $2\times 2$ matrix in the auxiliary space $V_a=\mathbb{C}^2$. For this purpose we rewrite \eqref{monod2} as
\be \lb{monod3}
    T_a(\lambda) = \hat{R}_{a,1}(\lambda) X(\lambda) P_{a,1}
\ee
where the operator $X(\lambda)$
\be
    X(\lambda) = \hat{R}_{1,2}(\lambda) \dots \hat{R}_{L-1,L}(\lambda) P_{L-1,L} \dots P_{1,2}
\ee
acts nontrivially only in the quantum spaces $\mathscr{H} = h_1\otimes\dots\otimes h_L$ and is a scalar in the auxiliary space $V_a$. Moreover, we know, due to equations \eqref{perm1} and \eqref{hatR}, how to express $X(\lambda)$ in terms of fermions.

What remains is to express $\hat{R}_{a,1}$ and $P_{a,1}$ as the $2\times 2$ matrix in the auxiliary space $V_a$. The permutation matrix \eqref{perm} can be rewritten as
\begin{gather}
    P_{a,1} = \frac{1}{2} \Bigl( \mathbb{I}\otimes \mathbb{I} + \sigma^x \otimes \sigma^x + \sigma^y \otimes \sigma^y + \sigma^z \otimes \sigma^z \Bigr) = \nonumber \\
    = \frac{1}{2} \left[ \left( \begin{array}{cc} \mathbb{I} & 0\\ 0 & \mathbb{I} \\ \end{array} \right)  + \left( \begin{array}{cc} 0 & \sigma^x \\ \sigma^x & 0 \\ \end{array} \right) + \left( \begin{array}{cc} 0 & -i\sigma^y \\ i\sigma^y & 0 \\ \end{array} \right) + \left( \begin{array}{cc} \sigma^z & 0 \\ 0 & -\sigma^z \\ \end{array} \right) \right] = \nonumber\\
    = \left( \begin{array}{cc} \frac{1}{2}(\mathbb{I}+\sigma^z) & \frac{1}{2}(\sigma^x -i\sigma^y) \\ \frac{1}{2}(\sigma^x +i\sigma^y) & \frac{1}{2}(\mathbb{I}-\sigma^z) \\ \end{array} \right) =
\intertext{and using \eqref{fermions} and \eqref{ferm3} we get }
    = \left( \begin{array}{cc} \psi_1\bpsi_1 & \bpsi_1 \\ \psi_1 & \bpsi_1\psi_1  \end{array}\right) = \left( \begin{array}{cc} \mathbb{I}-N_1 & \bpsi_1 \\ \psi_1 & N_1  \end{array}\right) \lb{Pferm}
\end{gather}
where $N_1=\bpsi_1\psi_1$. For $\hat{R}_{a,1}(\lambda)$, we get
\be \lb{Rferm}
    \hat{R}_{a,1}(\lambda) = \mathbb{I}_{a,i} + \lambda P_{a,1} = \left(  \begin{array}{cc} (\lambda+1)\mathbb{I} -\lambda N_1 & \lambda\bpsi_1 \\ \lambda \psi_1 & \lambda N_1 +\mathbb{I} \end{array}\right).
\ee
Using \eqref{Pferm} and \eqref{Rferm}, the monodromy matrix \eqref{monod3} can be written in the following form:
\begin{gather}
    T_a(\lambda) = \left(\begin{array}{cc} (\lambda+1)\mathbb{I} -\lambda N_1 & \lambda\bpsi_1 \\ \lambda \psi_1 & \lambda N_1 +\mathbb{I} \end{array}\right) X(\lambda) \left( \begin{array}{cc} \mathbb{I}-N_1 & \bpsi_1 \\ \psi_1 & N_1  \end{array}\right)  = \left(
    \begin{array}{c c}
        A(\lambda) & B(\lambda) \\
        C(\lambda) & D(\lambda) \\
    \end{array}
\right)
\end{gather}
where
\begin{align}
 A(\lambda) &= (\lambda +1 -\lambda N_1)X(\lambda)(1 -N_1) +\lambda \bar{\psi}_1 X(\lambda)\psi_1, \lb{Aferm}\\
 B(\lambda) &= (\lambda +1 -\lambda N_1) X(\lambda)\bar{\psi}_1 + \lambda \bar{\psi}_1 X(\lambda)N_1, \lb{Bferm}\\
 C(\lambda) &= \lambda \psi_1 X(\lambda) (1-N_1) + (\lambda N_1 +1) X(\lambda) \psi_1, \lb{Cferm}\\
 D(\lambda) &=  \lambda \psi_1 X(\lambda) \bar{\psi}_1 + (\lambda N_1 +1) X(\lambda) N_1.\lb{Dferm}
\end{align}

\section{Bethe vectors of XXX} \lb{sec:BetheXXX}
\setcounter{equation}0

The goal of our text is to find expression for the Bethe vectors \eqref{BetheVect}
\be
    \ket{\lambda_1,\dots,\lambda_M} = B(\lambda_1)\dots B(\lambda_M)\ket{0}.
\ee
For this purpose, the fermionic realization \eqref{Bferm} of the creation operator $B(\lambda)$ is very convenient. The operator $X(\lambda)=\hat{R}_{12}(\lambda)\dots\hat{R}_{L-1,L}(\lambda)P_{L-1,L}\dots P_{12}$ can be  written in terms of fermions due to equations \eqref{Pferm} and \eqref{Rferm}. From equation \eqref{vacuum}, our special representation, where $\ket{0}_k = \left( \begin{smallmatrix}1\\0\end{smallmatrix}\right)$, and the definition of free fermions \eqref{fermions} we can see that
\be
    \psi_k\ket{0} = 0
\ee
for all $k=1,\dots,L$.

If we were to write $B(\lambda)$ in the normal form, our work would be simple. Unfortunately, it seems a rather difficult task. Instead, we have to use the ``weak approach,'' i.e. to apply $B(\lambda)$ to the pseudovacuum $\ket{0}$ and try to commute the fermions $\bpsi_k$ to the left and see what happens.

The details of this section are postponed to Appendix \ref{app:XXXmag}. Here, we only write down the results.

We get the 1-magnon simply by application of \eqref{Bferm} to pseudovacuum \eqref{vacuum}
\begin{align}
B(\mu)\ket{0} &=  n(\mu)
\sum_{k=1}^{L} [\mu]^{k} \bar{\psi}_k \ket{0} \lb{1mag}
\intertext{where we use the concise notation}
[\mu] &=  \frac{\mu+1}{\mu},\quad\text{and}\quad n(\mu)=\frac{\mu^L}{\mu+1}. \lb{koef01}
\end{align}
The 2-magnon state is of the form
\begin{align} \lb{2mag}
B(\mu) B(\lambda) \ket{0} = n(\mu) n(\lambda)
 \sum\limits_{1\leq r<s\leq L} \Big( [\lambda]^{r} [\mu]^{s}  \frac{\lambda - \mu +1}{\lambda-\mu} +
 [\mu]^{r} [\lambda]^{s} \frac{\mu - \lambda + 1}{\mu-\lambda}
 \Big) \bar{\psi}_{r}\bar{\psi}_{s} \ket{0}
\end{align}
and the 3-magnon state is
\begin{align}
& B(\nu)B(\mu) B(\lambda) \ket{0} = n(\nu) n(\mu) n(\lambda) \times \notag\\
& \times  \sum_{1\leq q<r<s\leq L}  \sum_{\sigma\in S_3} \sigma \Bigl( [\nu]^q [\mu]^r[\lambda]^s
 \frac{\nu-\mu+1}{\nu-\mu}\cdot \frac{\nu-\lambda+1}{\nu-\lambda}\cdot \frac{\mu-\lambda+1}{\mu-\lambda}\Bigr) \bpsi_q \bpsi_r \bpsi_s \ket{0}. \lb{3mag}
\end{align}

From the results (\ref{1mag}), (\ref{2mag}) and \eqref{3mag} we can conjecture that the general $M$-magnon state is of the form
 \be
 \lb{Mmag1}
 \begin{array}{c}
 | \lambda_1 , \dots , \lambda_M \rangle \equiv B(\lambda_1) \cdots B(\lambda_M)\ket{0}
 =   \\ [0.3cm]
 = \Bigl( \prod\limits_{i=1}^M n(\lambda_i) \Bigr)
 \sum\limits_{1 \leq k_1 < ... <k_M \leq L} \;\;\;
 \sum\limits_{\sigma_\lambda \in S_M}  \sigma_\lambda \cdot \Bigl(
 \prod\limits_{i<j}^M \frac{\lambda_i - \lambda_j +1}{\lambda_i - \lambda_j}
 \prod\limits_{i=1}^M [\lambda_i]^{k_i} \Bigr) \bar{\psi}_{k_1} \cdots \bar{\psi}_{k_M} \ket{0} \equiv  \\ [0.3cm]
\equiv \Bigl( \prod\limits_{i=1}^M n(\lambda_i) \Bigr)
 \prod\limits_{i<j}^M \frac{1}{\lambda_i - \lambda_j}
 \sum\limits_{1 \leq k_1 < ... <k_M \leq L} \;\;\;
 \sum\limits_{\sigma_\lambda \in S_M}  (-1)^{p(\sigma_\lambda)}  \\ [0.3cm]
  \sigma_\lambda \cdot \Bigl(
 \prod\limits_{i<j}^M (\lambda_i - \lambda_j +1)
 \prod\limits_{i=1}^M [\lambda_i]^{k_i} \Bigr) \bar{\psi}_{k_1} \cdots \bar{\psi}_{k_M} \ket{0} \; ,
\end{array}
 \ee
 where $\sigma_\lambda$ is a permutation of the parameters $\{\lambda_1,\dots,\lambda_M \}$,
 $p(\sigma_\lambda)=0,1({\rm mod}2)$ is the parity of the permutation $\sigma_\lambda$ and
 $\sum\limits_{\sigma_\lambda \in S_M}$ is the sum over all such permutations. However, in the light of previous results this is no more a conjecture but a special representation of \eqref{Bethe2}.

\section{Fermionic realization of XXZ}
\setcounter{equation}0

Substituting \eqref{ferm1}-\eqref{ferm4} into (\ref{Rxxz}) gives a fermionic representation for the generators \eqref{Rxxz1} of the Hecke algebra
\begin{align}
\hat{R}^{(q)}_{k k+1}
= \bar{\psi}_{k+1} \, \psi_k + \bar{\psi}_{k} \, \psi_{k+1} - q \, \bar{\psi}_k \, \psi_k
- q^{-1} \, \bar{\psi}_{k+1} \, \psi_{k+1}
 + (q+q^{-1}) \bar{\psi}_k \, \psi_k \, \bar{\psi}_{k+1} \, \psi_{k+1} + q \mathbb{I}.\lb{Rferm-00}
 \end{align}
In the following, we will use the baxterized R-matrix \eqref{Rbaxt} multiplied by $\mu^{1/2}$ for a simpler formula, which is of the form
\begin{multline}\lb{Rbaxtferm}
    \hat{R}_{k,k+1}(\mu) = (1-\mu) \Bigl[ \bar{\psi}_{k+1}  \psi_k + \bar{\psi}_{k}  \psi_{k+1} - q  \bar{\psi}_k  \psi_k - q^{-1}  \bar{\psi}_{k+1}  \psi_{k+1} + \\
 + (q+q^{-1}) \bar{\psi}_k  \psi_k  \bar{\psi}_{k+1}  \psi_{k+1} \Bigr] + (q-\mu q^{-1}) \mathbb{I}.
\end{multline}

We repeat the construction used in section \ref{sec:Fermi} with the R-matrix of the form
\eqref{Rferm-00} instead of \eqref{hatR} and the Yang-Baxter equation \eqref{YBEbaxt} instead of \eqref{YBE} resp. \eqref{YBEbraid}.

We recall the monodromy matrix of the form \eqref{monod2}. Again, we write it in the form \eqref{monod3}
\be
    T_a(\mu) = \hat{R}_{a,1}(\mu)\cdot \underbrace{\hat{R}_{12}(\mu)\dotsm \hat{R}_{L-1,L}(\mu) P_{L-1,L} \dotsm P_{12}}_{X(\mu)} \cdot P_{a,1}.
\ee
The fermionic representation of $X(\mu)$ is obtained by \eqref{Pferm} and \eqref{Rbaxtferm}.

As we have seen, we need to express the monodromy matrix \eqref{monod3} as a matrix in the auxiliary space $V_a$. The generator of the Hecke algebra \eqref{Rxxz} is of the form
\be
    \hat{R}_{a,1}^{(q)} = \begin{pmatrix}
    q-q^{-1} N_1 & \bpsi_1 \\
    \psi_1 & q N_1
    \end{pmatrix}.
\ee
Then \eqref{Rbaxt} is
\begin{gather}
    \hat{R}_{a,1}(\mu) =(1-\mu)\hat{R}_{a,1}^{(q)} + \mu(q-q^{-1})\mathbb{I} = \notag\\ =\begin{pmatrix}
    (q-\mu q^{-1})\mathbb{I} -(1-\mu)q^{-1} N_1 & (1-\mu)\bpsi_1 \\
        (1-\mu)\psi_1 & (1-\mu) q N_1 + \mu(q - q^{-1})
    \end{pmatrix}.\lb{Rbaxtferm1}
\end{gather}
The form of $P_{a,1}$ is known from \eqref{Pferm}.

Using \eqref{Pferm} and \eqref{Rbaxtferm1} we get for the matrix elements of $T_a(\mu)$
\be
    T_a(\mu) = \hat{R}_{a,1}(\mu) X(\mu) P_{a,1} = \begin{pmatrix} A(\mu)&B(\mu)\\C(\mu)&D(\mu) \end{pmatrix}
\ee
that
\begin{align}
    A(\mu) &= \Bigl[(q-\mu q^{-1})\mathbb{I} -(1-\mu)q^{-1} N_1 \Bigr] X(\mu) (\mathbb{I}-N_1) +(1-\mu)\bpsi_1 X(\mu) \psi_1, \\
    B(\mu) &= \Bigl[ (q-\mu q^{-1})\mathbb{I} -(1-\mu)q^{-1} N_1 \Bigr] X(\mu) \bpsi_1 +(1 - \mu)\bpsi_1 X(\mu) N_1, \lb{BfermXXZ}\\
    C(\mu) &= (1-\mu)\psi_1 X(\mu) (\mathbb{I}-N_1) + \Bigl[(1 - \mu) q N_1 + \mu(q - q^{-1}) \Bigr] X(\mu) \psi_1, \\
        D(\mu) &= (1-\mu)\psi_1 X(\mu) \bpsi_1 + \Bigl[(1 - \mu) q N_1 + \mu(q - q^{-1}) \Bigr] X(\mu) N_1.
\end{align}

 \section{Bethe vectors for the homogeneous XXZ model}
 \setcounter{equation}0

 As in section \ref{sec:BetheXXX}, we are interested in the Bethe vectors \eqref{BetheVect}
 \be
    \ket{\lambda_1,\dotsc,\lambda_M} = B(\lambda_1)\dotsm B(\lambda_M) \ket{0}
 \ee
 with the opeartor $B(\mu)$ of the form \eqref{BfermXXZ}. The details are postponed to Appendix \ref{app:XXZmag}.

 For the 1-magnon we get
 \be\lb{1magXXZ}
    \ket{\mu} \equiv B(\mu) \ket{0} =  n_q(\mu) \; \sum_{k=1}^{L} \; \bigl[ \mu \bigr]_q^{\; k} \;  \bpsi_k\ket{0},
 \ee
 where we introduce
 \be \lb{koef03}
    \bigl[ \mu \bigr]_q=  \frac{q-\mu q^{-1}}{1-\mu},
 \ee
 and the normalization
 \be\lb{koef04}
    n_q(\mu) = \frac{(q-q^{-1})(1-\mu)^L}{q-\mu q^{-1}}.
 \ee

 The 2-magnon state is obtained in the following form:
 \be \lb{2magXXZ}
    \ket{\lambda,\mu} \equiv B(\lambda) B(\mu) = n_q(\lambda) n_q(\mu) \sum_{1\leq r<s\leq L} \Bigl\{ \frac{\lambda q^{-1} -\mu q}{\lambda-\mu} [\lambda]_q^{\, r} [\mu]_q^{\, s}
     + \frac{\mu q^{-1} -\lambda q}{\mu-\lambda} [\mu]_q^{\, r} [\lambda]_q^{\, s}  \Bigr\} \bpsi_r\bpsi_s\ket{0}.
 \ee

 We can see that the situation is very similar to that in section \ref{sec:BetheXXX}. Again, we propose that  the general $M$-magnon state possess the form
 \begin{gather}\lb{MmagXXZ}
    \ket{\lambda_1,\dotsc,\lambda_M}=\prod_{l=1}^M n_q(\lambda_l)
    \!\!\!\!\! \sum_{1\leq k_1<\dotsb<k_M\leq L}
    \;\;\; \sum_{\sigma_\lambda \in S_M} \sigma_\lambda \Bigl( \prod_{i<j}^M \frac{\lambda_i q^{-1}-\lambda_j q}{\lambda_i - \lambda_j} \prod_{i=1}^M [\lambda_i]^{k_i}_q \Bigr) \bpsi_{k_1}\dotsm\bpsi_{k_M}\ket{0}
 \end{gather}
 where $S_M$ is the symmetric group of order $M$ and $\sigma_\lambda\in S_M$ permutes the parameters  $\{ \lambda_1,\dotsc,\lambda_M\}$. Again, this is just a special representation of \eqref{Bethe2}. In the next section we prove formula (\ref{MmagXXZ}) by using the coordinate Bethe ansatz.


\section{Fermionic models and coordinate Bethe ansatz}
\setcounter{equation}0

In this Section we will use the coordinate Bethe ansatz method to construct
 Bethe vectors for the periodic chain models which are formulated in terms of free fermions.
 The coordinate Bethe ansatz method is named after the seminal work by Hans
Bethe \cite{Bethe}. Bethe found eigenfunctions and spectrum of the one-dimensional spin-1/2 isotropic
magnet (which we called above as XXX Heisenberg closed spin chain model).
 The review of the applications of the coordinate Bethe ansatz method can be found in the
  book \cite{Gaudin} (see also \cite{Low} and references therein).

\subsection{R--matrix, hamiltonian and a vacuum state}

Recall that the fermionic representation of
the Hecke algebra (\ref{rmatrR}) is based on the realization of the $R$-matrix in the form
$$
\wh{R}_{k,k+1}= \bar{\psi}_{k+1}\psi_k+\bar{\psi}_k\psi_{k+1}-
q\bar{\psi}_k\psi_k-q^{-1}\bar{\psi}_{k+1}\psi_{k+1}+
(q+q^{-1})\bar{\psi}_k\psi_k\bar{\psi}_{k+1}\psi_{k+1}+q  \; .
$$
Consider the hamiltonian for the periodic fermionic chain model
("small polaron model", see \cite{Korep} and references therein)
\begin{eqnarray*}
\MC{H}&=&{\tsum_{k=1}^{L-1}}\wh{R}_{k,k+1}+\wh{R}_{L,1}-qL=\\
&=&{\tsum_{k=1}^{L-1}}\Bigl(\bar{\psi}_{k+1}\psi_k+\bar{\psi}_k\psi_{k+1}-
q\bar{\psi}_k\psi_k-q^{-1}\bar{\psi}_{k+1}\psi_{k+1}+
(q+q^{-1})\bar{\psi}_k\psi_k\bar{\psi}_{k+1}\psi_{k+1}\Bigr)+\\
&&\hskip20mm+ \bar{\psi}_1\psi_L+\bar{\psi}_L\psi_1-
q\bar{\psi}_L\psi_L-q^{-1}\bar{\psi}_1\psi_1+
(q+q^{-1})\bar{\psi}_L\psi_L\bar{\psi}_1\psi_1=
\end{eqnarray*}
\be
\lb{Hferm}
\begin{array}{c}
={\tsum_{k=1}^{L-1}}\Bigl(\bar{\psi}_{k+1}\psi_k+\bar{\psi}_k\psi_{k+1}+
(q+q^{-1})\bar{\psi}_k\psi_k\bar{\psi}_{k+1}\psi_{k+1}\Bigr)+   \\
 + \bar{\psi}_1\psi_L+\bar{\psi}_L\psi_1+
(q+q^{-1})\bar{\psi}_L\psi_L\bar{\psi}_1\psi_1-
(q+q^{-1}){\tsum_{k=1}^L}\bar{\psi}_k\psi_k  \; .
\end{array}
\ee

This model is not coincident with the XXZ spin chain in view of the
representation of the matrix $\wh{R}_{L,1}$ given in (\ref{Rxxz})
in terms of fermions (\ref{fermions}). In the XXZ case the fermionic representation
of $\wh{R}_{L,1}$ is nonlocal.

The vacuum state $| 0 \rangle $ of the hamiltonian is defined by the equations
$\psi_k| 0 \rangle =0$ for $k=1,\,2,\,\ldots,\,L$.

\subsection{The 1-magnon states}

We look for the 1-magnon solution in the form
\begin{equation}
 \lb{1magn}
 |1 \rangle  = \sum_{n=1}^{L} c_n \, \bar{\psi}_{n} | 0 \rangle  \; .
\end{equation}
 Substitution of (\ref{Hferm}) and (\ref{1magn}) in the eigenvalue problem
 $\MC{H} |1 \rangle = E |1 \rangle$ gives the following
 equation for the coefficients $c_n$ (the 4-fermionic term in (\ref{Hferm}) does not contribute to
 the equations):
\begin{equation}
 \lb{solep2}
 c_{n-1} + c_{n+1} = (E+(q+q^{-1})) \, c_n \; , \;\;\; 1 \leq n \leq L \; ,
\end{equation}
 where $c_{n+L} = c_n$, i.e., $c_0 = c_L$ and $c_{L+1} = c_1$.
 Since equation (\ref{solep2}) is the discrete version of the ordinary
 differential equation of the second order with constant coefficients,
 one can solve (\ref{solep2}) if we insert $c_n = X^n$.
 As a result, we obtain the condition
 \begin{equation}
 \lb{Ex1}
E+(q+q^{-1})=X+X^{-1} \; ,
\end{equation}
 which is symmetric under the exchange $X \leftrightarrow X^{-1}$. Thus, the
 general solution of (\ref{solep2}) is
\begin{equation}
 \lb{Ex1a}
c_n=A_1X^n+A_2X^{-n} \; ,
\end{equation}
where arbitrary constants $A_1$, $A_2$ are independent of $n$. The boundary conditions
$c_k=c_{L+k}$ lead to the equation for $X$:
\begin{equation}
 \lb{solep5a}
  X^L =1  \; .
\end{equation}
However, in this case, we have $X^{-n}=X^{L-n}$, and linearly independent
solutions are
\begin{equation}
 \lb{solep7}
c_n=X^n \; ,\qquad\MR{where}\qquad X^L=1\,.
\end{equation}
Thus, to each solution $X=X_k$ of equation (\ref{solep5a})
 \begin{equation}
 \lb{M1a}
 X_k = \exp \Bigl( \frac{2\pi i k}{L} \Bigr) \;\;\;\;\; (k=0,\dots,L-1) \;
 \end{equation}
we have two one-magnon states (orthogonal to each other)
\begin{equation}
 \lb{M1}
|1\rangle_k ={\tsum_{n=1}^L}X^n_k \bar{\psi}_n|0\rangle \; , \;\;\;\;
|1\rangle'_k ={\tsum_{n=1}^L} X^{-n}_k \bar{\psi}_n|0\rangle
\end{equation}
with the same energy
 \be
 \lb{solep8}
 E = (q + q^{-1}) + (X_k + X_k^{-1})  \; .
 \ee
 On the other hand,
 we have $X_k^{-1}=X_{L-k}$
 and the set of vectors $|1\rangle_{L-k}$ coincides with the set of vectors  $|1\rangle'_k$.
  All these solutions  correspond to the spectrum of free fermions.

\subsection{The 2--magnon states}


We write $|n_1,n_2\rangle=\bar{\psi}_{n_1}\bar{\psi}_{n_2}| 0
\rangle $, where $1\leq n_1<n_2\leq L$. It is easy to find that
the action of the hamiltonian on the vector $
|2\rangle={\tsum_{1\leq n_1<n_2\leq L}}
c_{n_1,n_2}|n_1,n_2\rangle$ is
\begin{eqnarray*}
\MC{H}|2\rangle&=& {\tsum_{1\leq n_1<n_2\leq L}}
\Bigl((1-\delta_{n_1,1})c_{n_1-1,n_2}+
(1-\delta_{n_1+1,n_2})\bigl(c_{n_1,n_2-1}+c_{n_1+1,n_2}\bigr)+\\
&&\hskip10mm+ (1-\delta_{n_2,L})c_{n_1,n_2+1}-
\delta_{n_1,1}\bigl(1-\delta_{n_2,L}\bigr)c_{n_2,L}-
\bigl(1-\delta_{n_1,1}\bigr)\delta_{n_2,L}c_{1,n_1}+\\
&&\hskip20mm+
(q+q^{-1})\bigl(\delta_{n_1+1,n_2}+\delta_{n_1,1}\delta_{n_2,L}-2\bigr)
c_{n_1,n_2}\Bigr)|n_1,n_2\rangle\,.
\end{eqnarray*}
Equation $\MC{H}|2\rangle=\MC{E}|2\rangle$ is then equivalent to
the system of equations
$$
\begin{array}{l}
(1-\delta_{n_1,1})c_{n_1-1,n_2}+
(1-\delta_{n_1+1,n_2})\bigl(c_{n_1,n_2-1}+c_{n_1+1,n_2}\bigr)+
(1-\delta_{n_2,L})c_{n_1,n_2+1}-\\[4pt]
\hskip10mm- \delta_{n_1,1}\bigl(1-\delta_{n_2,L}\bigr)c_{n_2,L}-
\bigl(1-\delta_{n_1,1}\bigr)\delta_{n_2,L}c_{1,n_1}=\\
\hskip20mm= \Bigl(\MC{E}+2(q+q^{-1})- (q+q^{-1})
\bigl(\delta_{n_1+1,n_2}+\delta_{n_1,1}\delta_{n_2,L}\bigr)\Bigr)
c_{n_1,n_2}
\end{array}
$$
for any $1\leq n_1<n_2\leq L$.

The coordinate Bethe ansatz is based on the idea to write
$$
\MC{E}+2(q+q^{-1})=X_1+X_1^{-1}+X_2+X_2^{-1}
$$
and to find solution of the system in the form
$$
c_{n_1,n_2}=A_{12}X_1^{n_1}X_2^{n_2}+A_{21}X_2^{n_1}X_1^{n_2}\,,
$$
where $A_{12}$ and $A_{21}$ are independent of $n_1$ and $n_2$, but
they can depend on $X_1$ and $X_2$.

Substituting this assumption into the equation we obtain
$$
\begin{array}{l}
\delta_{n_1+1,n_2}\Bigl(A_{12}\bigl(X_1X_2-(q+q^{-1})X_2+1\bigr)+
A_{21}\bigl(X_1X_2-(q+q^{-1})X_1+1\bigr)\Bigr)(X_1X_2)^n_1+\\[4pt]
+
\bigl(A_{12}+X_1^LA_{21}\bigr)
\bigl(\delta_{n_1,1}X_2^{n_2}+\delta_{n_2,L}X_1X_2^{n_1}\bigr)+
\bigl(A_{21}+X_2^LA_{12}\bigr)
\bigl(\delta_{n_1,1}X_1^{n_2}+\delta_{n_2,L}X_1^{n_1}X_2\bigr)=\\[4pt]
=\delta_{n_1,1}\delta_{n_2,L}
\Bigl(\bigl((X_1X_2)^L+X_1X_2\bigr)\bigl(A_{12}+A_{21}\bigr)+
(q+q^{-1})\bigl(A_{12}X_1X_2^L+A_{21}X_1^LX_2\bigr)\Bigr)\,.
\end{array}
$$
To fulfill these equations we put
$$
\begin{array}{c}
A_{12}\bigl(X_1X_2-(q+q^{-1})X_2+1\bigr)+
A_{21}\bigl(X_1X_2-(q+q^{-1})X_1+1\bigr)=0\,,\\[4pt]
A_{12}+X_1^LA_{21}=0\,,\qquad A_{21}+X_2^LA_{12}=0\,,
\end{array}
$$
or equivalently
\begin{eqnarray*}
\frac{A_{21}}{A_{12}}&=&-\frac{X_1X_2-(q+q^{-1})X_2+1}{X_1X_2-(q+q^{-1})X_1+1}\,,\\
X_1^L&=&\frac{X_1X_2-(q+q^{-1})X_1+1}{X_1X_2-(q+q^{-1})X_2+1}\,,\\
X_2^L&=&\frac{X_1X_2-(q+q^{-1})X_2+1}{X_1X_2-(q+q^{-1})X_1+1}\,.
\end{eqnarray*}

\subsection{The 3--magnon states}

For $1\leq n_1<n_2<n_3\leq L$ we put
$|n_1,n_2,n_3\rangle=\bar{\psi}_{n_1}\bar{\psi}_{n_2}\bar{\psi}_{n_3}|
0 \rangle $. The action of the ha\-mil\-to\-nian $\MC{H}$ on a
vector $|3\rangle={\tsum_{1\leq n_1<n_2<n_3\leq
L}}c_{n_1,n_2,n_3}|n_1,n_2,n_3\rangle$ is
\begin{eqnarray*}
\MC{H}|3\rangle&=& {\tsum_{1\leq n_1<n_2<n_3\leq L}}\Bigl(
\bigl(1-\delta_{n_1,1}\bigr)c_{n_1-1,n_2,n_3}+
\bigl(1-\delta_{n_1+1,n_2}\bigr)
\bigl(c_{n_1,n_2-1,n_3}+c_{n_1+1,n_2,n_3}\bigr)+\\
&&+ \bigl(1-\delta_{n_2+1,n_3}\bigr)
\bigl(c_{n_1,n_2,n_3-1}+c_{n_1,n_2+1,n_3}\bigr)+
\bigl(1-\delta_{n_3,L}\bigr)c_{n_1,n_2,n_3+1}+\\
&&+ \delta_{n_1,1}\bigl(1-\delta_{n_3,L}\bigr) c_{n_2,n_3,L}+
\delta_{n_3,L}\bigl(1-\delta_{n_1,1}\bigr)c_{1,n_1,n_2}+\\
&&+ (q+q^{-1})\bigl(\delta_{n_1+1,n_2}+\delta_{n_2+1,n_3}+
\delta_{n_1,1}\delta_{n_3,L}-3\bigr)
c_{n_1,n_2,n_3}\Bigr)|n_1,n_2,n_3\rangle\,.
\end{eqnarray*}
Equation $\MC{H}|3\rangle=\MC{E}|3\rangle$ is equivalent to the
system of equation
$$
\begin{array}{l}
\bigl(1-\delta_{n_1,1}\bigr)c_{n_1-1,n_2,n_3}+
\bigl(1-\delta_{n_1+1,n_2}\bigr)\bigl(c_{n_1,n_2-1,n_3}+c_{n_1+1,n_2,n_3}\bigr)+\\[4pt]
\hskip5mm+
\bigl(1-\delta_{n_2+1,n_3}\bigr)\bigl(c_{n_1,n_2,n_3-1}+c_{n_1,n_2+1,n_3}\bigr)+
\bigl(1-\delta_{n_3,L}\bigr)c_{n_1,n_2,n_3+1}+\\[4pt]
\hskip10mm+
\delta_{n_1,1}\bigl(1-\delta_{n_3,L}\bigr)c_{n_2,n_3,L}+
\delta_{n_3,L}\bigl(1-\delta_{n_1,1}\bigr)c_{1,n_1,n_2}=\\[4pt]
\hskip20mm= \Bigl(\MC{E}+3(q+q^{-1})-
(q+q^{-1})\bigl(\delta_{n_1+1,n_2}+
\delta_{n_2+1,n_3}+\delta_{n_1,1}\delta_{n_3,L}\bigr)\Bigr)c_{n_1,n_2,n_3}\,,
\end{array}
$$
where $1\leq n_1<n_2<n_3\leq L$. When we put
$$
\MC{E}+3(q+q^{-1})=X_1+X_1^{-1}+X_2+X_2^{-1}+X_3+X_3^{-1}
$$
and look for  solution of $c_{n_1,n_2,n_3}$ in the form
$$
c_{n_1,n_2,n_3}={\tsum_{\sigma\in S_3}}A_\sigma
X_{\sigma(1)}^{n_1}X_{\sigma(2)}^{n_2}X_{\sigma(3)}^{n_3}\,,
$$
we obtain the following system of the equations:
$$
\begin{array}{l}
\delta_{n_1+1,n_2}{\tsum_{\sigma\in S_3}}A_{\sigma}
\bigl(X_{\sigma(1)}X_{\sigma(2)}-(q+q^{-1})X_{\sigma(2)}+1\bigr)\;
\bigl(X_{\sigma(1)}X_{\sigma(2)}\bigr)^{n_1}X_{\sigma(3)}^{n_3}+\\[6pt]
+\delta_{n_2+1,n_3}{\tsum_{\sigma\in S_3}}A_{\sigma}
\bigl(X_{\sigma(2)}X_{\sigma(3)}-(q+q^{-1})X_{\sigma(3)}+1\bigr)\;
X_{\sigma(1)}^{n_1}\bigl(X_{\sigma(2)}X_{\sigma(3)}\bigr)^{n_2}+\\[6pt]
+\delta_{n_1,1}{\tsum_{\sigma\in S_3}}A_{\sigma}
\bigl(X_{\sigma(2)}^{n_2}X_{\sigma(3)}^{n_3}-
X_{\sigma(1)}^{n_2}X_{\sigma(2)}^{n_3}X_{\sigma(3)}^L\bigr)+\\[6pt]
+\delta_{n_3,L}{\tsum_{\sigma\in S_3}}A_{\sigma}
\bigl(X_{\sigma(1)}^{n_1}X_{\sigma(2)}^{n_2}X_{\sigma(3)}^{L+1}-
X_{\sigma(1)}X_{\sigma(2)}^{n_1}X_{\sigma(3)}^{n_2}\bigr)+\\
+\delta_{n_1,1}\delta_{n_3,L}{\tsum_{\sigma\in S_3}}A_{\sigma}
\bigl(X_{\sigma(1)}^{n_2}X_{\sigma(2)}^LX_{\sigma(3)}^L+
X_{\sigma(1)}X_{\sigma(2)}X_{\sigma(3)}^{n_2}-
(q+q^{-1})X_{\sigma(1)}X_{\sigma(2)}^{n_2}X_{\sigma(3)}^L\bigr)=0\,.
\end{array}
$$
Let $\pi_1$ be the transposition $1\leftrightarrow2$ and $\pi_2$
the transposition $2\leftrightarrow3$. To cancel the terms at
$\delta_{n_1+1,n_2}$ and $\delta_{n_2+1,n_3}$, it is sufficient for any
$\sigma\in S_3$ to put
$$
\begin{array}{l}
A_{\sigma}\bigl(X_{\sigma(1)}X_{\sigma(2)}-(q+q^{-1})X_{\sigma(2)}+1\bigr)+
A_{\sigma\circ\pi_1}\bigl(X_{\sigma(1)}X_{\sigma(2)}-(q+q^{-1})X_{\sigma(1)}+1\bigr)=0\,,\\[4pt]
A_{\sigma}\bigl(X_{\sigma(2)}X_{\sigma(3)}-(q+q^{-1})X_{\sigma(3)}+1\bigr)+
A_{\sigma\circ\pi_2}\bigl(X_{\sigma(2)}X_{\sigma(3)}-(q+q^{-1})X_{\sigma(2)}+1\bigr)=0\,.
\end{array}
$$
If we consider the element $\epsilon\in S_3$ defined by the
relations $\epsilon(1)=3$, $\epsilon(2)=1$, $\epsilon(3)=2$, we
obtain
$$
\begin{array}{l}
{\tsum_{\sigma\in S_3}}A_{\sigma}
\bigl(X_{\sigma(2)}^{n_2}X_{\sigma(3)}^{n_3}-
X_{\sigma(1)}^{n_2}X_{\sigma(2)}^{n_3}X_{\sigma(3)}^L\bigr)=
{\tsum_{\sigma\in S_3}}
\bigl(A_{\sigma}X_{\sigma(2)}^{n_2}X_{\sigma(3)}^{n_3}-
A_{\sigma}X_{\sigma\circ\epsilon(2)}^{n_2}
X_{\sigma\circ\epsilon(3)}^{n_3}X_{\sigma\circ\epsilon(1)}^L\bigr)=\\[6pt]
\hskip20mm= {\tsum_{\sigma\in S_3}}\bigl(A_\sigma-
A_{\sigma\circ\epsilon^{-1}}X_{\sigma(1)}^L\bigr)
X_{\sigma(2)}^{n_2}X_{\sigma(3)}^{n_3}\,.
\end{array}
$$
So we put for any $\sigma\in S_3$
$$
A_\sigma-A_{\sigma\circ\epsilon^{-1}}X_{\sigma(1)}^L=0\,,\quad\MR{i.e.}\quad
A_{\sigma\circ\epsilon}=A_{\sigma}X_{\sigma(3)}^L\,.
$$
It is easy to show that these three assumptions solve the whole system
for $c_{n_1,n_2,n_3}$.

We obtained for the $A_\sigma$ conditions
\begin{equation}
 \label{3-1}
\begin{array}{rcl}
A_{\sigma\circ\pi_1}&=&
-\dfrac{X_{\sigma(1)}X_{\sigma(2)}-(q+q^{-1})X_{\sigma(2)}+1}
{X_{\sigma(1)}X_{\sigma(2)}-(q+q^{-1})X_{\sigma(1)}+1}\,A_{\sigma}\,,\\[9pt]
A_{\sigma\circ\pi_2}&=&
-\dfrac{X_{\sigma(2)}X_{\sigma(3)}-(q+q^{-1})X_{\sigma(3)}+1}
{X_{\sigma(2)}X_{\sigma(3)}-(q+q^{-1})X_{\sigma(2)}+1}\,A_{\sigma}\,.
\end{array}
\end{equation}
It follows from these relations that for any $\sigma\in S_3$
$$
A_{(\sigma\circ\pi_1)\circ\pi_1}=
-\frac{X_{\sigma\circ\pi_1(1)}X_{\sigma\circ\pi_1(2)}-(q+q^{-1})X_{\sigma\circ\pi_1(2)}+1}
{X_{\sigma\circ\pi_1(1)}X_{\sigma\circ\pi_1(2)}-(q+q^{-1})X_{\sigma\circ\pi_1(1)}+1}\,
A_{\sigma\circ\pi_1}=A_{\sigma}
$$
holds. Similarly, we can show that
$A_{(\sigma\circ\pi_2)\circ\pi_2}=A_{\sigma}$.

Moreover, we have
\begin{eqnarray*}
A_{((\sigma\circ\pi_1)\circ\pi_2)\circ\pi_1}&=&
-\frac{X_{\sigma(2)}X_{\sigma(3)}-(q+q^{-1})A_{\sigma(3)}+1}
{X_{\sigma(2)}X_{\sigma(3)}-(q+q^{-1})A_{\sigma(2)}+1}\,
A_{(\sigma\circ\pi_1)\circ\pi_2}=\\
&=&\frac{X_{\sigma(2)}X_{\sigma(3)}-(q+q^{-1})A_{\sigma(3)}+1}
{X_{\sigma(2)}X_{\sigma(3)}-(q+q^{-1})A_{\sigma(2)}+1}
\frac{X_{\sigma(1)}X_{\sigma(3)}-(q+q^{-1})A_{\sigma(3)}+1}
{X_{\sigma(1)}X_{\sigma(3)}-(q+q^{-1})A_{\sigma(1)}+1}\,
A_{\sigma\circ\pi_1}=\\
&=&-\frac{X_{\sigma(2)}X_{\sigma(3)}-(q+q^{-1})A_{\sigma(3)}+1}
{X_{\sigma(2)}X_{\sigma(3)}-(q+q^{-1})A_{\sigma(2)}+1}\,
\frac{X_{\sigma(1)}X_{\sigma(3)}-(q+q^{-1})A_{\sigma(3)}+1}
{X_{\sigma(1)}X_{\sigma(3)}-(q+q^{-1})A_{\sigma(1)}+1}\times\\
&&\hskip30mm\times
\frac{X_{\sigma(1)}X_{\sigma(2)}-(q+q^{-1})X_{\sigma(2)}+1}
{X_{\sigma(1)}X_{\sigma(2)}-(q+q^{-1})X_{\sigma(1)}+1}\,A_{\sigma}\,,\\
A_{((\sigma\circ\pi_2)\circ\pi_1)\circ\pi_2}&=&
-\frac{X_{\sigma(1)}X_{\sigma(2)}-(q+q^{-1})A_{\sigma(2)}+1}
{X_{\sigma(1)}X_{\sigma(2)}-(q+q^{-1})A_{\sigma(1)}+1}\,
A_{(\sigma\circ\pi_2)\circ\pi_1}=\\
&=& \frac{X_{\sigma(1)}X_{\sigma(2)}-(q+q^{-1})A_{\sigma(2)}+1}
{X_{\sigma(1)}X_{\sigma(2)}-(q+q^{-1})A_{\sigma(1)}+1}
\frac{X_{\sigma(1)}X_{\sigma(3)}-(q+q^{-1})A_{\sigma(3)}+1}
{X_{\sigma(1)}X_{\sigma(3)}-(q+q^{-1})A_{\sigma(1)}+1}\,
A_{\sigma\circ\pi_2}=\\
&=& -\frac{X_{\sigma(1)}X_{\sigma(2)}-(q+q^{-1})A_{\sigma(2)}+1}
{X_{\sigma(1)}X_{\sigma(2)}-(q+q^{-1})A_{\sigma(1)}+1}\,
\frac{X_{\sigma(1)}X_{\sigma(3)}-(q+q^{-1})A_{\sigma(3)}+1}
{X_{\sigma(1)}X_{\sigma(3)}-(q+q^{-1})A_{\sigma(1)}+1}\times\\
&&\hskip30mm\times
\frac{X_{\sigma(2)}X_{\sigma(3)}-(q+q^{-1})X_{\sigma(3)}+1}
{X_{\sigma(2)}X_{\sigma(3)}-(q+q^{-1})X_{\sigma(2)}+1}\,A_{\sigma}\,.
\end{eqnarray*}
So for any $\sigma\in S_3$ the relation
$A_{((\sigma\circ\pi_1)\circ\pi_2)\circ\pi_1}=
A_{((\sigma\circ\pi_2)\circ\pi_1)\circ\pi_2}$ holds. Therefore,
$A_{\sigma}$ is really a function of the symmetric group $S_3$.

Since $\epsilon=\pi_2\circ\pi_1$, the equality
$A_{\sigma\circ\epsilon}=X_{\sigma(3)}^LA_\sigma$ leads to
the relation
\begin{eqnarray*}
A_{\sigma\circ\epsilon}&=&A_{(\sigma\circ\pi_2)\circ\pi_1}=
-\frac{X_{\sigma(1)}X_{\sigma(3)}-(q+q^{-1})X_{\sigma(3)}+1}
{X_{\sigma(1)}X_{\sigma(3)}-(q+q^{-1})X_{\sigma(1)}+1}\,A_{\sigma\circ\pi_2}=\\
&=& \frac{X_{\sigma(1)}X_{\sigma(3)}-(q+q^{-1})X_{\sigma(3)}+1}
{X_{\sigma(1)}X_{\sigma(3)}-(q+q^{-1})X_{\sigma(1)}+1}\,
\frac{X_{\sigma(2)}X_{\sigma(3)}-(q+q^{-1})X_{\sigma(3)}+1}
{X_{\sigma(2)}X_{\sigma(3)}-(q+q^{-1})X_{\sigma(2)}+1}\,A_{\sigma}=
X_{\sigma(3)}^LA_{\sigma}\,.
\end{eqnarray*}
So for any $\sigma\in S_3$ the relation
$$
X_{\sigma(3)}^L=
\frac{X_{\sigma(1)}X_{\sigma(3)}-(q+q^{-1})X_{\sigma(3)}+1}
{X_{\sigma(1)}X_{\sigma(3)}-(q+q^{-1})X_{\sigma(1)}+1}\,
\frac{X_{\sigma(2)}X_{\sigma(3)}-(q+q^{-1})X_{\sigma(3)}+1}
{X_{\sigma(2)}X_{\sigma(3)}-(q+q^{-1})X_{\sigma(2)}+1}
$$
has to hold. It is possible  to rewrite these relations in the form
\begin{equation}
 \label{3-4}
X_i^L=\prod_{k\neq i}
\frac{X_iX_k-(q+q^{-1})X_i+1}{X_iX_k-(q+q^{-1})X_k+1}\,.
\end{equation}

\subsection{The $M$--magnon states}

For $1\leq n_1<n_2<\ldots<n_{M-1}<n_M\leq L$ we denote
$$
|\vec{n}\rangle=|n_1,n_2,\ldots,n_M\rangle=
\bar{\psi}_{n_1}\bar{\psi}_{n_2}\ldots\bar{\psi}_{n_M}| 0 \rangle
$$
and take the vector
$$
|M\rangle={\tsum_{\vec{n}}}c_{\vec{n}}|\vec{n}\rangle=
{\tsum_{1\leq n_1<n_2<\ldots<n_M\leq L}}
c_{n_1,\ldots,n_M}|n_1,\ldots,n_M\rangle\,.
$$
When we write
$$
(\vec{n}\pm\vec{e}_k)=(n_1,\ldots,n_{k-1},n_k\pm1,n_{k+1},\ldots,n_M)\,,
$$
it is possible to show that
\begin{eqnarray*}
\MC{H}|M\rangle&=& \sum_{\vec{n}}\Bigl(
(1-\delta_{n_1,1})c_{\vec{n}-\vec{e}_1}+
{\tsum_{k=1}^{M-1}}(1-\delta_{n_k+1,n_{k+1}})
\bigl(c_{\vec{n}+\vec{e}_k}+c_{\vec{n}-\vec{e}_{k+1}}\bigr)+
(1-\delta_{n_M,L})c_{\vec{n}+\vec{e}_M}+\\
&&+
(-1)^{M-1}\delta_{n_1,1}(1-\delta_{n_M,L})c_{n_2,\ldots,n_M,L}+
(-1)^{M-1}\delta_{n_M,L}(1-\delta_{n_1,1})c_{1,n_1,\ldots,n_{M-1}}+\\
&&+
(q+q^{-1}){\tsum_{k=1}^{M-1}}\delta_{n_k+1,n_{k+1}}c_{\vec{n}}+
(q+q^{-1})\delta_{n_1,1}\delta_{n_M,L}c_{\vec{n}}-
M(q+q^{-1})c_{\vec{n}}\Bigr)|\vec{n}\rangle\,.
\end{eqnarray*}
Equation $\MC{H}|M\rangle=\MC{E}|M\rangle$ is then equivalent to
the system
$$
\begin{array}{l}
(1-\delta_{n_1,1})c_{\vec{n}-\vec{e}_1}+
{\tsum_{k=1}^{M-1}}(1-\delta_{n_k+1,n_{k+1}})
\bigl(c_{\vec{n}+\vec{e}_k}+c_{\vec{n}-\vec{e}_{k+1}}\bigr)+
(1-\delta_{n_M,L})c_{\vec{n}+\vec{e}_M}+\\[6pt]
\hskip10mm+
(-1)^{M-1}\delta_{n_1,1}(1-\delta_{n_M,L})c_{n_2,\ldots,n_M,L}+
(-1)^{M-1}\delta_{n_M,L}(1-\delta_{n_1,1})c_{1,n_1,\ldots,n_{M-1}}=\\[6pt]
\hskip20mm= \Bigl(\MC{E}+M(q+q^{-1})-
(q+q^{-1}){\tsum_{k=1}^{M-1}}\delta_{n_k+1,n_{k+1}}-
(q+q^{-1})\delta_{n_1,1}\delta_{n_M,L}\Bigr)c_{\vec{n}}\,.
\end{array}
$$
When we write the eigenvalue of the hamiltonian as
\begin{equation}
 \label{N-E}
\MC{E}={\tsum_{k=1}^M}\bigl(X_k+X_k^{-1}\bigr)-M(q+q^{-1}),
\end{equation}
look for the solution in the form
$$
c_{\vec{n}}={\tsum_{\sigma\in S_M}}A_{\sigma}
X_{\sigma(1)}^{n_1}X_{\sigma(2)}^{n_2}\ldots
X_{\sigma{M}}^{n_M}\, ,
$$
and substitute these assumptions into the system, we obtain
$$
\begin{array}{l}
{\tsum_{k=1}^{M-1}}\,\delta_{n_k+1,n_{k+1}}{\tsum_{\sigma\in S_M}}
A_{\sigma}\bigl(1+X_{\sigma(k)}X_{\sigma(k+1)}-
(q+q^{-1})X_{\sigma(k+1)}\bigr)\times\\[4pt]
\hskip70mm\times
X_{\sigma(1)}^{n_1}\ldots\bigl(X_{\sigma(k)}X_{\sigma(k+1)}\bigr)^{n_k}
\ldots X_{\sigma(M)}^{n_M}+\\[6pt]
+ \delta_{n_1,1}{\tsum_{\sigma\in S_M}}
A_{\sigma}\bigl(X_{\sigma(2)}^{n_2}\ldots X_{\sigma(M)}^{n_M}+
(-1)^MX_{\sigma(1)}^{n_2}X_{\sigma(2)}^{n_3}\ldots
X_{\sigma(M-1)}^{n_M}X_{\sigma(M)}^L\bigr)+\\[6pt]
+ \delta_{n_M,L}{\tsum_{\sigma\in S_M}}
A_{\sigma}\bigl(X_{\sigma(1)}^{n_1}\ldots
X_{\sigma(M-1)}^{n_M-1}X_{\sigma(M)}^{L+1}+
(-1)^MX_{\sigma(1)}X_{\sigma(2)}^{n_1}X_{\sigma(3)}^{n_2}\ldots
X_{\sigma(M)}^{n_{M-1}}\bigr)-\\[6pt]
-(-1)^M\delta_{n_1,1}\delta_{n_M,L}{\tsum_{\sigma\in S_M}}
A_{\sigma}\Bigl(X_{\sigma(1)}^{n_2}X_{\sigma(2)}^{n_3}\ldots
X_{\sigma(M-1)}^LX_{\sigma(M)}^L+
X_{\sigma(1)}X_{\sigma(2)}X_{\sigma(3)}^{n_2}\ldots
X_{\sigma(M)}^{n_{M-1}}+\\[4pt]
\hskip50mm+ (-1)^M(q+q^{-1})X_{\sigma(1)}X_{\sigma(2)}^{n_2}\ldots
X_{\sigma(M-1)}^{n_{M-1}}X_{\sigma(M)}^L\Bigr)=0\,.
\end{array}
$$
Let $\pi_k$, $k=1,\,\ldots,\,M-1$, be transpositions
$k\leftrightarrow k+1$. When the relation
\begin{equation}
 \label{N-1a}
\bigl(X_{\sigma(k)}X_{\sigma(k+1)}-(q+q^{-1})X_{\sigma(k+1)}+1\bigr)
A_{\sigma}+
\bigl(X_{\sigma(k)}X_{\sigma(k+1)}-(q+q^{-1})X_{\sigma(k)}+1\bigr)
A_{\sigma\circ\pi_k}=0\,,
\end{equation}
is true for any $\sigma\in S_M$ and $k=1,\,\ldots,\,M-1$, the
terms at $\delta_{n_k+1,n_{k+1}}$ vanish.

Let $\epsilon\in S_M$ be defined by the relations $\epsilon(k)=k-1$
for $k=2,\,\ldots,\,M$ and $\epsilon(1)=M$. If we require
\begin{equation}
 \label{N-2a}
A_{\sigma}+(-1)^MX_{\sigma(1)}^LA_{\sigma\circ\epsilon^{-1}}=0\,,
\qquad\MR{i.e.}\qquad
A_{\sigma\circ\epsilon}=(-1)^{M-1}X_{\sigma(M)}^LA_{\sigma}
\end{equation}
for any $\sigma\in S_M$, the terms at $\delta_{n_1,1}$
and $\delta_{n_M,L}$ are annulled.

Combining (\ref{N-1a}) and (\ref{N-2a}) we get
$$
\begin{array}{l}
{\tsum_{\sigma\in S_M}}A_{\sigma}
\Bigl(X_{\sigma(1)}^{n_2}X_{\sigma(2)}^{n_3}\ldots
X_{\sigma(M-1)}^LX_{\sigma(M)}^L+
X_{\sigma(1)}X_{\sigma(2)}X_{\sigma(3)}^{n_2}\ldots
X_{\sigma(M)}^{n_{M-1}}+\\[6pt]
\hskip30mm+ (-1)^M(q+q^{-1})X_{\sigma(1)}X_{\sigma(2)}^{n_2}\ldots
X_{\sigma(M-1)}^{n_{M-1}}X_{\sigma(M)}^L\Bigr)=\\[6pt]
\hskip10mm= {\tsum_{\sigma\in S_M}}
A_{\sigma}\Bigl(1+X_{\sigma(1)}X_{\sigma(2)}-(q+q^{-1})X_{\sigma(2)}\Bigr)
X_{\sigma(3)}^{n_2}X_{\sigma(4)}^{n_3}\ldots
X_{\sigma(M)}^{n_{M-1}}=0\,.
\end{array}
$$
Therefore, the assumptions (\ref{N-1a}) and (\ref{N-2a}) solve the
system for $c_{\vec{n}}$.

We rewrite relation (\ref{N-1a}) as
\begin{equation}
 \label{N-1}
A_{\sigma\circ\pi_k}=
-\frac{X_{\sigma(k)}X_{\sigma(k+1)}-(q+q^{-1})X_{\sigma(k+1)}+1}
{X_{\sigma(k)}X_{\sigma(k+1)}-(q+q^{-1})X_{\sigma(k)}+1}\,
A_{\sigma}\,.
\end{equation}
From this relation it is easy to show that for any $\sigma\in S_M$
and $k=1,\,\ldots,\,M-1$ the relations
$$
A_{(\sigma\circ\pi_k)\circ\pi_k}=A_{\sigma}\,,\qquad
A_{((\sigma\circ\pi_k)\circ\pi_{k+1})\circ\pi_k}=
A_{((\sigma\circ\pi_{k+1})\circ\pi_k)\circ\pi_{k+1}}\,.
$$
are valid. Therefore, $A_{\sigma}$ is really a function on symmetry
group $S_M$.

If we write
$\epsilon=\pi_{M-1}\circ\pi_{M-2}\circ\ldots\circ\pi_2\circ\pi_1$
and use (\ref{N-1}), it is possible to rewrite (\ref{N-2a}) as
\begin{eqnarray*}
A_{\sigma\circ\epsilon}&=&
-\frac{X_{\sigma(1)}X_{\sigma(M)}-(q+q^{-1})X_{\sigma(M)}+1}
{X_{\sigma(1)}X_{\sigma(M)}-(q+q^{-1})X_{\sigma(1)}+1}\,
A_{\sigma\circ\pi_{M-1}\circ\ldots\circ\pi_2}=\\
&=& (-1)^2\,
\frac{X_{\sigma(1)}X_{\sigma(M)}-(q+q^{-1})X_{\sigma(M)}+1}
{X_{\sigma(1)}X_{\sigma(M)}-(q+q^{-1})X_{\sigma(1)}+1}\,
\frac{X_{\sigma(2)}X_{\sigma(M)}-(q+q^{-1})X_{\sigma(M)}+1}
{X_{\sigma(2)}X_{\sigma(M)}-(q+q^{-1})X_{\sigma(2)}+1}\times\\
&&\hskip90mm\times
A_{\sigma\circ\pi_{M-1}\circ\ldots\circ\pi_3}=\\
&=& (-1)^{M-1}\prod_{k=1}^{M-1}
\frac{X_{\sigma(k)}X_{\sigma(M)}-(q+q^{-1})X_{\sigma(M)}+1}
{X_{\sigma(k)}X_{\sigma(M)}-(q+q^{-1})X_{\sigma(k)}+1}\,
A_{\sigma}= (-1)^{M-1}X_{\sigma(M)}^L\,A_{\sigma}\,.
\end{eqnarray*}
This implies that for any $i=1,\,2,\,\ldots,\,M$ the relation
\begin{equation}
 \label{N-2}
X_i^L=\prod_{k\neq i}
\frac{X_iX_k-(q+q^{-1})X_i+1}{X_iX_k-(q+q^{-1})X_k+1}\,.
\end{equation}
has to be true.

\subsection{Comparison with the standard XXZ model}

In the standard $XXZ$ model the eigenvalues of the hamiltonian are also
given by relation (\ref{N-E}). Moreover, relations (\ref{N-1})
are also of the same form, i.e., the relations
$$
A_{\sigma\circ\pi_k}=
-\frac{X_{\sigma(k)}X_{\sigma(k+1)}-(q+q^{-1})X_{\sigma(k+1)}+1}
{X_{\sigma(k)}X_{\sigma(k+1)}-(q+q^{-1})X_{\sigma(k)}+1}\,
A_{\sigma}
$$
are valid also for the XXZ model.

However, there is one important difference. In the relation corresponding to (\ref{N-2a}) the multiplier $(-1)^{M-1}$ is missing, i.e.
for XXZ, the relation
$$
A_{\sigma\circ\epsilon}=X_{\sigma(M)}^LA_{\sigma}\,.
$$
is valid. Therefore, we obtain in the XXZ model the relation
 \be
 \lb{N-3}
X_i^L=(-1)^{M-1}\prod_{k\neq i}
\frac{X_iX_k-(q+q^{-1})X_i+1}{X_iX_k-(q+q^{-1})X_k+1}\,
 \ee
 instead of
(\ref{N-2}). Comparing (\ref{N-2}) and (\ref{N-3}), we conclude that
 the spectrum of the fermion (soft polaron) and the standard $XXZ$ model are the
same for odd $M$, but if $M$ is even the spectrum of these models
can be different.

 The Bethe equations (\ref{N-3}) are equivalent to the Bethe equations (\ref{BAExxz})
 if we substitute
 $$
 X_k = \frac{q - q^{-1} \, \lambda_k}{1- \lambda_k}  \; .
 $$
 For this substitution, the right-hand side of (\ref{N-3}) is simplified and
 coincides with the right hand-side of (\ref{BAExxz}).


\section{Some remarks on the open Hecke chain}
\setcounter{equation}0

Let $H_n(q)$ be the Hecke algebra generated by the invertible
elements $T_k$ $(k=1,\dots,n-1)$ subject to relations (\ref{braidg}) and (\ref{ahecke}).
For future convenience, instead of (\ref{intrH2}), we will consider the following form of the Hamiltonian:
\be
\lb{free}
{\cal H}_n   = \sum_{k=1}^{n-1} T_k - \frac{(n-1)}{2} (q-q^{-1}) = \sum_{k=1}^{n-1} s_k
  \; \in \; H_{n}(q)  \; ,
\ee
where we have introduced the new generators of the $A$-type Hecke algebra $H_n(q)$
\be
\lb{new1}
s_k = T_k - \frac{q-q^{-1}}{2} = \frac{i}{2} \; T_k(x)|_{x^{1/2}=i} \; ,
\ee
and $T_k(x)|_{x^{1/2}=i}$ are the baxterized elements (\ref{baxtH}) taken at the point $x^{1/2}=i$.

\vspace{0.3cm}
\noindent
{\bf Remark 1.} The representation theory of the Hecke algebras $H_{n}(q)$ is well known.
For the details of this representation theory see, e.g.,
\cite{Jon1,Wenzl,Mur,Cher5,Isa1,IsOgH,OgPya,IsMolOs,GD,LLT} and references therein.
Each irreducible representation (irrep) of the Hecke algebra $H_{n}(q)$ ($q$ is a generic parameter)
corresponds to the
Young diagram $\Lambda$ with
$n$ nodes. The dimension of the irrep $\Lambda$ is given by the hook formula
(see, e.g., \cite{BSag} and \cite{OgPya})
\be
\lb{dime}
{\rm dim}(\Lambda) = \frac{n!}{\prod_{\alpha \in \Lambda} h_\alpha} \;\; ,
\ee
where $h_\alpha$ is a hook length of the nod $\alpha \in \Lambda$.
Recall, that the Young diagram $\Lambda$
with $m$ rows of the lengths $(\lambda_1,\lambda_2,\dots,\lambda_m)$
$$
\lambda_1 \geq \lambda_2 \geq \cdots \geq \lambda_m \; , \;\;\;\; \sum_{k=1}^m \lambda_k =n \; ,
$$
is called dual to the diagram $\Lambda'$ if $(\lambda_1,\lambda_2,\dots,\lambda_m)$
are the lengths of the columns of $\Lambda'$. It is clear that ${\rm dim}(\Lambda) = {\rm dim}(\Lambda')$.

\vspace{0.3cm}
\noindent
{\bf Remark 2} (see, e.g., \cite{Isa07}). Consider the element $j$ of the braid group ${\cal B}_{n}$
 \be
\lb{longj}
j := (T_1  \cdots T_{n-1}) (T_1 \cdots T_{n-2}) \cdots
(T_1 T_2) \cdot T_1 \; ,
 \ee
such that $T_i \cdot j = j \cdot T_{n-i}$ $(\forall \; i)$. It means that
$j$ commutes with the element
$(\sum\limits_{i=1}^{n-1} T_i)$ of the group algebra of ${\cal B}_{n}$.
Thus, for the Hecke quotient $H_{n}(q)$
of group algebra of ${\cal B}_{n}$,
the longest element $j \in H_{n}(q)$ commutes with
the Hamiltonian ${\cal H}_{n}$
(\ref{free}) and is a conservation charge
for the model of the open Hecke chain. 

\vspace{0.3cm}

The quantum integrable systems with the Hamiltonians (\ref{free}) were considered in
\cite{Isa07}, \cite{IsKir}. In next Subsection, we list the
characteristic identities for the Hamiltonians ${\cal H}_n$ in the cases
$n=2,\dots,6$. These identities define the whole energy spectrum of the
Hecke chains of the length $n=2,\dots,6$.

\subsection{Characteristic identities for ${\cal H}_n$}

Here, we use the notation
$$
\lambda=q-q^{-1} \; , \;\;\; \bar{q}= q+q^{-1} \; , \;\;\; v = \frac{1}{2} (q+q^{-1}) \; .
$$

\noindent{\bf 1. The case $n=2$.} The characteristic identity for the Hamiltonian
${\cal H}_2 = T_1 - \frac{\lambda}{2}$ is
$$
\begin{array}{c}
({\cal H}_2 - \frac{1}{2} \overline{q})({\cal H}_2 + \frac{1}{2} \overline{q}) = 0 \; .
\end{array}
$$
Two eigenvalues $ \frac{1}{2} \overline{q}$ and $ -\frac{1}{2} \overline{q}$
correspond to the 1-dimensional irreps $T_1=q$ and $T_1=-q^{-1}$ labeled, respectively,  by the Young diagrams
$(2)$ and $(1^2)$.

\noindent{\bf 2. The case $n=3$.} Here, we have the set of commuting elements \cite{Isa07}
\be
\lb{bethe6}
j_1=T_1 + T_2 \; , \;\;\; j_2=T_{1} T_{2} + T_{2} T_{1} \; , \;\;\;
j_3=T_{2}\, T_{1} \, T_{2} + T_1 + T_2 \; .
\ee
Note that the elements $j_2,j_3$ are expressed in terms of $j_1$:
$$
j_2= j_1^2 -\lambda j_1 -2 \; , \;\;\; j_3 = \frac{1}{2} (j_1^3-2 \lambda j_1^2 +(\lambda^2-1) j_1
+ 2 \lambda) \; .
$$
The element ${\cal H}_3=j_1-\lambda$ is the Hamiltonian (\ref{free}) for the open Hecke chain and
$j_3$ is a central element in $H_3$.
The characteristic identity for the Hamiltonian ${\cal H}_3$ is:
\be
\lb{char3}
({\cal H}_3 + \overline{q} )({\cal H}_3 - \overline{q} )({\cal H}_3  -1)({\cal H}_3 +1)=0 \; .
\ee
This means that ${\rm Spec}({\cal H}_3) = \{ \pm\overline{q}, \pm 1 \}$. The first two
eigenvalues $\pm \overline{q}$ correspond to the one dimensional representations
$T_i = \pm q^{\pm 1}$ $(i=1,2)$ of $H_3(q)$, which are related to the Young diagrams
$(3)$, $(1^3)$. The eigenvalues
$(\pm 1)$ correspond to the 2-dimensional irrep
$(2,1)$ of $H_3(q)$. \\

\noindent{\bf 3. The case $n=4$.} In this case we have the following set of commuting
elements
$$
\begin{array}{c}
j_1 = \sum_{i=1}^3 T_i  \; , \;\;\;
j_2 = \{ T_1, \, T_2  \}_{+} +
\{ T_2, \, T_3 \}_{+} + 2 T_3 T_1 \; , \\ \\
j_3 =   \{ T_1 T_3 , \, T_2  \}_{+} +
 (T_1 + T_3) T_2  (T_1 + T_3) + \lambda T_3 T_1
+ 2 \sum_{i=1}^3 T_i  \; , \\ \\
j_4 = \{ T_2 T_3  T_2 , \, T_1 \}_{+} +
\{ T_2 T_1  T_2 , \, T_3 \}_{+} +
\{ T_2, \, T_3 \}_{+} +  \{ T_2 , \, T_1 \}_{+}  \; , \\ \\
j_5 = T_1 T_2 T_3  T_2 T_1
+ T_2 T_3  T_2 + T_1 T_2  T_1 + T_1 + T_2 + T_3 \; , \\ \\
\end{array}
$$
The element $j_5$ is a central element in $H_{4}(q)$.
The longest element $j= T_1 T_2 T_3  T_1 T_2 T_1$ in $H_4(q)$ (cf. (\ref{longj})) is 
$j= (j_5-j_1)(j_1-\lambda)-j_4$ and therefore 
commutes with the Hamiltonian ${\cal H}_{4}=j_1-\frac{3}{2}\lambda$ (\ref{free}). This Hamiltonian
satisfies the characteristic identity
\be
\lb{char4}
\begin{array}{c}
({\cal H}_4 + \frac{3}{2}\, \overline{q})\cdot ({\cal H}_4 - \frac{3}{2}\, \overline{q}) \cdot
({\cal H}_4 + \frac{1}{2}\, \overline{q}) \left( ({\cal H}_4  + \frac{1}{2}\, \overline{q})^2 - 2 \right)
 \cdot \\ [0.2cm]
\cdot ({\cal H}_4 - \frac{1}{2}\, \overline{q})
\left( ({\cal H}_4 -  \frac{1}{2}\, \overline{q})^2 -2 \right)
\cdot \left( {\cal H}_4^2 -\frac{1}{4}\overline{q}^{2}-2 \right)=0 \; .
\end{array}
\ee
Thus, the spectrum of ${\cal H}_{4}$ consists of the eigenvalues: \\
$(\frac{3}{2}\, \overline{q}, \, -\frac{3}{2}\, \overline{q})$ for the two dual 1-dim. irreps
$(4)$, $(1^4)$ of the Hecke algebra $H_4(q)$; \\
$(\frac{1}{2}\, \overline{q}, \frac{1}{2}\, \overline{q} \pm \sqrt{2})$ and
$(-\frac{1}{2}\, \overline{q}, -\frac{1}{2}\, \overline{q} \pm \sqrt{2})$ for the two dual 3-dim. irreps
$(3,1)$, $(2,1^2)$; \\
$(\pm \frac{1}{2} \sqrt{\overline{q}{}^2+8})$ for the 2-dim. irrep
$(2^2)$ of $H_4(q)$. \\

\noindent{\bf 4. The case $n=5$.} For the Hamiltonian ${\cal H}_{5}=\sum\limits_{i=1}^4
T_i - 2\lambda$ the characteristic polynomial is an odd function of order $25$,
and the characteristic identity is
\be
\lb{char5a}
\begin{array}{lc}
(1^5)\cdot (5): & ({\cal H}_5 + 2\, \overline{q})\cdot ({\cal H}_5 - 2 \, \overline{q}) \cdot \\ [0.2cm]
(3,1^2): & {\cal H}_5 \, ({\cal H}_5^2 -1)({\cal H}_5^2 -5)
 \cdot \\ [0.2cm]
(2,1^3): & \left( ({\cal H}_5+\overline{q})^2  - \frac{(\sqrt{5}-1)^2}{4} \right)
\left( ({\cal H}_5+\overline{q})^2  - \frac{(\sqrt{5}+1)^2}{4} \right) \cdot \\ [0.2cm]
(4,1): & \left( ({\cal H}_5-\overline{q})^2  - \frac{(\sqrt{5}-1)^2}{4} \right)
\left( ({\cal H}_5-\overline{q})^2  - \frac{(\sqrt{5}+1)^2}{4} \right) \cdot \\ [0.2cm]
 (2^2,1): & ({\cal H}_5^2 + \overline{q}{\cal H}_5-1)
 \left( {\cal H}_5^3  + \overline{q}  {\cal H}_5^2 -5 {\cal H}_5 -2 \overline{q} \right) \cdot \\ [0.2cm]
(3,2): & ({\cal H}_5^2 - \overline{q}{\cal H}_5-1)
\left( {\cal H}_5^3  -  \overline{q}  {\cal H}_5^2 -5 {\cal H}_5 +2 \overline{q} \right)
  =0 \; .
\end{array}
\ee
The last two lines give the eigenvalues of ${\cal H}_5$, which correspond to the 5 dimensional
representations labeled by two dual Young diagrams $(2^2,1)$, $(3,2)$. The sum of the dimensions
of the irreps for $H_5(q)$ is equal to 26. We obtain the 25-th order of the characteristic identity since
the eigenvalue $0$ has multiplicity 2. This eigenvalue appears in the self-dual irrep $(3,1^2)$.

\noindent
{\bf 5. The case $n=6$.} For the Hamiltonian ${\cal H}_{6}=\sum\limits_{i=1}^5
T_i - \frac{5}{2}\lambda$ the characteristic polynomial is an even function of
${\cal H}_{6}$ of order $72$. The characteristic polynomial is much more complicated:
\be
\lb{char6}
\begin{array}{lc}
(1^6)\cdot(6): &
({\cal H}_6 + \frac{5}{2}\, \overline{q})\cdot ({\cal H}_6 - \frac{5}{2} \, \overline{q}) \cdot  \\ [0.2cm]
(2,1^4): & ({\cal H}_6+\frac{3}{2}\overline{q}) \left( ({\cal H}_6+\frac{3}{2}\overline{q})^2  - 1 \right)
\left( ({\cal H}_6+\frac{3}{2}\overline{q})^2  - 3 \right)\cdot
 \\ [0.2cm]
(5,1): &  ({\cal H}_6-\frac{3}{2}\overline{q}) \left( ({\cal H}_6-\frac{3}{2}\overline{q})^2  - 1 \right)
 \left( ({\cal H}_6-\frac{3}{2}\overline{q})^2  - 3 \right)\cdot  \\ [0.2cm]
(3,1^3): &
\begin{array}{c}
 ({\cal H}_6+\frac{1}{2}\overline{q}) \left( ({\cal H}_6+\frac{1}{2}\overline{q})^2  - 1 \right)
\left( ({\cal H}_6+\frac{1}{2}\overline{q})^2  - 3 \right)  \\ [0.1cm]
 \left( ({\cal H}_6+\frac{1}{2}\overline{q})^2  - (\sqrt{3}+1)^2 \right)
\left( ({\cal H}_6+\frac{1}{2}\overline{q})^2  - (\sqrt{3}-1)^2 \right)
\cdot
\end{array}
 \\ [0.5cm]
(4,1^2): &
\begin{array}{c}
({\cal H}_6-\frac{1}{2}\overline{q}) \left( ({\cal H}_6-\frac{1}{2}\overline{q})^2  - 1 \right)
 \left( ({\cal H}_6-\frac{1}{2}\overline{q})^2  - 3 \right)  \\ [0.1cm]
 \left( ({\cal H}_6-\frac{1}{2}\overline{q})^2  - (\sqrt{3}+1)^2 \right)
\left( ({\cal H}_6-\frac{1}{2}\overline{q})^2  - (\sqrt{3}-1)^2 \right) \cdot
\end{array}
\end{array}
\ee
\be
\lb{char6b}
\begin{array}{l}
  (4,2): \;\; \left[ \; \left( 3\overline{q}^3-20\overline{q}  + (24-14 \overline{q}^2) {\cal H}_6  +
20 \overline{q} {\cal H}_6^2  - 8 {\cal H}_6^3 \right) \right. \cdot \\ [0.2cm]
  \;\;\;\;\;\; \left\{ 9 \overline{q}^6-228 \overline{q}^4+512 \overline{q}^2-256  -
(1408 \overline{q} - 1280 \overline{q}^3+84 \overline{q}^5) {\cal H}_6  + \right. \\ [0.1cm]
  \;\;\; + (768 - 2528 \overline{q}^2+316 \overline{q}^4) {\cal H}_6^2
 + (2048\overline{q} -608 \overline{q}^3)  {\cal H}_6^3  + \\ [0.1cm]
  \;\;\; \left. + (624 \overline{q}^2 -576 ) {\cal H}_6^4 - 320 \overline{q} {\cal H}_6^5
+ 64 {\cal H}_6^6 \right\} \cdot \\ [0.4cm]
 (2^3): \;\;  \left(
3\overline{q}^4+16\overline{q}^2-64  + (8\overline{q}^3+160\overline{q}) {\cal H}_6  +
(-8\overline{q}^2+128) {\cal H}_6^2  - 32 \overline{q}  {\cal H}_6^3  - 16 {\cal H}_6^4 \right)
\cdot \\ [0.4cm]
(3,2,1):  \;  \left( \overline{q}^8+16 \overline{q}^4-256 \overline{q}^2-256  +
(2560+ 1024 \overline{q}^2+64 \overline{q}^4) {\cal H}_6  - \right. \\ [0.1cm]
  \;\;\; -(5120 + 1152 \overline{q}^2+192 \overline{q}^4+16\overline{q}^6) {\cal H}_6^2
 - (2048+512\overline{q}^2)  {\cal H}_6^3  + \\ [0.1cm]
   \;\;\; \left.\left.
+(8448 +1536\overline{q}^2+96 \overline{q}^4) {\cal H}_6^4 + 1024 {\cal H}_6^5
-(3072+256\overline{q}^2) {\cal H}_6^6 + 256 {\cal H}_6^8 \right) \; \right] \cdot \\ [0.4cm]
\;\;\;  \cdot \left[ {\cal H}_6 \to - {\cal H}_6  \right] \; ,
\end{array}
\ee
where in the left-hand-side we indicate the corresponding representations which have dimensions
$$
{\rm dim}(6)= {\rm dim}(1^6)=1, \; {\rm dim}(3,1^3)= {\rm dim}(4,1^2)= 10, \; {\rm dim}(5,1)
= {\rm dim}(2,1^4) =5 \, .
$$
The factors in
(\ref{char6b}) correspond to the representation
$(2^3)$  with dim =5, the representation $(3,2,1)$  with dim=16, the representation
$(4,2)$ with dim=9 and their dual irreps which can be obtained from the previous ones by substitution
${\cal H}_6 \to - {\cal H}_6$.
The sum of the dimensions (\ref{dime}) for all these representations of
$H_6$ is equal to 76. Since the order of the characteristic polynomial is equal to $72$, we conclude
that some of these eigenvalues are degenerated. Two of such eigenvalues appear in the dual hook-type
irreps $(3,1^3)$, $(4,1^2)$ and other two appear in the dual nontrivial irreps
$(3,3)$, $(2^3)$ (see below). It is clear that the degenerated eigenvalues are
$\pm \frac{1}{2}\overline{q}$ (with the multiplicities 3).
The important problem is to find an additional operator $j_k$ which commutes with the
Hamiltonian ${\cal H}_6$ and removes this degeneracy.

{\bf Remark.} The factor which corresponds to the self-dual representation
$(3,2,1)$ with dim$=8$ presented in (\ref{char6b}), can also be written in the concise form
(we remove the common factor $2^8$)
\be
\lb{char6c}
\begin{array}{c}
\left\{ v^8+v^4-4 v^2-1+10 x+16 x v^2+4 x v^4-20 x^2-18 x^2 v^2-12 x^2 v^4  \right. \\
\left. -4 x^2 v^6-8 x^3-8 x^3 v^2+33 x^4+24 x^4 v^2+6 x^4 v^4+4 x^5-12 x^6-4 x^6 v^2+x^8 \right\} =
\\ [0.2cm]
=  Z^4 Y^4 -Z^2 Y^2(6x+1)(2x-1) + \left[ 4 Z Y  + (4 x^2+6 x-1)\right] (2x-1)^2 \; ,
\end{array}
\ee
where $x = {\cal H}_6$, $v=\frac{\overline{q}}{2}$, $Z=x+v$, $Y=x-v$.

\subsection{Characteristic polynomials for ${\cal H}_n$ in the representations $(n-2,2)$}

Now we impose additional relations on the generators $T_k$ of the Hecke algebra (\ref{braidg}), (\ref{ahecke}):
 \be
\lb{tlq}
T_k(q^2) \, T_{k - 1}(q^4) \, T_k(q^2) = 0  \; , \;\;\;
T_k(q^2) \, T_{k + 1}(q^4) \, T_k(q^2) = 0  \; ,
\ee
where $T_k(x)$ are the baxterized elements (\ref{baxtH}).
The factor of the Hecke algebra
over  the relations (\ref{tlq}) is called the Temperley-Lieb algebra $TL_n$. It is known that
all irreps of the algebra $TL_n$ coincide with irreps $\rho_{(n-k,k)}$
(here $n \geq 2k$) of the Hecke algebra  numerated by the Young diagrams
$(n-k,k)$ with only two rows. The spectrum of all Hamiltonians $\rho_{(n-k,k)}({\cal H}_n)$
$(k=0,1,\dots,[\frac{n}{2}])$
gives the energy spectrum of the $XXZ$ Heisenbegr spin chain of the length $n$
(see, e.g., \cite{Kuli} and references therein).
In this subsection we present the characteristic polynomials for the Hamiltonians $\rho_{(n-2,2)}({\cal H}_n)$.
The
dimensions of the representations $(n-2,2)$ are dim$\rho_{(n-2,2)}=\frac{n(n-3)}{2}$.
Our method of the calculation is the following. We construct
explicitly matrix representations of $\rho_{(n-k,k)}({\cal H}_n)$  (see the next subsection) and then
use the Mathematica to evaluate the characteristic polynomials of $\rho_{(n-k,k)}({\cal H}_n)$.

\noindent
{\bf 1. The case $n=4$ and representation $(2,2)$ with dim$=2$.} The Hamiltonian ${\cal H}_4=x$
(see (\ref{free})) has the
characteristic polynomial (see (\ref{char4}))
\be
\lb{2-2}
\left( x^2 -v^{2}-2 \right) = Y Z -2  \; ,
\ee
where $Z=x+v$, $Y=x-v$.

\noindent
{\bf 2. The case $n=5$ and representation $(3,2)$ with dim$=5$.} The Hamiltonian ${\cal H}_5=x$ has the
characteristic polynomial as a product of two factors of orders $2$ and $3$ (see (\ref{char5a}))
\be
\lb{3-2}
\begin{array}{c}
(x^2 - \overline{q}x-1)
\left( x^3  -  \overline{q}  x^2 -5 x +2 \overline{q} \right)  = \\
  = (Y Z -1)  \cdot \left\{ Y  Z^2 - (2 Y + 3 Z) \right\}\; .
  \end{array}
\ee
where $Z=x$, $Y=x-2v$.

\noindent
{\bf 3. The case $n=6$ and representation $(4,2)$ with dim$=9$.} The Hamiltonian ${\cal H}_6=x$ has the
characteristic polynomial (see (\ref{char6b})) which is factorized into two factors of the $3$-rd and
$6$-th orders
\be
\lb{4-2}
\begin{array}{l}
   \;\;\; \left\{ -3v^3+5 v +(7 v^2-3) x  -
5 v x^2  + x^3 \right\} \cdot \\ [0.2cm]
\left\{ 9 v^6-57 v^4+32 v^2-4 + (-44 v+160 v^3-42 v^5) x  \right.
\\ \left.
+ (12-158 v^2+79 v^4) x^2
 +(64 v-76 v^3) x^3 + (39 v^2-9) x^4   -10 v x^5  + x^6 \right\}
\end{array}.
\ee
In terms of the new variables $Z=x-v$, $Y=x-3 v$ the factors in (\ref{4-2}) are simplified to be
$$
\begin{array}{l}
   \;\;\; \left\{ Y Z^2-(Y+2 Z) \right\} \cdot
\left\{ Y^2 Z^4- (5 Y+4 Z) Y Z^2 +2 (5 Y + Z) Z -4 \right\}
\end{array}.
$$

\noindent
{\bf 4. The case $n=7$ and the representation $(5,2)$ with dim$=14$.} For the Hamiltonian
${\cal H}_7=x$ the characteristic polynomial in this representation is factorized into two
factors of the $6$-th and $8$-th orders. In terms of the new variables  $Z=x-2v$, $Y=x-4v$ it reads
$$
\begin{array}{c}
\left\{ Z^4 Y^2 -3 Z^2 Y (Y+Z) + Z (Z+4 Y) - 1 \right\} \cdot \\ [0.2cm]
\left\{ Z^6 Y^2 -(9 Y+5 Z) Z^4 Y  + Z^2 (Z+6 Y) (5 Z+2 Y) -(5 Z+2 Y)^2 \right\} \; .
\end{array}
$$

\noindent
{\bf 5. The case $n=8$ and the representation $(6,2)$ with dim$=20$.} For the Hamiltonian
${\cal H}_8=x$ the characteristic polynomial in this representation is factorized into two
factors of the $8$-th and $12$-th orders
$$
\begin{array}{c}
\left\{ Z^6 Y^2 -2 Z^4 Y (3 Y+2Z) + Z^2 (Z+5 Y) (3 Z+Y) -(3 Z+Y)^2 \right\} \cdot \\ [0.2cm]
\left\{ Y^3 Z^9 - Z^7 Y^2 (14 Y+6 Z) + Z^5 Y (9 Z^2+68 Z Y+49 Y^2) \right. \\
 -Z^3(2 Z^3+85 Z^2 Y+168 Z Y^2+49 Y^3) + 2 Z^2 (9 Z^2+68 Z Y+49 Y^2) \\
\left.
-8 Z (7 Y+3 Z) + 8 \right\}
\end{array}
$$
where $Z=x- 3 v$ and $Y=x- 5 v$.

\noindent
{\bf 6. The case $n=9$ and the representation $(7,2)$ with dim$=27$.} For the Hamiltonian
${\cal H}_9=x$ the characteristic polynomial in this representation is factorized into two
factors of the $12$-th and $15$-th orders
$$
\begin{array}{c}
\left\{ Y^3 Z^9 -5 Y^2 Z^7 (2 Y+Z) +3 Y Z^5 \left(8 Y^2+13 Z
Y+2 Z^2\right)  \right. \\
-Z^3 \left(16 Y^3+64 Z Y^2+38 Z^2 Y+Z^3\right) +3 Z^2 \left(8 Y^2+13 Z Y+2  Z^2\right) \\
\left.
-5 Z (2 Y+Z) +1  \right\} \cdot
\end{array}
$$
$$
\begin{array}{c}
\left\{ Y^3 Z^{12} -Y^2 Z^{10} (20 Y+7 Z) +Y Z^8 \left(126  Y^2+121 Z Y+14 Z^2\right) \right. \\
\!\!\!\!\! -Z^6 \left(304 Y^3+620 Z Y^2+212 Z^2 Y+7 Z^3\right) +Z^4 (2 Y+7 Z) \left(126 Y^2+121 Z Y+14
Z^2\right) \\
\left.
-Z^2 (2 Y+7 Z)^2 (20 Y+7 Z) +(2 Y+7 Z)^3  \right\}
\end{array}
$$
where $Z=x- 4 v$, $Y=x- 6 v$.

\noindent
{\bf 7. The case $n=10$ and the representation $(8,2)$ with dim$=35$.} The characteristic
polynomial in this representation is factorized into two factors of the $15$-th and $20$-th orders
\be
\lb{8-2i}
\begin{array}{l}
P_{(8,2)} = \left\{ Z^{12} Y^3 -3 Z^{10} Y^2 (5 Y+2 Z) + Z^8 Y \left(69 Y^2+76 Z Y+10 Z^2\right)
- \right. \\
 - Z^6  \left(119 Y^3+278 Z Y^2+109 Z^2 Y+4 Z^3\right)
 + Z^4 (Y+4 Z) \left(69 Y^2+76 Z Y + 10 Z^2\right) - \\
 \left. - 3 Z^2 (5 Y+2 Z) (Y+4 Z)^2 +(Y+4 Z)^3 \right\} \cdot
   \end{array}
\ee
$$
\begin{array}{l}
\cdot \left\{ Z^{16} Y^4 - Z^{14} Y^3 (27 Y+8 Z) + Z^{12} Y^2 \left(261 Y^2+194 Z Y+20 Z^2\right)
\right. \\
-Z^{10} Y \left(1143 Y^3+1632 Z Y^2+439 Z^2 Y+16 Z^3\right) \\
+ Z^8 \left(2349 Y^4+5982 Z Y^3+3216 Z^2 Y^2+326 Z^3 Y+2 Z^4\right) \\
   -  Z^6 \left(2187 Y^4+9720 Z Y^3+9812 Z^2 Y^2+2124 Z^3 Y+40 Z^4\right) \\
   +3 Z^4 \left(243 Y^4+2214 Z Y^3+4098 Z^2 Y^2+1816 Z^3 Y+84 Z^4\right)  \\
-2 Z^3 \left(729 Y^3+3033 Z Y^2+2584 Z^2 Y+304 Z^3\right)  \\
\left. +36 Z^2 \left(27 Y^2+54 Z Y+14 Z^2\right) -80 Z (3 Y+2 Z) +16 \right\}
\end{array}
$$
where $x={\cal H}_{10}$, $Z=x-5v$, $Y=x-7v$.

\noindent
{\bf 8. The case $n=11$ and the representation $(9,2)$ with dim$=44$.} The characteristic
polynomial in this representation is factorized into two factors of the $20$-th and $24$-th order
\be
\lb{9-2i}
\begin{array}{l}
P_{(9,2)} = \left\{ Z^{16} Y^4  -7 Z^{14} Y^3  (3 Y+Z)  +5 Z^{12} Y^2  \left(31 Y^2+26 Z Y+3
Z^2\right)
\right. \\
-Z^{10} Y  \left(510 Y^3+822 Z Y^2+249 Z^2 Y+10 Z^3\right)  \\
+ Z^8 \left(775 Y^4+2228 Z Y^3+1351 Z^2 Y^2+153 Z^3 Y+Z^4\right) \\
-  Z^6  \left(525 Y^4+2635 Z Y^3+3002 Z^2 Y^2+730 Z^3 Y+15 Z^4\right)\\
+Z^4 \left(125 Y^4+1290 Z Y^3+2697 Z^2 Y^2+1346 Z^3 Y+69  Z^4\right) \\
   -Z^3 \left(200 Y^3+941 Z Y^2+904 Z^2 Y+119 Z^3\right) \\
\left.   +3 Z^2 \left(35 Y^2+79 Z Y+23 Z^2\right) -5 Z (4 Y+3 Z) +1  \right\} \cdot
\end{array}
\ee
$$
\begin{array}{l}
\left\{ Z^{20} Y^4 -  Z^{18} Y^3 (35 Y+9 Z)+ Z^{16} Y^2 \left(475 Y^2+290 Z Y+27 Z^2\right) \right. \\
- Z^{14} Y \left(3230 Y^3+3558 Z Y^2+805 Z^2 Y+30 Z^3\right)   \\
+ Z^{12}  \left(11875 Y^4+21404 Z  Y^3+8949 Z^2 Y^2+839 Z^3 Y+9 Z^4\right) \\
- Z^{10} \left(23883 Y^4+67717 Z Y^3+47590 Z^2 Y^2+8550 Z^3 Y+243 Z^4\right)  \\
+Z^8 \left(25365 Y^4+113066 Z Y^3+128825 Z^2 Y^2+40518 Z^3 Y+2349 Z^4\right)  \\
-Z^6 (2 Y+9 Z) \left(6650 Y^3+17345 Z Y^2+10194 Z^2 Y+1143 Z^3\right)  \\
+Z^4 (2 Y+9 Z)^2 \left(855 Y^2+1183 Z Y+261 Z^2\right) \\
\left. - Z^2(2 Y+9 Z)^3 (50 Y+27 Z) +(2 Y+9 Z)^4 \right\}
\end{array}
$$
where $x={\cal H}_{11}$, $Z=x-6v$, $Y=x-8v$.


\noindent
{\bf 9.  The case $n=12$ and the representation $(10,2)$ with dim$=54$.} The characteristic
polynomial in this representation is factorized into two factors of the $24$-th and $30$-th order
\be
\lb{10-2i}
\begin{array}{l}
P_{(10,2)} = \left\{ Z^{20} Y^4 -4 Z^{18} Y^3 (7 Y+2 Z) +3 Z^{16} Y^2 \left(100 Y^2+68 Z Y+7
Z^2\right)
\right. \\
 - Z^{14} Y \left(1591 Y^3+1954 Z Y^2+491 Z^2 Y+20 Z^3\right) \\
 + Z^{12} \left(4508 Y^4+9064 Z Y^3+4218 Z^2 Y^2+436 Z^3 Y+5 Z^4\right) \\
 - Z^{10} \left(6907 Y^4+21850 Z Y^3+17112 Z^2 Y^2+3406 Z^3 Y+105 Z^4\right)  \\
 + Z^8 \left(5527 Y^4+27480 Z Y^3+34909 Z^2 Y^2+12200 Z^3 Y+775 Z^4\right) \\
 -2 Z^6 (Y+5 Z) \left(1082 Y^3+3150 Z Y^2+2061 Z^2 Y+255 Z^3\right) \\
 +Z^4 (Y+5 Z)^2 \left(411 Y^2+634 Z Y+155 Z^2\right) \\
 \left.  -7 Z^2 (5 Y+3 Z) (Y+5 Z)^3 +(Y+5 Z)^4 \right\} \cdot
\end{array}
\ee
$$
\begin{array}{l}
\left\{ Z^{25} Y^5-2 Z^{23} Y^4 (22 Y+5 Z) +  Z^{21} Y^3 \left(792 Y^2+412 Z Y+35 Z^2\right) \right. \\
- Z^{19} Y^2 \left(7623 Y^3+6866 Z Y^2+1355 Z^2 Y+50 Z^3\right) \\
+ Z^{17} Y \left(43076  Y^4+60390 Z Y^3+20954 Z^2 Y^2+1834 Z^3 Y+25 Z^4\right) \\
   - Z^{15} \left(147983 Y^5+307010 Z Y^4+168558 Z^2 Y^3+26510 Z^3 Y^2+883 Z^4 Y +2 Z^5\right)  \\
+ Z^{13} \left(310123 Y^5+930974 Z Y^4+770091 Z^2 Y^3+196172 Z^3 Y^2+12139 Z^4 Y+70 Z^5\right)  \\
-2  Z^{11} \left(194326 Y^5+840587 Z Y^4+1026993 Z^2 Y^3+404211 Z^3 Y^2+42041 Z^4 Y+475 Z^5\right) \\
 + Z^9  \left(278179 Y^5+1759824 Z Y^4+3173478 Z^2 Y^3+1898244 Z^3 Y^2+317599 Z^4 Y+6460 Z^5\right) \\
 - Z^7 \left(102487 Y^5+1011560 Z Y^4+2742586 Z^2 Y^3+2500438 Z^3 Y^2+665275 Z^4 Y+23750 Z^5\right) \\
 +  Z^5 \left(14641 Y^5+284834 Z Y^4+1251866 Z^2 Y^3+1769240 Z^3 Y^2+753437 Z^4 Y+47766 Z^5\right) \\
-2  Z^4 \left(14641 Y^4+135520 Z Y^3+320474 Z^2 Y^2+219128 Z^3 Y+25365 Z^4\right) \\
 +8  Z^3 \left(2662 Y^3+13673 Z Y^2+16106 Z^2 Y+3325 Z^3\right) \\
\left. -8 Z^2 \left(847 Y^2+2222 Z Y+855 Z^2\right) +80  Z (11 Y+10 Z)-32 \right\} \; ,
\end{array}
$$
where $x = {\cal H}_{12}$, $Z = x-7 v$ and $Y = x-9 v$.

\noindent
{\bf 10. The case $n=13$ and the representation $(11,2)$ with dim$=65$.} The characteristic polynomial in this
representation is factorized into two factors of the $30$-th and $35$-th order
\be
\lb{11-2i}
\begin{array}{l}
Y^5 Z^{25}-9 Y^4 (4 Y+Z) Z^{23}+7 Y^3 \left(75 Y^2+43 Z Y+4 Z^2\right) Z^{21}-  \\
Y^2 \left(4056 Y^3+4031 Z Y^2+874 Z^2 Y+35 Z^3\right) Z^{19}+ \\
Y \left(18231 Y^4+28222 Z Y^3+10781 Z^2 Y^2+1030 Z^3 Y+15 Z^4\right) Z^{17} - \\
\left(49380 Y^5+113163 Z Y^4+68496 Z^2 Y^3+11805 Z^3 Y^2+424 Z^4 Y+Z^5\right) Z^{15}+ \\
\left(80891 Y^5+268257  Z Y^4+244835 Z^2 Y^3+68531 Z^3 Y^2+4605 Z^4 Y+28 Z^5\right) Z^{13} - \\
\left(78576 Y^5+375429 Z Y^4+506270 Z^2 Y^3+219334 Z^3 Y^2+24905 Z^4 Y+300 Z^5\right) Z^{11} + \\
\left(43200 Y^5+301984 Z Y^4+601090 Z^2 Y^3+396150 Z^3 Y^2+72634 Z^4 Y+1591 Z^5\right) Z^9- \\
2 \left(6048 Y^5+66096 Z Y^4+197920 Z^2 Y^3+198881 Z^3 Y^2+58122 Z^4 Y+2254 Z^5\right) Z^7+  \\
 \left(1296 Y^5+28080 Z Y^4+136554 Z^2 Y^3+212848 Z^3 Y^2+99644 Z^4 Y+6907 Z^5\right) Z^5- \\
 \left(2160 Y^4+22176 Z Y^3+57906  Z^2 Y^2+43570 Z^3 Y+5527 Z^4\right) Z^4+ \\
 \left(1296 Y^3+7360 Z Y^2+9551 Z^2 Y+2164 Z^3\right) Z^3- \\
 3 \left(112 Y^2+324 Z Y+137 Z^2\right) Z^2+35 (Y+Z) Z-1
\end{array}
\ee
$$
\begin{array}{l}
Y^5 Z^{30}-Y^4 (54 Y+11 Z) Z^{28}+Y^3 \left(1239 Y^2+563 Z Y+44 Z^2\right) Z^{26}- \\
Y^2 \left(15894 Y^3+12153 Z Y^2+2140 Z^2 Y+77 Z^3\right) Z^{24}+ \\
Y \left(126279  Y^4+145446 Z Y^3+43545 Z^2 Y^2+3578 Z^3 Y+55 Z^4\right) Z^{22}- \\
\left(650946 Y^5+1067749 Z Y^4+486798 Z^2 Y^3+68967 Z^3 Y^2+2466 Z^4 Y+11 Z^5\right) Z^{20}+ \\
\left(2219569 Y^5+5028863 Z Y^4+3303181 Z^2 Y^3+723221 Z^3 Y^2+45485 Z^4 Y+484 Z^5\right) Z^{18}- \\
\left(5017266 Y^5+15459111 Z Y^4+14203130 Z^2 Y^3+4550870 Z^3 Y^2+451913 Z^4 Y+8712 Z^5\right) Z^{16}+ \\
\left(7433784 Y^5+30999936 Z Y^4+39276278 Z^2 Y^3+17901642 Z^3 Y^2+2662000 Z^4 Y+83853 Z^5\right)Z^{14}- \\
2 \left(3523048 Y^5+19967988 Z Y^4+34790550 Z^2 Y^3+22275099 Z^3 Y^2+4830078 Z^4 Y+236918 Z^5
\right) Z^{12}+ \\
\left(4121784 Y^5+32057560 Z Y^4+77399690
Z^2 Y^3+69595592 Z^3 Y^2+21779274 Z^4 Y+1627813 Z^5\right) Z^{10}- \\
(2 Y+11 Z) \left(715128 Y^4+3714176 Z Y^3+5556166 Z^2 Y^2+2682086 Z^3 Y+310123 Z^4\right) Z^8+ \\
(2 Y+11 Z)^2 \left(71532 Y^3+233700 Z Y^2+191279 Z^2 Y+35332 Z^3\right) Z^6- \\
(2 Y+11 Z)^3 \left(3924 Y^2+7128 Z Y+2299 Z^2\right) Z^4+7 (2 Y+11 Z)^4 (15 Y+11 Z) Z^2-(2 Y+11 Z)^5
\end{array}
$$
where $x = {\cal H}_{13}$, $Z = x-8 v$ and $Y = x-10 v$.

\vspace{0.5cm}

In all examples considered above, the characteristic polynomials for
$\rho_{(n-2,2)}({\cal H}_n)$ are factorized into two factors with integer coefficients.
So we formulate the following Conjecture.

\noindent
{\bf Conjecture.} {\it Let $n \geq 4$. For irrep of the Hecke algebra $H_n(q)$
with Young diagram $(n-2,2)$ and dimension $\frac{n(n-3)}{2}=p_{n-1} + p_n$
the characteristic polynomial for the
Hamiltonian $\rho_{(n-2,2)}({\cal H}_n)=x$ is represented as the
product of two polynomial factors: the short factor {\rm short}$_{n}$ of the order $p_{n-1}$
and the long factor {\rm long}$_{n}$ of the order $p_n$ with integer coefficients, where
$$
p_n = \frac{1}{8}\left( ((-1)^n -1)3-4n +2n^2 \right) =
\left\{
\begin{array}{l}
\frac{1}{4} n (n-2), \; for\;\; {\rm even} \; n \\ [0.2cm]
\frac{1}{4} (n+1)(n-3)n \; for\;\; {\rm odd} \; n
\end{array}
\right.
\; .
$$
I.e.,
$$
\begin{array}{c}
p_3=0 \; , \;\; p_4=2 \; , \;\; p_5=3 \; , \;\; p_6=6 \; , \;\; p_7=8 \; , \;\;
 p_8=12 \; , \;\; p_9=15 \; , \\
  p_{10}=20 \; , \;\; p_{11}=24 \; , \;\; p_{12} = 30 \; , \;\;
 p_{13} = 35 \; , \;\; p_{14} = 42 \; , \dots \; .
 \end{array}
$$
These polynomial factors are $(p_n = k_n + \overline{k}_n)$
\be
\lb{short1}
\begin{array}{c}
{\rm short}_n = Z^{\overline{k}_{n-1}} Y^{k_{n-1}} \left\{ 1 -  (n-4) Z^{-1} Y^{-1} -
 \frac{(n-4)(n-5)}{2} Z^{-2} + \right. \\
 + \frac{(n-5)(n-6)}{2}  Z^{-2} Y^{-2}
  + \frac{(n-6)(n^2-7n+8)}{2} Z^{-3} Y^{-1} + \frac{(n-6)(n-7)(n^2-5 n-4)}{8} Z^{-4} - \\
   \left.
- \frac{(n-6)(n-7)(n-8)}{6} Z^{-3} Y^{-3} - \frac{(n^4-20n^3+137n^2-338n+116)}{4} Z^{-4} Y^{-2} -
   (\dots) Z^{-5} Y^{-1} \right\} + \\
   \dots \dots \dots \dots \dots  \\
  +(-1)^{[(n-2)/4]} \frac{(1+(-1)^n)}{2} \left\{ (\frac{n}{2}-1) Z+ Y \right\}^{k_{n-1}}
+(-1)^{[(n-1)/4]} \frac{(1-(-1)^n)}{2} \; ,
\end{array}
\ee

\be
\lb{long1}
\begin{array}{c}
{\rm long}_n =Z^{\overline{k}_n} Y^{k_n}\left\{ 1 -
 (n-2) Z^{-1} Y^{-1} - \frac{(n-1)(n-4)}{2} Z^{-2} + \right. \\
\!\!\!\!\!   +  \frac{(n-2)(n-5)}{2} Z^{-2} Y^{-2}
+ \left( \! \frac{(n-4)(n-5)(n+9)}{3} + \frac{(n-6)(n-7)(n-8)}{6} \! \right) Z^{-3} Y^{-1}
 + \frac{(n-6)(n-1)(n^2-3n-12)}{8} Z^{-4} \\
  - \frac{(n-2)(n-6)(n-7)}{6} Z^{-3} Y^{-3}
 - \frac{(n^4-12n^3+37n^2+18n-124)}{4} Z^{-4} Y^{-2}  - \\
 \left.
 - \frac{(n^5-12n^4+23n^3+128n^2-252n-224)}{8} Z^{-5} Y^{-1}  +
  \dots   \dots   \dots   \dots   \dots   \dots   \dots \right\} +    \\
+ (-1)^{^{\left[\frac{(n-1)}{4}\right]}} \frac{(1-(-1)^n)}{2} \left\{(n-2)Z + 2Y\right\}^{k_n}
  +(-1)^{\left[\frac{n}{4}\right]} \frac{(1+(-1)^n)}{2} 2^{k_n} \; ,
\end{array}
\ee
where $[x]$ -- integer part of $x$ (e.g., $[\frac{n}{4}]$ -- integer part of $n/4$),
$$
 Z = x - \left( n -5 \right) \, v \; , \;\;\;
 Y = x - \left( n-3 \right) v \; , \;\;\;v = \frac{q + q^{-1}}{2} \; ,
$$
$$
k_n = \left[ \frac{n}{2} \right] - 1
 = \left\{
\begin{array}{l}
\frac{1}{2} (n-2) \;\;\; {\rm even} \; n \; ,\\ [0.2cm]
\frac{1}{2} (n-3) \;\;\; {\rm odd} \; n
\end{array}
\right.
$$
and
$$
\overline{k}_n = \frac{1}{8}((-1)^n+7-8n+2n^2) =
\left\{
\begin{array}{l}
\frac{1}{4} (n-2)^2 = k_n^2 \;\;\; {\rm even} \; n \; , \\ [0.2cm]
\frac{1}{4} (n-1)(n-3) = k_n^2 + k_n \;\;\; {\rm odd} \; n.
\end{array}
\right.
$$
 }
For $n$ -- odd, one can write (\ref{short1}) and (\ref{long1}) as the series
$$
\begin{array}{c}
{\rm long}_{n} \sim
{\rm short}_{n+1} \sim Z^{\overline{k}_{n}} Y^{k_{n}} \left( 1 - Z^{-2}(C_{2,1}\frac{Z}{Y} +C_{2,0} )
+ Z^{-4}(C_{4,2}\frac{Z^2}{Y^2} + C_{4,1}\frac{Z}{Y} +C_{4,0} ) \right. \\ [0.2cm]
\left. - Z^{-6}(C_{6,3}\frac{Z^3}{Y^3} +C_{6,2}\frac{Z^2}{Y^2} + C_{6,1}\frac{Z}{Y} +C_{6,0} )
+ Z^{-8} \sum\limits_{j=0}^4 C_{8,j} \left( \frac{Z}{Y}\right)^j - \dots \right) =
\end{array}
$$
$$
\begin{array}{c}
 = Z^{\overline{k}_{n}} Y^{k_{n}}\left( \sum\limits_{m=0}^{\frac{\overline{k}_{n}}{2}} (-Z^{-2})^m
\sum\limits_{j=0}^m
 C_{2m,j} \left( \frac{Z}{Y}\right)^j \right).
\end{array}
$$
For $n$ -- even, eqs. (\ref{short1}) and (\ref{long1}) also can be written as the series over $\frac{Z}{Y}$. But we do not present it explicitly here.

\vspace{1cm}

 \noindent
 {\bf Remark 1.}
For the hook-type representations $(n-1,1)$ we have
 the following spectrum for the Hamiltonian (see \cite{GD}, \cite{IsOgO})
 ${\cal H}_n = \sum\limits_{k=1}^{n-1} (T_k -q)$
 $$
 Spec({\cal H}_n) = Spec(\sum_{k=1}^{n-1} (T_k -q)) =
  2 \cos \left(\frac{\pi m}{n} \right) - (q+q^{-1}) \; ,
 \;\;\; m=1,\dots,n-1 \; .
 $$
 If,  as usual (cf. (\ref{Ex1})), we substitute
 $2 \cos \left(\frac{\pi m}{n} \right) = X^{1 \over 2} + X^{-{1 \over 2}}$, then for
 $X$ we will have the characteristic identity $X^{n} - 1 =0$, $X \neq  1$. We see that the spectrum of
 the open $XXZ$ spin chain (for even $m$) contains the spectrum of one-magnon states
 (except for the case $X=1$) for closed
 $XXZ$ spin chain (see (\ref{solep8})).

 \vspace{1cm}

 \noindent
 {\bf Remark 2.} We have calculated the characteristic polynomials $P_{(n-3,3)}$ for the Hamiltonian (\ref{free})
 in the irreps $(n-3,3)$ $(n=6,7,8,9)$ and observed the same factorization of $P_{(n-3,3)}$ into two
 factors, which are  the polynomials with the integer coefficients.

  \vspace{1cm}

 \noindent
 {\bf Remark 3.}  The quantum inverse scattering (R-matrix) method and
 the algebraic Bethe ansatz method for the open XXZ spin chain were elaborated
 by Sklyanin in \cite{Skl} (about analytical Bethe ansatz approach see \cite{MeNe,Zho} and references therein).
 The quantum group symmetry in the open XXZ spin chain was discovered in \cite{PS}.

 \subsection{Method of calculation}

In this subsection, we explain the method of construction of the explicit matrix
irreps (related to the fixed Young diagrams $\Lambda$) for the Hecke algebra $H_{n}(q)$.
A similar method was also considered in \cite{GD}, \cite{LLT}.

First of all, we define the affine extension $\hat{H}_{n}(q)$ of the Hecke algebra $H_{n}(q)$.
The affine Hecke algebra $\hat{H}_{n}(q)$ (see, e.g.,
Chapter 12.3 in \cite{ChPr}) is an extension of the Hecke
algebra $H_{n}(q)$ by the additional affine elements $y_k$ $(k=1,\dots,n)$
subjected to the relations:
\be
\lb{afheck}
y_{k+1} = T_k \, y_k \, T_k \; , \;\;\; y_k \, y_j = y_j \, y_k \; , \;\;\;
 y_j \, T_i  =  T_i \, y_j \;\; (j \neq i,i+1) \; .
\ee
The elements $\{ y_k \}$ are called {\it Jucys--Murphy elements}
and form a commutative subalgebra in $\hat{H}_{n}$,
while the symmetric functions in $y_k$ form the center in $\hat{H}_{n}$.
Let us introduce the intertwining elements \cite{Isaev1}
(presented in another form in \cite{Cher5})
\be
\label{impint}
U_{m+1} =  (T_m y_m - y_m T_m) \frac{1}{f(y_m, y_{m+1})} \; \in \hat{H}_{n}(q)
 \;\;\; (1 \leq m \leq n-1) \; ,
\ee
where $f(y_m, y_{m+1})$ is an arbitrary function of the two variables $y_m$, $y_{m+1}$. The elements  $U_{i}$
satisfy relations \cite{Isa1}:
\be
\label{impo}
U_m \, U_{m+1} \, U_m = U_{m+1} \, U_{m} \, U_{m+1} \; ,
\ee
\be
\label{importt}
\begin{array}{c}
U_{m+1} y_m = y_{m+1} U_{m+1}  , \; U_{m+1} y_{m+1} = y_{m} U_{m+1}  , \\ [0.2cm]
\; [U_{m+1}, \, y_k ] = 0 \;\;\; (k \neq m,m+1) \; ,
\end{array}
\ee
\be
\label{import2}
U_{m+1}^2 = \frac{(q y_{m} - q^{-1} \, y_{m+1}) \, (q \, y_{m+1} - q^{-1} \, y_m)}{
 f(y_m, y_{m+1}) f(y_{m+1},y_m)} \; .
\ee
As it is seen from (\ref{importt}), the operators $U_{m+1}$ "permute" the elements $y_m$ and $y_{m+1}$,
and this confirms the statement that the center of the Hecke algebra
$\hat{H}_{n}(q)$ is generated by the symmetric functions in $\{ y_i \}$
$(i = 2, \dots , n)$.

One may check that the Hamiltonian (\ref{free}) satisfies
\be
\lb{hu1}
[{\cal H}_{n} , \, y_k] = U_{k+1} f_{k+1} - U_{k} f_{k}  \; , \;\;\;
 [{\cal H}_{n} , \, y_k^{-1}] = \overline{U}_{k+1} f_{k+1}
   - \overline{U}_{k} f_{k+1}  \; ,
\ee
where $f_{k+1}=f(y_k,y_{k+1})$, $\overline{U}_{k+1} = U_{k+1}(y_k y_{k+1})^{-1}$ and
$U_1=U_{n+1}=0$. From (\ref{hu1}) follows that
\be
\lb{hu2}
[{\cal H}_{n} , \, \sum_{i=1}^k y_k] = U_{k+1} f_{k+1} \; , \;\;\;
 [{\cal H}_{n} , \, \sum_{i=1}^k  y_i^{-1}] = \overline{U}_{k+1} f_{k+1} \; .
\ee
Further, it is convenient to fix
$$
f(y_k,y_{k+1})=y_k-y_{k+1} \; .
$$
Now we have
\be
\label{import}
U_{m+1}^2 = \frac{(q \, y_{m+1} - q^{-1} \, y_m) \, (q^{-1} \, y_{m+1}-q y_{m} )}{
(y_{m+1}-y_m)^2} \; ,
\ee
$$
U_{m+1} = T_m  + \frac{\lambda y_{m+1}}{(y_m - y_{m+1})} = \left( T_m -
\frac{\lambda}{2} \right)  + \frac{\lambda}{2} \frac{(y_m + y_{m+1})}{(y_m - y_{m+1})} \; ,
$$
and, therefore,
\be
\lb{hu3}
s_m  = U_{m+1} + v_{m+1} \; ,
\ee
where
 \be
\lb{hu33}
v_{m+1} = \frac{\lambda}{2} \frac{(y_{m+1}+y_m )}{(y_{m+1}-y_m )} \; .
\ee
Due to the relations $s_m^2 = \overline{q}^2/4$,  we conclude that
$$
U_{m+1}^2 + v_{m+1}^2 = \frac{(q+q^{-1})^2}{4} \; , \;\;\; U_{m+1} v_{m+1} + v_{m+1} U_{m+1}=0 \; .
$$

As it was indicated in Remark 2 (see the beginning of this Section),
each irreducible representation of the Hecke algebra $H_{n}(q)$
corresponds to the Young diagram $\Lambda$ with $n$ nodes.
In the representation space of $\Lambda$ the basis vectors $\psi_\alpha$
 are labeled by the standard Young tableaux 
$T_\alpha$ of the shape $\Lambda$. Consider the homomorphism $\rho$: $\hat{H}_{n}(q) \to H_{n}(q)$
which is defined by the map $y_1 \to 1$. 
The images $\rho(y_i)$ of the Jucys--Murphy elements $y_i$ are diagonal in the chosen basis
\be
 \label{jmev}
\rho(y_j) \, \psi_\alpha =c_j(T_\alpha) \, \psi_\alpha \ , 
\ee
where $\psi_\alpha$ is the basis vector associated to the tableau
$T_\alpha$ of shape
$\Lambda$. The eigenvalue $c_j(T_\alpha)$ is the quantum 
 content of the box of the tableau $T_\alpha$ with the number $j$. For the box $j$,
 which is located in the $a$-th row and $b$-th column of $T_\alpha$, 
 the quantum content is defined by
\be\label{jmev1}
c_j(T_\alpha)=q^{2(b-a)} \; .
\ee

Finally, by using eqs. (\ref{hu3}) and (\ref{hu33}) we
write the Hamiltonian (\ref{free}) in the form
\be
\lb{hu4}
\begin{array}{c}
{\cal H}_{n} = \sum_{m=1}^{n-1} s_m  =
\sum_{m=1}^{n-1}  \left(U_{m+1} + v_{m+1} \right) =
\sum_{m=1}^{n-1}  \left(U_{m+1} + \frac{\lambda}{2} \frac{(y_{m+1}+y_m )}{(y_{m+1}-y_m )}\right) = \\
= \sum_{m=1}^{n-1}  \left(\sqrt{U^2_{m+1}} \cdot \widetilde{U}_{m+1} +
\frac{\lambda}{2} \frac{(y_{m+1}+y_m )}{(y_{m+1}-y_m )}\right) \; ,
\end{array}
\ee
where 
$$
\sqrt{U^2_{m+1}} = 
\sqrt{\frac{(q \, y_{m+1} - q^{-1} \, y_m) \, (q^{-1} \, y_{m+1}-q y_{m} )}{
(y_{m+1}-y_m)^2}} \; ,
$$
the operators 
$\widetilde{U}_{m+1}$ permute indices $m \leftrightarrow (m+1)$
 in the standard Young tableaux
of the same shape $\Lambda$: $T_\alpha \to T_{\alpha'}$,
i.e. $\widetilde{U}_{m+1} \cdot \psi_\alpha \sim \psi_{\alpha'}$.  Or, if 
after permutation $m \leftrightarrow (m+1)$ the Young tableau $T_{\alpha'}$
is not standard, then operator $U_{m+1}$ put
corresponding basis vector to zero: $\widetilde{U}_{m+1} \cdot \psi_\alpha =0$.

\vspace{0.3cm}

\noindent
{\bf Example 1.} Consider the basis for the representation, 
which is related to the Young
diagram  $\Lambda = (2,1)$ with quntum content:
\be
\lb{y21}
\begin{tabular}{|c|c|}
\hline
  $\!\!\! $  $_1$ & $\!\! $  $_{q^2}\!\!$   \\[0.1cm]
\hline
  $\!\!\! $ $_{q^{\! -2}}\!\!\!$     \\[0.1cm]
    \cline{1-1}
  \multicolumn{1}{c}{}
\end{tabular}  \;\;\;\; .
\ee
We have 2 standard tableaux of the shape (\ref{y21}):
\be
\lb{y21a}
\psi_0 = \begin{tabular}{|c|c|}
\hline
  $\!\!\! $  $_1$ & $\!\! $  $_2\!\!$   \\[0.1cm]
\hline
  $\!\!\! $ $_3\!\!\!$     \\[0.1cm]
    \cline{1-1}
  \multicolumn{1}{c}{}
\end{tabular} \;\; , \;\;\;
\psi_1 =  U_3 \psi_0 = \begin{tabular}{|c|c|}
\hline
  $\!\!\! $  $_1$ & $\!\! $  $_3\!\!$   \\[0.1cm]
\hline
  $\!\!\! $ $_2\!\!\!$     \\[0.1cm]
    \cline{1-1}
  \multicolumn{1}{c}{}
\end{tabular}
\ee
 and the space of the representation $\Lambda = (2,1)$ is two-dimensional.
Using (\ref{hu3}), we obtain the action of $s_i$ on the vectors $\psi_0, \psi_1$  (\ref{y21a})
$$
\begin{array}{c}
s_1 \psi_0 = \frac{1}{2} \overline{q} \psi_0 \; , \;\; s_1 \psi_1 = - \frac{1}{2} \overline{q}
\psi_1 \; , \\ [0.2cm]
s_2 \psi_0 = \psi_1 + \frac{\lambda}{2} \frac{(q^{-2}+q^{2}) }{(q^{-2}-q^{2})}    \psi_0 \; , \;\;
 s_2 \psi_1 = U_3^2 \psi_0 +  \frac{\lambda}{2}
\frac{(q^2+q^{-2})}{(q^2-q^{-2})}    \psi_1  \; ,
\end{array}
$$
where according to
(\ref{import}), (\ref{jmev}) and (\ref{jmev1}) we have
$U_3^2 \psi_0 = \frac{(q^3-q^{-3})(q-q^{-1})}{(q^2-q^{-2})^2} \psi_0$.
The solution of the eigenvalue problem
$(\sum_{i=1}^3 s_i - \nu) (\psi_0 + a \, \psi_1) = 0$,
($a$ is a constant), gives the spectrum $\nu = \pm 1$ 
which leads to two factors in the characteristic identity (\ref{char3}) for $n=3$.

\vspace{0.2cm}

\noindent
{\bf Example 2.} Consider the basis for the representation $(3^2)$ which is related to the Young
diagram with quantum content:
\be
\lb{y33}
\begin{tabular}{|c|c|c|}
\hline
  $\!\!\! $  $_1$ & $\!\! $  $_{q^2}\!\!$  &
  $\!\! $  $_{q^4}\!\!$ \\[0.1cm]
\hline
  $\!\!\! $ $_{q^{\! -2}}\!\!\!$ &   $\!\!\!$ $_1\!\!$  & $\!\!\! $ $_{q^2}\!\!$  \\[0.1cm]
  \hline
\end{tabular}
\ee
We have 5 standard tableaux (the operators $U_{k+1}$ permute numbers $k$ and $k+1$ in the standard
tableaux) of the shape (\ref{y33}):
\be
\label{char6d}
\psi_0 = \begin{tabular}{|c|c|c|}
\hline
  $\!\!\! $  $_1$ & $\!\! $  $_2$  &
  $\!\! $  $_3$ \\[0.1cm]
\hline
  $\!\!\! $ $ _4$  &   $\!\!\!$ $_5$  & $\!\!\! $ $_6$  \\[0.1cm]
  \hline
\end{tabular} \; , \;\;
\psi_1 = U_4 \psi_0 = \begin{tabular}{|c|c|c|}
\hline
  $\!\!\! $  $_1$ & $\!\! $  $_2$  &
  $\!\! $  $_4$ \\[0.1cm]
\hline
  $\!\!\! $ $ _3$  &   $\!\!\!$ $_5$  & $\!\!\! $ $_6$  \\[0.1cm]
  \hline
\end{tabular} \; , \;\;
\psi_2 = U_5 U_4 \psi_0 = \begin{tabular}{|c|c|c|}
\hline
  $\!\!\! $  $_1$ & $\!\! $  $_2$  &
  $\!\! $  $_5$ \\[0.1cm]
\hline
  $\!\!\! $ $ _3$  &   $\!\!\!$ $_4$  & $\!\!\! $ $_6$  \\[0.1cm]
  \hline
\end{tabular} \; , \;\;
$$
$$
\psi_3 =  U_3 U_4 \psi_0 = \begin{tabular}{|c|c|c|}
\hline
  $\!\!\! $  $_1$ & $\!\! $  $_3$  &
  $\!\! $  $_4$ \\[0.1cm]
\hline
  $\!\!\! $ $ _2$  &   $\!\!\!$ $_5$  & $\!\!\! $ $_6$  \\[0.1cm]
  \hline
\end{tabular} \; , \;\;
\psi_4 = U_5 U_3 U_4 \psi_0 = \begin{tabular}{|c|c|c|}
\hline
  $\!\!\! $  $_1$ & $\!\! $  $_3$  &
  $\!\! $  $_5$ \\[0.1cm]
\hline
  $\!\!\! $ $ _2$  &   $\!\!\!$ $_4$  & $\!\!\! $ $_6$  \\[0.1cm]
  \hline
\end{tabular} \; .
\ee
Using (\ref{hu3}) we find the action of the operators $s_m$ to the basis vectors
$\psi_\alpha$ (\ref{char6d})
$$
\begin{array}{c}
s_1 \psi_0 = \frac{1}{2} \overline{q} \psi_0 \; , \;\; s_1 \psi_1 = \frac{1}{2} \overline{q}
\psi_1 \; ,
\;\; s_1 \psi_2 = \frac{1}{2} \overline{q} \psi_2 \; ,
\;\; s_1 \psi_3 = - \frac{1}{2} \overline{q} \psi_3
\; , \;\; s_1 \psi_4 = - \frac{1}{2} \overline{q} \psi_4 \; ,  \\ [0.2cm]
s_2 \psi_0 = \frac{1}{2} \overline{q} \psi_0 \; , \;\;
 s_2 \psi_1 = \psi_3 - \frac{\lambda}{2} \frac{ (q^2+q^{-2}) }{(q^2-q^{-2})} \psi_1  \; ,  \;\;
    s_2 \psi_2 = \psi_4 - \frac{\lambda}{2} \frac{ (q^2+q^{-2}) }{(q^2-q^{-2})}  \psi_2 \; , \\ [0.2cm]
 s_2 \psi_3 = U_3^2 \psi_1 - \frac{\lambda}{2} \frac{(q^{-2}+q^{2}) }{(q^{-2}-q^{2})} \psi_3 \; ,  \;\;
     s_2 \psi_4 = U_3^2 \psi_2 - \frac{\lambda}{2} \frac{(q^{-2}+q^{2}) }{(q^{-2}-q^{2})} \psi_4 \; ,
\end{array}
$$
$$
\begin{array}{c}
   s_3 \psi_0 = \psi_1 - \frac{\lambda}{2} \frac{ (q^4+q^{-2}) }{(q^4-q^{-2})} \psi_0 \; , \;\;
 s_3 \psi_1 = U_4^2 \psi_0 - \frac{\lambda}{2} \frac{(q^{-2}+q^4) }{(q^{-2}-q^4)} \psi_1 \; ,  \;\;
    s_3 \psi_2 = \frac{1}{2} \overline{q} \psi_2   \; , \\ [0.2cm]
 s_3 \psi_3 =  \frac{1}{2} \overline{q}  \psi_3 \; ,  \;\;
     s_3 \psi_4 = -\frac{1}{2} \overline{q} \psi_4 \; ,\\ [0.2cm]
      s_4 \psi_0 = \frac{1}{2} \overline{q} \psi_0 \; , \;\;
 s_4 \psi_1 = \psi_2 - \frac{\lambda}{2} \frac{(q^4+1)}{(q^{4}-1)}  \psi_1\; ,  \;\;
    s_4 \psi_2 = U_5^2 \psi_1 - \frac{\lambda}{2} \frac{(1+q^4) }{(1-q^4)} \psi_2 \; , \\ [0.2cm]
 s_4 \psi_3 =  \psi_4 - \frac{\lambda}{2} \frac{(q^4+1)}{(q^{4}-1)}  \psi_3 \; ,  \;\;
     s_4 \psi_4 = U_5^2 \psi_3 - \frac{\lambda}{2} \frac{(1+q^4)}{(1-q^4)}  \psi_4 \; ,\\ [0.2cm]
s_5 \psi_0 = \frac{1}{2} \overline{q} \psi_0 \; , \;\; s_5 \psi_1 = \frac{1}{2} \overline{q}
\psi_1 \; ,
\;\; s_5 \psi_2 = - \frac{1}{2} \overline{q} \psi_2 \; ,
\;\; s_5 \psi_3 = \frac{1}{2} \overline{q} \psi_3
\; , \;\; s_5 \psi_4 = - \frac{1}{2} \overline{q} \psi_4 \; ,
\end{array}
$$
where
$$
\begin{array}{c}
U_3^2 \psi_\alpha = \frac{(q^3-q^{-3})(q-q^{-1})}{(q^2-q^{-2})^2} \psi_\alpha \;\;\; (\alpha=1,2) \; , \;\;
 U_4^2 \psi_0 = \frac{(q^5-q^{-3})(q^3-q^{-1})}{(q^4-q^{-2})^2} \psi_0 \; , \\ [0.2cm]
U_5^2 \psi_\alpha = \frac{(q^5-q^{-1})(q^3-q)}{(q^4-1)^2} \psi_\alpha \;\;\; (\alpha=1,3)  \; .
\end{array}
$$
Then the equation for
eigenvalues $\nu$ of ${\cal H}_{6}$ and eigenvectors in the space of the irreducible representation
(\ref{y33}) is given as follows:
$$
(\sum_{i=1}^5 s_i - \nu) (\psi_0 + a_1 \psi_1 + a_2 \psi_2 + a_3 \psi_3+ a_4 \psi_4) = 0 \, ,
$$
which leads to the characteristic identity $\nu = {\cal H}_{6}$:
\be
\lb{y32b}
(\nu-\frac{ \overline{q}}{2}) \left\{
3\overline{q}^4+16\overline{q}^2-64  - (8\overline{q}^3+160\overline{q}) \nu  +
(-8\overline{q}^2+128) \nu^2  + 32 \overline{q}  \nu^3  - 16 \nu^4 \right\} =
0 \; .
\ee
Note that eigenvalue $\nu=\frac{1}{2} \overline{q}$ for
the Hamiltonian ${\cal H}_6$ has multiplicity 2 
since it also appears in the representation
 $(4,1^2)$ (\ref{char6}). The second factor in the characteristic identity
(\ref{y32b}) is related to the Young diagram
$(3^2)$ and can be obtained by transformation ${\cal H}_6 \to -{\cal H}_6$
from the factor presented in (\ref{char6b})
for the dual diagram $(2^3)$.\\
\\
{\bf Acknowledgement.} We thank to N.~Slavnov for valuable discussions of the details of
the algebraic Bethe ansatz and especially for the explanations of the two-component model for the XXZ Heisenberg chain. The work of SOK was supported by RSCF grant 14-11-00598. The work of API was partially supported by the RFBR grant 14-01-00474.

\begin{appendices}
  \numberwithin{equation}{section}

 \section{Details in XXX} \lb{app:XXXmag}

Equations \eqref{Pferm} and \eqref{Rferm} result in the following useful identities
\begin{align}
    P_{k,k+1}\;\bpsi_k &= \bpsi_{k+1} - \bpsi_{k+1} N_k + \bpsi_k N_{k+1}, \lb{aux1}\\
    \hat{R}_{k,k+1}(\lambda)\bpsi_{k+1} &= \bpsi_{k+1} + \lambda \bpsi_k +\lambda \bpsi_{k+1}N_k - \bpsi_k N_{k+1}\lb{aux2}
\intertext{where $N_k = \bpsi_k \psi_k$, again. We can see that}
 P_{k,k+1}\bar{\psi}_k \ket{0} &= \bar{\psi}_{k+1} \ket{0}, \lb{aux3}\\
 \hat{R}_{k,k+1}(\mu) \bar{\psi}_{k+1} \ket{0}&= \bar{\psi}_{k+1}\ket{0} +\mu\bar{\psi}_k\ket{0}, \lb{aux4}\\
\intertext{and}
 \hat{R}_{k,k+1}(\mu) \ket{0} &=  (\mu +1)\ket{0},\lb{aux5}\\
 P_{k,k+1} \ket{0} &= \ket{0} .\lb{aux5.1}
\intertext{For higher magnons we will also need}
\hat{R}_{k,k+1}(\lambda)  \bar{\psi}_{k}\bar{\psi}_{k+1} & = (\lambda+1)\bar{\psi}_{k}\bar{\psi}_{k+1}\label{aux6}
\end{align}
and
\begin{align}
\hat{R}_{l,l+1}(\lambda)\dots \hat{R}_{k-1,k}(\lambda) \bar{\psi}_{k} \ket{0} &=
\lambda^{k-l} \bar{\psi}_l \ket{0} +  \frac{\lambda^{k-l}}{\lambda+1}  \sum_{j=1}^{k-l} \left( \frac{\lambda+1}{\lambda} \right)^{j}   \bar{\psi}_{l+j} \ket{0}, \lb{aux7} \\
\hat{R}_{l,l+1}(\lambda)\dots \hat{R}_{k,k+1}(\lambda) \bar{\psi}_{k} \ket{0} &= \hat{R}_{l,l+1}(\lambda)\dots \hat{R}_{k-1,k}(\lambda) \bar{\psi}_{k} \ket{0} + \lambda (\lambda +1)^{k-l} \bar{\psi}_{k+1}\ket{0}.\lb{aux8}
\end{align}


It is obvious that the second term in (\ref{Bferm}) annihilates vacuum state. Then, using \eqref{aux7}, we get for the 1-magnon state
\begin{align}
\ket{\mu}\equiv B(\mu)\ket{0} &= (\mu +1 -\mu N_1) \hat{R}_{12}(\mu)\dots\hat{R}_{L-1,L}(\mu) P_{L-1,L}\dots P_{12} \bar{\psi}_1 \ket{0} = \nonumber\\
&= (\mu +1 -\mu N_1) \hat{R}_{12}(\mu) \dots \hat{R}_{L-1,L}(\mu) \bar{\psi}_L\ket{0} = \nonumber\\
& = (\mu +1 -\mu N_1) \Bigl[ \mu^{L-1} \bpsi_1\ket{0} + \frac{\mu^{L-1}}{\mu+1} \sum_{j=1}^{L-1} \Bigl( \frac{\mu+1}{\mu} \Bigr)^j \bpsi_{j+1} \Bigr] \ket{0} =  \nonumber\\
& =  \frac{\mu^L}{\mu+1}
\sum_{k=1}^{L} \left(\frac{\mu+1}{\mu}\right)^{k} \bar{\psi}_k \ket{0} = n(\mu) \sum_{k=1}^{L} [\mu]^{k} \bar{\psi}_k \ket{0}
\end{align}
if we use the notation \eqref{koef01}.


Using (\ref{aux7}) and (\ref{aux8}) we obtain
\begin{gather}
B(\lambda) \bar{\psi}_k \ket{0}  = (\lambda+1) \lambda^{L-2} \sum_{m=0}^{k-2} [\lambda]^{m} \bar{\psi}_{m+1} \bar{\psi}_k \ket{0} + \displaybreak[0] \nonumber\\
 \quad +\frac{\lambda^{L-2}}{\lambda+1} \sum_{m=0}^{k-2}\sum_{j=1}^{L-k} [\lambda]^{j+m} \bar{\psi}_{m+1} \bar{\psi}_{k+j} \ket{0} 
 +  \frac{\lambda^{L}}{\lambda+1} [\lambda]^{k-1} \sum_{j=1}^{L-k} [\lambda]^{j} \bar{\psi}_k \bar{\psi}_{k+j} \ket{0}.  \lb{Bpsi}
\end{gather}
 We get for the 2-magnon state using \eqref{Bpsi}
 \begin{gather}
 \ket{\lambda,\mu}\equiv B(\lambda)B(\mu)\ket{0}  = n(\mu) \sum_{k=1}^L [\mu]^{k} B(\lambda)\bar{\psi}_k\ket{0} 
 = n(\mu) n(\lambda)  \Bigg[ \sum_{k=2}^L \sum_{m=0}^{k-2} [\mu]^k [\lambda]^{m+2} \bar{\psi}_{m+1}\bar{\psi}_k  + \nonumber\displaybreak[0]\\
   \quad+ \frac{1}{\lambda(\lambda+1)} \sum_{k=2}^{L-1} \sum_{m=0}^{k-2} \sum_{j=1}^{L-k}  [\mu]^k [\lambda]^{j+m+1}   \bar{\psi}_{m+1}\bar{\psi}_{k+j}  
 + \sum_{k=1}^{L-1} \sum_{j=1}^{L-k} [\mu]^k [\lambda]^{k+j-1} \bar{\psi}_{k}\bar{\psi}_{k+j}  \Bigg] \ket{0} = \nonumber\displaybreak[0]\\
 = n(\mu) n(\lambda)  \Bigg\{ \sum_{1\leq r<s\leq L}  \Bigg[ [\mu]^s [\lambda]^{r+1}
 + [\mu]^r [\lambda]^{s-1} \Bigg] \bar{\psi}_{r}\bar{\psi}_{s} \ket{0}  +\nonumber\displaybreak[0]\\
  +  \frac{1}{\lambda(\lambda+1)} \sum_{s=3}^{L} \sum_{r=1}^{s-2} \sum_{k=r+1}^{s-1} [\mu]^k [\lambda]^{s+r-k}  \bar{\psi}_{r}\bar{\psi}_{s} \ket{0}\Bigg\}= \nonumber \displaybreak[0]\\
  = n(\mu) n(\lambda)  \sum_{1\leq r<s\leq L} \Bigg[ [\mu]^s [\lambda]^{r+1} + [\mu]^r [\lambda]^{s-1} + 
 \frac{1 }{\lambda(\lambda+1)}  \sum_{k=r+1}^{s-1} [\mu]^k [\lambda]^{s+r-k} \Bigg] \bar{\psi}_{r}\bar{\psi}_{s} \ket{0}. \lb{2mag_ap}
 \end{gather}
The finite sum in (\ref{2mag_ap}) can be calculated explicitly
by means of geometric progression
 \be\begin{gathered}
 \lb{2mag1}
  \frac{1 }{\lambda(\lambda+1)}  \sum_{k=r+1}^{s-1}
 [\mu]^k [\lambda]^{s+r-k} = \\
 =  \frac{\mu }{\lambda (\lambda - \mu)} [\mu]^{r+1} [\lambda]^{s-1}
  \left( \left( \frac{[\mu]}{[\lambda]} \right)^{s-r-1} - 1 \right) =
 \frac{\mu }{\lambda (\lambda - \mu)}
 \left( [\mu]^{s} [\lambda]^{r} - [\mu]^{r+1} [\lambda]^{s-1} \right).
 \end{gathered}\ee
 Substitution of (\ref{2mag1}) into (\ref{2mag_ap}) gives
 \begin{align}
  B(\lambda)B(\mu)\ket{0} 
 & = n(\mu) n(\lambda)  \sum\limits_{1\leq r<s\leq L} \Bigg[ [\mu]^s
 [\lambda]^{r+1} +  [\mu]^r [\lambda]^{s-1}  + \nonumber\\
 & \quad  + \frac{\mu }{\lambda (\lambda - \mu)}
 \left( [\mu]^{s} [\lambda]^{r} - [\mu]^{r+1} [\lambda]^{s-1} \right) \Bigg] \bar{\psi}_{r}\bar{\psi}_{s} \ket{0} = \nonumber\\
&  = n(\mu) n(\lambda)  \sum\limits_{1\leq r<s\leq L} \Bigg[ [\lambda]^{r} [\mu]^{s}  \frac{\lambda - \mu +1}{\lambda-\mu} +
 [\mu]^{r} [\lambda]^{s} \frac{\mu - \lambda + 1}{\mu-\lambda}
 \Bigg] \bar{\psi}_{r}\bar{\psi}_{s} \ket{0} . \lb{2mag2}
\end{align}


For the 3-magnon we need at first
\begin{gather}
 B(\nu) \bpsi_r \bpsi_s \ket{0} = (\nu +1 -\nu N_1) X_{12\dots L}(\nu) \bar{\psi}_1 \bar{\psi}_r \bar{\psi}_s = \nonumber\\
 = \nu^{L-3}(\nu+1)^2 \sum_{m=0}^{r-2} [\nu]^m \bpsi_{m+1}\bpsi_r \bpsi_s\ket{0} + \nu^{L-3} \sum_{l=1}^{s-r-1} \sum_{m=0}^{r-2} [\nu]^{l+m} \bpsi_{m+1}\bpsi_{r+l} \bpsi_s\ket{0} + \displaybreak[0]\nonumber\\
  +\nu^{L-3} \sum_{j=1}^{L-s} \sum_{m=0}^{r-2} [\nu]^{j+m} \bpsi_{m+1} \bpsi_r \bpsi_{s+j} \ket{0} + \displaybreak[0]\nonumber\\
+ \nu^{L-3}(\nu+1)^{-2} \sum_{j=1}^{L-s} \sum_{l=1}^{s-r-1} \sum_{m=0}^{r-2} [\nu]^{j+l+m} \bpsi_{m+1} \bpsi_{r+l} \bpsi_{s+j} \ket{0} + \displaybreak[0]\nonumber\\
+ \nu^{L-s+r-1} (\nu+1)^{s-r-2} \sum_{j=1}^{L-s} \sum_{m=0}^{r-2} [\nu]^{j+m} \bpsi_{m+1} \bpsi_{s} \bpsi_{s+j} \ket{0} + \displaybreak[0]\nonumber\\
 +  \nu^{L-s+2} (\nu+1)^{s-3} \sum_{j=1}^{L-s} [\nu]^{j} \bpsi_{r} \bpsi_{s} \bpsi_{s+j} \ket{0} + \nu^{L-r} (\nu+1)^{r-1} \sum_{l=1}^{s-r-1} [\nu]^{l} \bpsi_{r} \bpsi_{r+l} \bpsi_{s} \ket{0} + \displaybreak[0]\nonumber\\
 + \nu^{L-r} (\nu+1)^{r-3} \sum_{j=1}^{L-s} \sum_{l=1}^{s-r-1} [\nu]^{j+l} \bpsi_{r} \bpsi_{r+l} \bpsi_{s+j} \ket{0}. \label{Bpsirs}
\end{gather}
The 3-magnon state is obtained from the 2-magnon state \eqref{2mag2}
\be \ket{\nu,\mu,\lambda} \equiv B(\nu)B(\mu)B(\lambda)\ket{0} = n(\mu) n(\lambda)
 \sum\limits_{1\leq r<s\leq L} K_2(s,r) B(\nu) \bar{\psi}_{r}\bar{\psi}_{s} \ket{0} \lb{3mag1}
\ee
where we denote for more comfort
\be K_2(s,r) =  [\mu]^{s} [\lambda]^{r} \frac{\lambda - \mu +1}{\lambda-\mu} +
 [\mu]^{r} [\lambda]^{s} \frac{\mu - \lambda +1}{\mu - \lambda}.
\ee
Using\eqref{Bpsirs} we get
\begin{gather}
\ket{\nu,\mu,\lambda} = n(\nu) n(\mu) n(\lambda) \sum_{1\leq q<r<s\leq L} \Bigg[ [\nu]^{q+2} K_{2}(s,r) + \frac{1}{\nu^2} \sum_{l=1}^{r-q-1} [\nu]^{l+q} K_2(s,r-l) + \nonumber\displaybreak[0]\\
 + \frac{1}{\nu^2} \sum_{j=1}^{s-r-1} [\nu]^{j+q} K_2(s-j,r) + \frac{1}{\nu^2(\nu+1)^2} \sum_{j=1}^{s-r-1} \sum_{l=1}^{r-q-1} [\nu]^{j+l+q} K_2(s-j,r-l) + \nonumber\displaybreak[0]\\
 + \frac{1}{(\nu+1)^2} \sum_{l=q+1}^{r-1} [\nu]^{s-l+q} K_2(r,l) +
 [\nu]^{s-2} K_2(r,q) + [\nu]^{r} K_2(s,q) + \nonumber\displaybreak[0]\\
 +\frac{1}{(\nu+1)^2} \sum_{l=r+1}^{s-1} [\nu]^{s+r-l} K_2(l,q) \Bigg] \bpsi_q \bpsi_r \bpsi_s \ket{0}= n(\nu) n(\mu) n(\lambda) \times \displaybreak[0]\nonumber\\
   \sum_{1\leq q<r<s\leq L}  \sum_{T\in S_3} T \Bigl( [\nu]^q [\mu]^r[\lambda]^s
 \frac{\nu-\mu+1}{\nu-\mu}\cdot \frac{\nu-\lambda+1}{\nu-\lambda}\cdot \frac{\mu-\lambda+1}{\mu-\lambda}\Bigr) \bpsi_q \bpsi_r \bpsi_s \ket{0}. \lb{3mag2}
\end{gather}

\section{Details in XXZ}\lb{app:XXZmag}

 It is convenient to introduce the following notation:
 \begin{align}
    d(\lambda) &= 1-\lambda,& b(\lambda) & = q-q^{-1}, & a(\lambda) &= q-\lambda q^{-1}.
 \end{align}
 We see that coefficient \eqref{koef03} resp. normalization \eqref{koef04} can be written as
 \be
 [\lambda]_q = \frac{a(\lambda)}{d(\lambda)},\qquad n_q(\lambda) = \frac{d(\lambda)^L b(\lambda)}{a(\lambda)}.
 \ee
  The R-matrix $\hat{R}_{k,k+1}(\lambda)$ is of the form \eqref{Rbaxtferm} and the operator $B(\lambda)$ is of the form \eqref{BfermXXZ}. For computing of Bethe vectors, the following set of identities seems to be very useful:
 \begin{align}
&\hat{R}_{k,k+1}(\lambda) \bpsi_{k+1} \ket{0} = d(\lambda)\psi_k \ket{0} + b(\lambda) \bpsi_{k+1} \ket{0}, \lb{aux1-q}\\
&\hat{R}_{k,k+1}(\lambda) \bpsi_{k} \ket{0} = \lambda b(\lambda)\psi_k \ket{0} + d(\lambda) \bpsi_{k+1} \ket{0}, \lb{aux2-q}\\
&\hat{R}_{k,k+1}(\lambda) \bpsi_{j}  = a(\lambda)\psi_j & \text{for } j\not\in\{ k,k+1\},\lb{aux3-q}
    \end{align}
    and
\begin{align}
    & \hat{R}_{l,l+1}(\lambda)\dots  \hat{R}_{k-1,k}(\lambda) \bpsi_k\ket{0} = d(\lambda)^{k-l} \bpsi_l \ket{0} + b(\lambda) \sum_{j=1}^{k-l} d(\lambda)^{k-l-j} a(\lambda)^{j-1} \bpsi_{l+j} \ket{0}, \lb{aux4-q} \\
    & \hat{R}_{l,l+1}(\lambda)\dots  \hat{R}_{k,k+1}(\lambda) \bpsi_k\ket{0} = \lambda b(\lambda) \hat{R}_{l,l+1}(\lambda)\dots  \hat{R}_{k-1,k}(\lambda) \bpsi_k\ket{0} + d(\lambda)a(\lambda)^{k-l}\bpsi_{k+1}\ket{0}.\lb{aux5-q}
 \end{align}


Using \eqref{aux4-q}, we can straightforwardly calculate the $q$-deformed 1-magnon state for $B(\mu)$ defined in (\ref{BfermXXZ})
\be \lb{1mag-q}
    \ket{\mu}\equiv B(\mu)\ket{0} = \frac{d(\mu)^L b(\mu)}{a(\mu)} \sum_{k=1}^L \Bigl( \frac{a(\mu)}{d(\mu)} \Bigr)^k \bpsi_k \ket{0} = n_q(\mu) \sum_{k=1}^L [\mu]_q^k \bpsi_k \ket{0}
\ee
recalling \eqref{koef03} and \eqref{koef04}.


Using the formulas mentioned \eqref{aux1-q}-\eqref{aux5-q}, we can calculate the $q$-deformed 2-magnon state. First of all we need $B(\lambda)\bpsi_k\ket{0}$
    \begin{align}
B(\lambda)\bpsi_k\ket{0} & = b(\lambda) a(\lambda) d(\lambda)^{L-2} \sum_{i=0}^{k-2} [\lambda]_q^i \bpsi_{1+i}\bpsi_{k}\ket{0} + \nonumber \\
    &\quad + b(\lambda) a(\lambda)^{k-2} d(\lambda)^{L-k+1} \sum_{j=1}^{L-k} [\lambda]_q^j \bpsi_{k}\bpsi_{k+j}\ket{0} + \displaybreak[0]\notag\\
    &\quad + \lambda b(\lambda)^3 a(\lambda)^{-1} d(\lambda)^{L-2} \sum_{i=0}^{k-2} \sum_{j=1}^{L-k} [\lambda]_q^{i+j} \bpsi_{1+i}\bpsi_{k+j}\ket{0}.\lb{Bpsi-q}
    \end{align}
Hence, we get the $q$-deformed 2-magnon state
$$
\ket{\lambda,\mu} \equiv B(\lambda)B(\mu)\ket{0} = n_q(\mu) \sum_{k=1}^{L} [\mu]_q^k B(\lambda) \bpsi_k \ket{0} = $$
$$  = n_q(\mu)n_q(\lambda) \sum_{1\leq r<s\leq L} \Bigg\{ \frac{a(\lambda)}{d(\lambda)}\left( 1  + \frac{\lambda b(\lambda)^2 a(\lambda)^{-1} d(\mu)}{a(\mu)d(\lambda) - a(\lambda)d(\mu)} \right) [\lambda]_q^r[\mu]_q^s + $$
$$  + \left( \frac{d(\lambda)}{a(\lambda)} - \frac{a(\mu)}{d(\mu)} \frac{\lambda b(\lambda)^2 a(\lambda)^{-1} d(\mu)}{a(\mu)d(\lambda) - a(\lambda)d(\mu)} \right)[\lambda]_q^s[\mu]_q^r\Bigg\}\bpsi_r \bpsi_s\ket{0}= $$
\be\lb{2mag-q}
    = n_q(\mu) n_q(\lambda) \sum_{1\leq r<s\leq L} \Bigg\{  \frac{\lambda q^{-1} -\mu q}{\lambda - \mu}  [\lambda]_q^r[\mu]_q^s +  \frac{\mu q^{-1} -\lambda q}{\mu - \lambda}  [\mu]_q^r [\lambda]_q^s  \Bigg\} \bpsi_r \bpsi_s\ket{0}.
\ee

\end{appendices}

\end{document}